\documentclass[11pt,a4paper,final]{article}
\usepackage{amsmath}
\usepackage{amsfonts}
\usepackage{amssymb}
\usepackage{amsthm}
\usepackage{mathrsfs}
\usepackage{mathtools}
\usepackage{dsfont}
\usepackage{graphicx}
\usepackage{color}
\usepackage[round]{natbib}
\usepackage[margin=1in]{geometry}
\usepackage[titletoc]{appendix}
\usepackage{bm}
\usepackage{algpseudocode}
\usepackage{algpascal}
\usepackage{afterpage}
\usepackage{enumitem}
\usepackage{multirow}
\usepackage{booktabs}

\usepackage{graphicx}
\usepackage{graphics}
\usepackage{pict2e}
\usepackage{epic}

\usepackage{booktabs}       
\usepackage[table]{xcolor} 
\usepackage{graphicx}      
\usepackage{bm}            

\usepackage{epstopdf} 
\usepackage[ruled,vlined]{algorithm2e}
\usepackage{float}
\usepackage{caption}
\usepackage{subcaption}
\usepackage{latexsym}

\usepackage[backref=page]{hyperref}
\hypersetup{
    colorlinks=true,
    citecolor=blue,
    linkcolor=red,
    filecolor=blue,      
    urlcolor=blue,
    pdfpagemode=FullScreen,
    }
\usepackage{setspace}

\newlength{\vslength}
\newtheorem{theorem}{Theorem}

\newtheorem{lemma}{Lemma}[section]

\newtheorem{proposition}{Proposition}[section]

\newtheorem{corollary}{Corollary}[section]

\newtheorem{condition}{Condition}

\newcommand{\beginsurvey}{
        \setcounter{condition}{0}
        \renewcommand{\thecondition}{S\arabic{condition}}%
     }
\newcommand{\beginmodel}{
        \setcounter{condition}{0}
        \renewcommand{\thecondition}{M\arabic{condition}}%
     }
\newcommand{\beginweights}{
        \setcounter{condition}{0}
        \renewcommand{\thecondition}{W\arabic{condition}}%
     }
\newcommand{\beginprior}{
        \setcounter{condition}{0}
        \renewcommand{\thecondition}{P\arabic{condition}}%
     }

\theoremstyle{remark}
\newtheorem{rem}{\bf Remark}[section]
\newenvironment{Rem}{\begin{rem}}{\end{rem}}
\newtheorem{example}{\bf Example}[section]

\renewenvironment{proof}{\noindent{\it{Proof.} }}{\qed}
\def\cA{\mathcal A}

\def\cN{\mathcal N}
\def\cP{\mathcal P}

\newcommand{\bbE}{{\mathbb E}}
\newcommand{\bbR}{{\mathbb R}}

\newcommand{\bbN}{{\mathbb N}}

\newcommand{\bbP}{{\mathbb P}}

\newcommand{\Var}{\mbox{{\rm Var}}}
\newcommand{\Cov}{\mbox{{\rm Cov}}}

\newcommand{\ie}{{\it i.e. }}

\newcommand{\iid}{{i.i.d.}}

\newcommand{\bc}{\begin{center}}
\newcommand{\ec}{\end{center}}
\newcommand{\be}{\begin{equation}}
\newcommand{\ee}{\end{equation}}
\newcommand{\ba}{\begin{array}}
\newcommand{\ea}{\end{array}}
\newcommand{\bean}{\setlength\arraycolsep{1pt}\begin{eqnarray*}}
\newcommand{\eean}{\end{eqnarray*}}
\newcommand{\bea}{\setlength\arraycolsep{1pt}\begin{eqnarray}}
\newcommand{\eea}{\end{eqnarray}}
\newcommand{\ben}{\begin{enumerate}}
\newcommand{\een}{\end{enumerate}}
\newcommand{\bed}{\begin{itemize}}
\newcommand{\eed}{\end{itemize}}

\def\log{\text{log }}

\def\boxit#1{\vbox{\hrule\hbox{\vrule\kern6pt
 \vbox{\kern6pt#1\kern6pt}\kern6pt\vrule}\hrule}}

\def\bse{\begin{eqnarray*}}
\def\ese{\end{eqnarray*}}
\def\be{\begin{eqnarray}}
\def\ee{\end{eqnarray}}
\def\bq{\begin{equation}}
\def\eq{\end{equation}}
\def\bse{\begin{eqnarray*}}
\def\ese{\end{eqnarray*}}

\def\Cov{\hbox{Cov}}

\def\boxit#1{\vbox{\hrule\hbox{\vrule\kern6pt
          \vbox{\kern6pt#1\kern6pt}\kern6pt\vrule}\hrule}}

\def\log{\text{log }}

\def\bse{\begin{eqnarray*}}
\def\ese{\end{eqnarray*}}
\def\be{\begin{eqnarray}}
\def\ee{\end{eqnarray}}
\def\bq{\begin{equation}}
\def\eq{\end{equation}}
\def\bse{\begin{eqnarray*}}
\def\ese{\end{eqnarray*}}

\newcommand{\ua}       {\boldsymbol{a}}

\newcommand{\ub}       {\boldsymbol{b}}

\newcommand{\ug}       {\boldsymbol{g}}

\newcommand{\uv}       {\mbox{\boldmath$v$}}

\newcommand{\uw}       {\boldsymbol{w}}
\newcommand{\uX}       {\boldsymbol{X}}

\newcommand{\uY}       {\boldsymbol{Y}}

\newcommand{\udelta}            {\mbox{\boldmath$\delta$}}

\newcommand{\uiota}             {\mbox{\boldmath$\uiota$}}

\newcommand{\tL}       {\Tilde{L}}

\newcommand{\ttheta}       {\Tilde{\theta}}
\newcommand{\cp}       {\xrightarrow{c.p.}}
\newcommand{\cd}       {\xrightarrow{c.d.}}
\newcommand{\as}       {\xrightarrow{a.s.}}

\newcommand{\dist}       {\xrightarrow{d}}
\newcommand{\asP}       {{a.s.[P_*]}}
\newcommand{\htheta}       {{\hat{\theta}_n}}
\newcommand{\ctheta}       {{\check{\theta}_n}}

\newcommand{\nablatheta}       {\nabla_\theta \,}
\newcommand{\nablathetasq}       {\nabla_\theta^2\,}
\newcommand{\nablathetac}       {\nabla_\theta^3\,}

\numberwithin{equation}{section}
\numberwithin{lemma}{section}
\usepackage{comment}

\graphicspath{{Figures/}}
\newenvironment{theorem-non}[1][Theorem]{{#1.} \it}{\par}

\begin{document}
\thispagestyle{empty}
\title{
    \vspace*{-12mm}
    Model-based inference in complex surveys through \\[-3mm]
    targeted resampling of weights
    
}
\author{Snigdha Das$^{1}$, Dipankar Bandyopadhyay$^{2}$, Debdeep Pati$^{3}$ \\[3mm]
    {\small\it {}$^{1}$Department of Biostatistics, Yale University} \\
    {\small\it {}$^{2}$Department of Biostatistics, Virginia Commonwealth University} \\
    {\small\it {}$^{3}$Department of Statistics, University of Wisconsin-Madison}
}

\date{}
\maketitle

\vspace*{-6mm}
\begin{abstract}
Survey data often arises from complex sampling designs, such as stratified or multistage sampling, with unequal inclusion probabilities. When sampling is informative, traditional inference methods yield biased estimators and poor coverage. Classical pseudo-likelihood-based methods provide accurate asymptotic inference but lack finite-sample uncertainty quantification and the ability to integrate prior information. Existing Bayesian approaches, like the Bayesian pseudo-posterior estimator and weighted finite population Bayesian bootstrap, have limitations; the former struggles with uncertainty quantification, while the latter is computationally intensive and sensitive to bootstrap replicates. We propose the {\em Survey-adjusted Weighted Likelihood Bootstrap} (S-WLB), a computationally efficient framework that accounts for informative sampling by resampling random weights from a carefully chosen distribution centered around the underlying sampling weights. We establish conditional consistency and asymptotic normality of the S-WLB, showing that its uncertainty intervals are asymptotically valid. We further develop the {\em Survey-adjusted Posterior Bootstrap} (S-PB), which extends S-WLB to incorporate prior information. Using an Edgeworth expansion, we characterize the second-order influence of the prior on S-PB while retaining the first-order asymptotic properties inherited from S-WLB. We demonstrate the performance of the proposed methods through simulations and applications to nationally representative survey datasets like NHANES and NSDUH.

\bigskip
    
\end{abstract}

\section{Introduction}
\label{s:intro}
Survey data is ubiquitous in public health, economics and social sciences, with prominent examples of nationally-representative large survey databases {being} the National Health and Nutrition Examination Survey \citep[NHANES;][]{naavaal2022association}, the National Survey on Drug Use and Health \citep[NSDUH;][]{back2010gender}, {among} others.  The advent of large-scale complex surveys emphasizes the need for efficient inference, both statistically and computationally. Unlike simple random sampling which gives every unit in the population an equal chance of being selected, survey designs employ complex schemes such as stratification, clustering and multistage sampling that result in unequal sample selection probabilities. 
When the inclusion probabilities are correlated with the variable of interest, the sampling mechanism becomes {\em informative}. Ignoring this non-representativeness and applying classical inference, which assumes the data are independent and identically distributed (\iid), can lead to biased estimation and invalid inference.

All methods of survey analyses involve the use of sampling weights, which reflect unequal inclusion probabilities and can be thought of as the number of units in the population that each sampled unit represents.
Inference for survey data is commonly conducted under either a design-based or a model-based paradigm. Design-based approaches treat the finite population values as fixed and attribute uncertainty solely to the sampling design. 
Classical examples include the weighted estimators of \citet{horvitz1952generalization} and \citet{hajek1971comment}, with subsequent developments based on penalized splines \citep{zheng2003penalized, zheng2005inference, chen2010bayesian}, empirical likelihood procedures \citep{RaoWu2010}, nonparametric Dirichlet process mixture models \citep{KUNIHAMA201641}, multilevel regression with post-stratification \citep{Si2015BA, si2020bayesian},  and pseudo-population resampling schemes \citep{dong2014nonparametric, zhou2016twostep, zangeneh2015bayesian, wang2022bootstrap}. These methods are primarily designed for inference on finite-population quantities such as totals, means and proportions,\,and\,often rely on reconstructing or imputing the non-sampled population units.

In this work, we focus on the model-based paradigm for survey inference. Under this framework, survey data are viewed as arising from two random processes: {\em the superpopulation model}, which generates the finite population values, and {\em the sample selection mechanism}, which determines which units are included in the sample. We focus on inference for the parameters governing the superpopulation model rather than finite-population quantities. Each sampled unit is assumed to be accompanied by a survey weight, interpreted as the inverse of its inclusion probability, while the underlying sampling design and the procedure used to construct these weights remain unknown. Our goal is to perform valid likelihood-based inference for the superpopulation parameters while accounting for informative sampling through the observed survey weights. Model-based approaches commonly incorporate survey weights through weighted likelihood formulations, where each observation’s likelihood is raised to the power of its sampling weight. These include the frequentist pseudo maximum likelihood approach \citep{pfeffermann1998parametric, chambers2003analysis} and the Bayesian pseudo-posterior methods \citep{savitsky2016bayesian, savitsky2018scalable, leon2019fully, williams2018bayesian,  williams2020bayesian,williams2021uncertainty} for parametric inference. A detailed review of these approaches is provided in Section~\ref{s:existing_methods}.

While existing model-based procedures provide important tools for inference under informative sampling, several challenges remain. The frequentist pseudo maximum likelihood estimator (PMLE) achieves valid asymptotic coverage but typically relies on plug-in sandwich variance estimators that may be difficult to derive and compute in complex models. 
The pseudo-posterior framework provides a Bayesian alternative, but it does not yield correct uncertainty quantification. 
An intuitive explanation for this issue is that the sampling weight of a particular unit $i$ is meant to represent other population units that are {\em similar} to unit $i$. 
However, when the likelihood of a unit is raised to the power of its weight, the interpretation becomes problematic as it suggests the presence of exact copies of that unit in the population, with the number of copies equal to its weight, thereby distorting the underlying probabilistic model. 
Consequently, although the sandwich covariance of the PMLE yields valid asymptotic coverage, the pseudo-posterior distribution does not accurately reflect the uncertainty induced by informative sampling. Several modifications have been proposed to address this limitation. However, existing approaches either require computationally intensive Markov chain Monte Carlo (MCMC) algorithms \citep{leon2019fully}, rely on post-processing adjustments \citep{williams2021uncertainty}, or depend on reconstructing a potentially large pseudo finite population \citep{gunawan2020bayesian}, which can be challenging to scale in large-scale survey applications.

In this paper, we first propose a {\em Survey-adjusted Weighted Likelihood Bootstrap} (S-WLB) algorithm, inspired by the {\em weighted likelihood bootstrap} \citep[WLB;][]{newton1994approximate}. The WLB repeatedly maximizes a randomly weighted likelihood function, using \iid\ weights to reproduce the asymptotic variability of the maximum likelihood estimator. Motivated by its random weighting principle, we develop a novel survey-adjusted weighting scheme designed particularly for informative sampling, in which the distribution of the random weights incorporate the observed survey weights. This carefully chosen distribution of random weights accounts for sampling bias while ensuring that the asymptotic covariance of the S-WLB estimator matches the sandwich covariance of the pseudo maximum likelihood estimator (PMLE). Consequently, S-WLB provides asymptotically valid uncertainty intervals without requiring explicit estimation of the sandwich covariance matrix.
We then introduce the {\em Survey-adjusted Posterior Bootstrap} (S-PB), which combines the survey-adjusted random weighting mechanism of S-WLB with the prior penalization framework of \citet{pompe2021}, originally developed for \iid\ data. The resulting procedure incorporates prior information under informative sampling while retaining the sandwich covariance of the PMLE in its first-order asymptotic distribution. We further establish a second-order Edgeworth expansion for S-PB, showing that the influence of the prior enters through a correction involving the gradient of the log-prior, while the leading normal approximation and the remaining higher-order correction are independent of the prior.
Thus, S-PB incorporates prior information as in conventional Bayesian posteriors without compromising the first-order frequentist validity inherited from S-WLB. Our simulation studies further demonstrate that informative priors can improve finite-sample estimation accuracy and yield shorter uncertainty intervals.
Both S-WLB and S-PB inherit the key computational advantages of the WLB. In many complex models, obtaining a point estimator is often straightforward even when full probabilistic inference is computationally infeasible. Both methods require only that the stationary point of the corresponding randomized objective function forms a consistent sequence of estimators. Consequently, they are simple to implement, relying solely on repeated optimization rather than posterior sampling, while their inherent parallelizability provides a significant computational advantage over MCMC methods, particularly for large datasets.

\section{Background and existing methods}
\label{s:background}

\subsection{Set-up and notations}
\label{s:framework}

In the following, we introduce a formal superpopulation framework.
Let $P_{\theta_0}$ denote the true probability distribution for our random variable of interest $X$.
$P_{\theta_0}$ is one element of a finite-dimensional model
$\cP_\Theta = \{P_\theta : \theta \in \Theta \subset \bbR^K\}$, $K \geq 1$, also known as the {\em superpopulation model}. $P_\theta$ admits a density function $f_\theta$.
We consider an underlying finite population of random variables $\uX_N^* = \{X_l^*\}_{l=1}^N$, where \(X_1^*, \ldots, X_N^*\) are independent draws from \(P_{\theta_0}\). A {\em sampling design} is specified by a probability distribution \(P_{\pi,N}\) on the vector of inclusion indicators $\udelta_N = \{\delta_l\}_{l=1}^N \in \{0,1\}^N,$
where \(\delta_l = 1\) if unit \(X_l^*\) is included in the sample and \(\delta_l = 0\) otherwise. Under {\em informative sampling}, the sampling design may depend on information contained in the finite population \(\uX_N^*\). Consequently, \(P_{\pi,N}\) is defined conditionally on the realized finite population. The marginal inclusion probabilities are defined as $\pi_l = \bbP(\delta_l = 1 \mid \uX_N^*),$
and the second-order inclusion probabilities as $\pi_{lk} = \bbP(\delta_l \delta_k = 1 \mid \uX_N^*)$, $1 \le l \neq k \le N$.

Our {\em observed sample} is denoted by $\uX_n = \{X_i\}_{i=1}^n = \{X^*_l : \delta_l = 1\}_{l=1}^N$, drawn from the underlying population according to the sampling design, where $n = \sum_{l=1}^N \delta_l$ denotes the realized sample size. Here, $n$ depends on $N$, and may be random unless fixed by the sampling design; we suppress this dependence in our notation. Associated with the sampled units are the corresponding {\em sampling weights} $\uw_n = \{w_i\}_{i=1}^n$. The sampling weights are assumed to be constructed such that each weight $w_i$ is inversely proportional to the probability $\pi_i$ of selecting the observation $X_i$ under the survey design. The underlying sampling design and the procedure used to construct the weights are not assumed to be known.  Observations with lower inclusion probabilities receive larger weights to account for their under-representation in the sample, whereas observations with higher inclusion probabilities receive smaller weights. Each weight  $w_i$  can be viewed as the number of units in the population represented by the sampled unit $i$. This re-weighting is intended to adjust for the selection bias and enable consistent estimation of the population characteristics.
Let $\Tilde{w}_{n,i} = n w_i \big/ \sum_{j=1}^n w_j $ denote the scaled weights, such that $\sum_{i=1}^n \Tilde{w}_{n,i} = n$. Our objective is to conduct inference for the superpopulation parameter $\theta_0$ using the observed sample $\uX_n$ and its associated survey weights $\uw_n$.

\subsection{Existing methods in model-based inference}
\label{s:existing_methods}

In the following, we review existing methods for model-based inference on superpopulation parameters from survey data. We begin with the pseudo maximum likelihood estimator (PMLE), the standard frequentist approach for parametric inference under informative sampling \citep{pfeffermann1998parametric, chambers2003analysis}.
Under the framework in Section \ref{s:framework}, the usual likelihood is replaced by a powered likelihood, with each observation's likelihood contribution $f_\theta(X_i)$ raised to its scaled sampling weight $\tilde{w}_{n,i}$, 
\vspace{-1ex}
\begin{equation}
\label{eq:powered-lik}
    L_{w,n} (\theta) ~=~ \prod_{i=1}^n \,[f_\theta(X_i)]^{\,\Tilde{w}_{n,i}} \,.
\end{equation}
The exponent $\tilde{w}_{n,i}$ adjusts for sampling informativeness by assigning relative importance to the likelihood contribution of unit $i$ to approximate the population likelihood. Weights are scaled by the sample size $n$, which asymptotically represents the information contained in the observed sample. Here and henceforth, let $\nabla_\theta$ and $\nabla_\theta^2$ denote the gradient and Hessian with respect to $\theta$, respectively. The PMLE $\htheta$ is obtained by solving the powered likelihood equation, $\nabla_\theta\, \log L_{w,n} (\theta) = 0.$ 
Under certain regularity conditions \citep{white1982maximum},
\begin{equation}
\label{eq:dist_PMLE}
    \sqrt{n}\big(\htheta - \theta_0 \big)  ~\dist~ \cN_K \big(0, J(\theta_0)^{-1} I_w(\theta_0) J(\theta_0)^{-1} \big)\,,
\end{equation}
where $\theta_0$ is the true parameter, and $J(\theta_0)$ and $I_w(\theta_0)$ are consistently estimated by
\begin{equation}
\begin{split}
\label{eq:I_and_J}
    J_{w,n}(\htheta) & ~=~ - \frac{1}{n} \sum_{i=1}^n \tilde{w}_{n,i} \ \nablathetasq \log f_\theta (X_i)  \, \big|_{\theta = \htheta},\\
    I_{w,n} (\htheta) & ~=~ \, \frac{1}{n} \sum_{i=1}^n \tilde{w}_{n,i}^2  \, \nabla_\theta \log f_\theta (X_i) \,\big(\nabla_\theta \log f_\theta (X_i) \big)^{\!\!\top} \Big|_{\theta = \htheta}.
\end{split}
\end{equation}
For inference on $\theta$, the standard errors are obtained from the observed sandwich covariance estimator $n^{-1}J_{w,n}(\htheta)^{-1}I_{w,n}(\htheta)J_{w,n}(\htheta)^{-1}$. Thus, if the true population density is $f_{\theta_0}$ for some true value $\theta_0$ of $\theta$, the sandwich covariance estimator of the PMLE enables optimal coverage of its confidence intervals, even when based on a pseudo powered likelihood.

The pseudo likelihood has also been used as the foundation for Bayesian inference under informative sampling. Specifically,
\citet{savitsky2016bayesian} proposed a Bayesian pseudo posterior, defined as the product of the pseudo likelihood $L_{w,n}(\theta)$ in \eqref{eq:powered-lik} and a prior distribution $p(\theta)$ on $\theta$, $p(\theta |\uX_n, \uw_n) \propto p(\theta)\, \prod_{i=1}^n \,[f_\theta(X_i)]^{\,\Tilde{w}_{n,i}} \,$. 
The proposed method provides a flexible framework for a broad class of population-generating models and enables consistent estimation under informative sampling. \citet{savitsky2018scalable} extended this framework with a divide-and-conquer technique for scalability, while \citet{williams2018bayesian,williams2020bayesian} incorporated higher-order dependencies among population units. Although this approach allows for consistent estimation, the pseudo posterior converges to a normal distribution with a covariance matrix differing from that of the PMLE. Specifically, $\sqrt{n} (\theta - \hat{\theta}_n)$ converges to $\cN_K(0, J(\theta_0)^{-1})$ as $n \to \infty$, along almost every sequence of observed samples, which leads to credible intervals that lack optimal frequentist coverage. 

Several adjustments to the pseudo posterior framework have been proposed in the literature to improve coverage accuracy. \citet{leon2019fully} introduced a conditional population model $p(\pi_i \mid X_i)$ for the inclusion probabilities $\pi_i$, adjusting the likelihood multiplicatively to obtain asymptotically unbiased estimates. Although this approach is shown to yield credible intervals with correct coverage under simple, single-stage, proportional-to-size sampling designs, it is computationally intensive, requiring separate MCMC sampling and numerical integration at each iteration. \citet{williams2021uncertainty} proposed an ad hoc post-processing correction to MCMC samples from the pseudo posterior distribution to achieve approximate nominal coverage.
The method relies on re-using the data for plug-in estimates post MCMC sampling, and does not conform to a coherent empirical Bayes framework. \citet{gunawan2020bayesian} tackled this problem using a Bayesian weighted estimator (BWE) by generating pseudo-representative samples through the weighted finite population Bayesian bootstrap \citep{cohen1997bayesian, dong2014nonparametric}, which uses survey weights to correct for sampling bias. Inference on $\theta$ is then conducted via MCMC by treating each regenerated sample as the observed data. While this method attains good empirical coverage, it lacks theoretical guarantees. Additionally, the required number of bootstrap samples for valid inference is unclear. The method is computationally demanding and practically infeasible for large survey datasets, as it requires reconstructing a potentially large finite population for each bootstrap replicate, and may lead to poor mixing of the Markov chain.

\section{Proposed methods}

\subsection{Survey-adjusted weighted likelihood bootstrap}
\label{s:S-WLB}

Building on the gaps identified in Section \ref{s:existing_methods}, we aim to develop a simple, computationally efficient algorithm for model-based inference that accounts for unequal selection probabilities under informative sampling. Our construction builds on the {\em weighted likelihood bootstrap} \citep[WLB;][]{newton1994approximate}, which provides a computationally efficient approximation to Bayesian posterior sampling for well-specified parametric models under \iid\ sampling. Suppose $X_1, \ldots,X_n$ denotes \iid\ observations from a density $f_\theta$, where $\theta\in\Theta\subset\mathbb{R}^K$.
A WLB sample \(\tilde\theta_n\) is obtained by repeatedly maximizing the randomly weighted likelihood,
\begin{equation}
\label{eq:wlb}
\tL_n (\theta)
=
\prod_{i=1}^n
 [f_\theta(X_i)]^{\,g_i}, \qquad g_{1},\ldots, g_{n}
~\stackrel{\mathrm{\iid}}{\sim}~
\mathrm{Exp}(1).
\end{equation}
\citet{newton1991thesis} showed that under a well specified parametric model $P_{\theta_0}$, the WLB is asymptotically first-order equivalent to a Bayesian posterior. Specifically, the probability distributions of $\sqrt{n} (\ttheta_n - \hat{\theta}_n)$ and $\sqrt{n} (\theta - \hat{\theta}_n)$ converge to the same limit, $\cN_K(0, I(\theta_0)^{-1})$ as $n \to \infty$, along almost every sample path $\{X_n\}_{n \geq 1}$. Here, $\hat{\theta}_n$ and $\theta$ denote the MLE and the parameter under a Bayesian posterior, respectively, and $I(\theta_0)$ denotes the Fisher information matrix at the true parameter $\theta_0$. A key observation is that random weights $\ug_n = \{g_i\}_{i=1}^n$ drawn independently from an Exp$(1)$ distribution allocate weights among the \iid \ data such that no single observation is overly emphasized.
The randomness introduced through these weights effectively reproduces the sampling variability of the MLE.


Under complex survey designs, however, the observed data \(\uX_n = \{X_i\}_{i=1}^n\) are not i.i.d.\ from $P_{\theta_0}$, but drawn from an informative sampling design with unequal inclusion probabilities and associated survey weights \(\uw_n = \{w_i\}_{i=1}^n\). To adjust for the sampling bias,  we replace the \iid \ exponential weights of the WLB with a vector of independent Gamma random variables $\ug_n = \{g_{n,i}\}_{i=1}^n$, whose distribution incorporates the scaled survey weights \(\tilde{\uw}_n=\{\tilde w_{n,i}\}_{i=1}^n\). Since $\tilde w_{n,i}$ is a scaled version of the survey weight that preserves the relative representation of unit $i$ under the sampling design while summing to the sample size, each random weight $g_{n,i}$ is generated to have mean $\tilde w_{n,i}$. Specifically, our sample $\tilde \theta_n$ is obtained by maximizing
\begin{equation}
\label{eq:swlb}
\tL_{w,n} (\theta)
=
\prod_{i=1}^n
 [f_\theta(X_i)]^{\,g_{n,i}}, \qquad g_{n,i}
\,\stackrel{\text{ind.}}{\sim}
\mathrm{Gamma}(1, \tilde w_{n,i}), \ i = 1, 2,\ldots,n,
\end{equation}
where the Gamma distribution is parameterized by shape and scale. We use these random weights within a bootstrap framework to construct our final algorithm, which we refer to as the {\em Survey-adjusted Weighted Likelihood Bootstrap} (S-WLB) and outline it in Algorithm \ref{alg:S-WLB}. The distribution of the random weights $\ug_n$ is chosen to simultaneously correct for the sampling bias and reproduce the asymptotic sandwich covariance of the PMLE.
The resulting procedure is straightforward to implement whenever weighted optimization is feasible and does not require estimation of the Fisher information matrix or plug-in variance estimators. Moreover, S-WLB produces independent draws, is naturally parallelizable, and avoids the tuning and mixing issues associated with MCMC.

\RestyleAlgo{ruled}
\begin{algorithm}[ht]
  \caption{The Survey-adjusted Weighted Likelihood Bootstrap
  \label{alg:S-WLB}} 
  \vspace{0.5ex}
  \textbf{Input:} Observed survey sample $\uX_n = (X_1, X_2, \ldots, X_n)$.\\[1ex]
  ~~~~~~~~~~ Scaled survey weights  $\tilde{\uw}_n = (\tilde w_{n,1},\tilde w_{n,2},\ldots,\tilde w_{n,n})$ with $\sum_{i=1}^n \tilde w_{n,i} = n$.\\[1ex]
  For $j = 1, 2, \ldots, B$: \hfill $\triangleright$ (in parallel)\\[1ex]
    ~~~~ 1. Draw random weights, $g^{(j)}_{n,i} \,\stackrel{\text{ind.}}{\sim} \text{Gamma}(1, \tilde w_{n,i})$, $i=1,2,\ldots, n$.\\[1ex]
    ~~~~ 2. Compute $\ttheta^{(j)} = \arg \max_{\theta \in \Theta} \, \sum_{i=1}^{n} g^{(j)}_{n,i} \, \log f_\theta(X_i)$\,.  \\[1ex]
    \textbf{Output:} $\{\ttheta^{(1)}, \ttheta^{(2)}, \ldots, \ttheta^{(B)}\}$.
    \vspace{0.5ex}
\end{algorithm}

\subsection{Survey-adjusted posterior bootstrap}
\label{s:S-PB}

Despite its useful properties, the S-WLB remains a likelihood-based procedure and does not explicitly incorporate prior information.
In many survey applications, however, reliable prior knowledge may be available from previous studies, historical datasets, domain knowledge, or pilot surveys. Incorporating such information can improve finite-sample estimation while facilitating Bayesian inference under informative sampling. A natural framework for addressing this limitation is the {\em posterior bootstrap} of \citet{pompe2021}, which extends the WLB by incorporating prior information through a penalization of the weighted log-likelihood. We combine this prior-incorporation strategy with the survey-adjusted random weight construction introduced in \eqref{eq:swlb} to develop a {\em Survey-adjusted Posterior Bootstrap} (S-PB).
Let $\pi(\theta)$ denote a prior density on $\theta$. Assuming that the prior factorizes across the coordinates of $\theta$,
$\pi(\theta)=\prod_{k=1}^{K}\pi_k(\theta_k)$,
the prior contribution is weighted by a non-negative vector
$\ug_0=(g_{0,1},\ldots,g_{0,K})^{\!\top}$,
where $g_{0,k}$ controls the influence of the prior on the corresponding parameter. The penalization term is given by
$\ug_0^{\top}\log\pi(\theta)
= \sum_{k=1}^{K} g_{0,k}\log\pi_k(\theta_k)$.
A posterior bootstrap sample is then obtained by solving
\begin{equation}
\label{eq:PB}
\tilde\theta_n
=
\arg\max_{\theta\in\Theta}\, 
\Big\{ \sum_{i=1}^n g_i
\log f_\theta(X_i) ~+~ \ug_0^{\top}\log\pi(\theta) \Big\}, 
\qquad
g_{1},\ldots,g_{n}
~\stackrel{\text{i.i.d.}}{\sim}~
\mathrm{Exp}(1).
\end{equation}
When $\ug_0=g_0\mathbf{1}_K$, where $g_0$ is a non-negative scalar and $\mathbf{1}_K$ is the $K$-dimensional vector of ones, the posterior bootstrap resembles the Weighted Bayesian Bootstrap of \citet{NewtonPolsonXu2021}. Instead of a fixed prior weight $g_0$, they use a random weight of the form $\lambda g$, where $g\sim\mathrm{Exp}(1)$ and $\lambda$ is a regularization parameter, typically chosen by cross-validation.

Under a survey setting, we replace the exponential weights in \eqref{eq:PB} with the survey-adjusted random weights defined in \eqref{eq:wlb} and obtain our samples by solving 
\begin{equation}
\begin{split}
\tilde\theta_n
&=
\arg\max_{\theta\in\Theta}\, 
\Big\{ \sum_{i=1}^n g_{n,i}
\log f_\theta(X_i) ~+~ \ug_0^{\top}\log\pi(\theta) \Big\},\\ 
 & g_{n,i}
\,\stackrel{\text{ind.}}{\sim}
\mathrm{Gamma}(1, \tilde w_{n,i}), \ i = 1, 2,\ldots,n,
\end{split}
\end{equation}
within a bootstrap framework as outlined in Algorithm~\ref{alg:S-PB}. Setting \(\ug_0=\mathbf{0}\) recovers the S-WLB. 

\RestyleAlgo{ruled}
\begin{algorithm}[ht]
  \caption{The Survey-adjusted Posterior Bootstrap 
  \label{alg:S-PB}} 
  \vspace{0.5ex}
  \textbf{Input:} Observed survey sample $\uX_n = (X_1, X_2, \ldots, X_n)$.\\[1ex]
  ~~~~~~~~~~ Scaled survey weights  $\tilde{\uw}_n = (\tilde w_{n,1},\tilde w_{n,2},\ldots,\tilde w_{n,n})$ with $\sum_{i=1}^n \tilde w_{n,i} = n$.\\[1ex]
  For $j = 1, 2, \ldots, B$: \hfill $\triangleright$ (in parallel)\\[1ex]
    ~~~~ 1. Draw random weights, $g^{(j)}_{n,i} \,\stackrel{\text{ind.}}{\sim} \text{Gamma}(1, \tilde w_{n,i})$, $i=1,2,\ldots, n$.\\[1ex]
    ~~~~ 2. Compute $\ttheta^{(j)} = \arg \max_{\theta \in \Theta} \big\{ \sum_{i=1}^{n} g^{(j)}_{n,i} \, \log f_\theta(X_i) ~+~ \ug_0^{\top}\log\pi(\theta) \big\}$.  \\[1ex]
    \textbf{Output:} $\{\ttheta^{(1)}, \ttheta^{(2)}, \ldots, \ttheta^{(B)}\}$.
    \vspace{0.5ex}
\end{algorithm}

As in the posterior bootstrap of \cite{pompe2021}, the proposed S-PB requires specification of the prior calibration parameter \(\ug_0\). In Section~\ref{s:theory_spb}, we establish the second-order asymptotic properties of S-PB under a fixed choice of \(\ug_0\). In practice, however, we adopt a data-dependent calibration based on the observed information matrices, as described in Section~\ref{s:choice_g0}.
Like S-WLB, S-PB inherits the computational advantages of the WLB, requiring only repeated optimization of randomized objective functions rather than posterior sampling. By incorporating prior information through the posterior bootstrap framework, we establish that the S-PB additionally enables Bayesian inference under informative sampling while preserving the first-order frequentist validity of the S-WLB and incorporating the appropriate second-order contribution of the prior. Consequently, when informative priors are available, the S-PB yields more accurate point estimates and shorter credible intervals than the S-WLB while maintaining valid frequentist uncertainty quantification.


\section{Theoretical properties}
\label{s:theory}

In this section, we establish the theoretical properties of the proposed methods. We begin by establishing the asymptotic behavior of the S-WLB under informative sampling, showing that it reproduces the asymptotic distribution of the PMLE. We then establish a second-order Edgeworth expansion for the S-PB, showing that it shares the same first-order asymptotic distribution as the S-WLB, with prior information affecting only the second-order terms.

We first introduce the probability structure underlying the observed survey sample. The observed sample \(\uX_n\) and the associated sampling weights \(\uw_n\) arise jointly from the superpopulation model $P_{\theta_0}^N$, governing the finite population variables \(\uX_N^*\), and the sampling design $P_{\pi,N}$, governing the inclusion indicators \(\udelta_N\).
Let $\uX_\infty^*=\{\uX_N^*\}_{N\ge1}$ and $\udelta_\infty=\{\udelta_N\}_{N\ge1}$
denote the sequences of finite population variables and sampling-indicator vectors, respectively. We define \(P_*\) as the joint probability measure governing
\((\uX_\infty^*,\udelta_\infty)\), induced by the sequence of superpopulation models
\(\{P_{\theta_0}^N\}_{N\ge1}\) and the sequence of sampling designs \(\{P_{\pi,N}\}_{N\ge1}\). Finally, let $\cP_\Pi
= \{P_{\pi,N}:\pi\in\Pi,\ N\ge1\}$ denote the class of sampling designs.
Let $\|\cdot\|$ denote the Euclidean norm, and the notation $\asP$ is read almost-sure convergence under $P_*$, \ie for $P_*$ almost every realization of \((\uX_\infty^*,\udelta_\infty)\).

For each $n$, the random weights in S-WLB and S-PB can be expressed as
$$g_{n,i}=\tilde w_{n,i}\,Y_i, \qquad Y_i \overset{\mathrm{i.i.d.}}{\sim}\mathrm{Exp}(1), \qquad i = 1, 2, \ldots,n.$$ 
The auxiliary random weights $\uY_{n} = \{Y_i\}_{i=1}^n$
are independent of the observed data $\uX_n$ and sampling weights $\uw_n$. This representation is equivalent to the Gamma construction introduced in \eqref{eq:swlb} and separates the randomness arising from the survey design and the bootstrap weights.
A sequence of estimators generated by the S-WLB algorithm is defined on the product measurable space
$(\Omega, \cA, P)
=
(\Omega_*\times\Omega_Y,\,
\cA_*\times\cA_Y,\,
P_*\times P_Y)$,
where \((\Omega_*,\cA_*,P_*)\) governs the observed samples and sampling weights, and \((\Omega_Y,\cA_Y,P_Y)\) governs the auxiliary random weights. Let $P(\cdot\mid\uX_n,\uw_n)$
denote the conditional distribution given the observed sample and sampling weights, viewed as random elements on \((\Omega_*,\cA_*,P_*)\).
To establish the theoretical properties of the proposed methods, we impose regularity conditions on the class of sampling designs $\cP_\Pi$ and the data-generating model $\cP_\Theta$, and the auxiliary random weights. For the S-PB, we additionally require regularity conditions on the prior distribution. The assumptions on the sampling design, the auxiliary random weights and the prior distribution are stated below, while those on the model are deferred to the Supplement.

\beginsurvey

\begin{condition}[Non-zero inclusion probability]
\label{c:survey-1}
    For some $\gamma \geq 1$, the sequence of  inclusion probabilities $\{\pi_N\}_{N \geq 1}$ is such that \ $\sup_{N \geq 1} \big({1}/{\pi_N} \big) ~\leq~ \gamma ~<~ \infty\,$, $\ \asP\,.$
\end{condition}
Condition~\ref{c:survey-1} requires the inclusion probabilities to be uniformly bounded away from zero, ensuring that the sampling weights remain finite. 
This condition ensures that no subset of the population is systematically excluded by the sampling design, thereby allowing samples of any size to contain representative information about the entire population. This assumption underlies nearly all consistency results in the literature \citep{pfeffermann1998parametric, savitsky2016bayesian, savitsky2018scalable, williams2018bayesian, williams2021uncertainty}.

\begin{condition}[Asymptotically independent sampling]
\label{c:survey-2}
    The sequence of second-order pairwise inclusion probabilities $\{\pi_{MN}\}_{M,N \geq 1,\, M \neq N}$ is such that 
    $$\sup_{1 \leq M < N} \left|\,\frac{\pi_{MN}}{\pi_M\pi_N} ~-~ 1 \,\right| ~=~ O\left(N^{-1}\right) \text{ for all } N \geq 2,\ \asP\,.$$
\end{condition}

Condition \ref{c:survey-2} confines the result to designs in which the dependence among the sampled units reduces to zero as the population size increases, 
and is typically assumed in the literature \citep{savitsky2016bayesian, savitsky2018scalable, williams2018bayesian}.

\begin{condition}[Constant sampling fraction]
\label{c:survey-3}
For some constant $c \in (0,1)$, $$\frac{n}{N} \to c \ \text{ as } \ N \to \infty , \ \asP.$$ 
\end{condition}

The constant $c$ is called the {\em sampling fraction}. Condition \ref{c:survey-3} ensures that the sample provides a representative amount of information about the population, and is also commonly assumed in the literature \citep{savitsky2016bayesian, williams2018bayesian, williams2021uncertainty}.

\begin{condition}[Weak dependence among first-order inclusion probabilities]
\label{c:survey-4}
Let

\noindent 
\(
\mathcal{H}
=
\left\{
h \text{ measurable} :
\bbE_{\theta_0}|h(X^*)|^2<\infty
\right\}\),
$X^*\sim f_{\theta_0}$.
For every $h\in\mathcal{H}$ and $r\in\{1,2,3\}$,
\[
\sum_{M<k<l\leq N}
\left|\,
\operatorname{Cov}_{\theta_0}
(\pi_l^{-r}h(X_l^*),
\pi_k^{-r}h(X_k^*))
\right|
~=~ O(N-M), \text{ uniformly over } 0\leq M <N.
\]
\end{condition}

The collection $\mathcal H$ includes, as needed, scalar components of the likelihood score and its derivatives, as well as their products arising in the subsequent asymptotic arguments. Condition~\ref{c:survey-4} limits the dependence induced by the first-order inclusion probabilities across the population units, requiring the cumulative absolute covariance to grow at most linearly with the number of units considered. The condition is automatically satisfied when each unit's inclusion probability depends only on its own characteristics, owing to the independence of the population units. More generally, it allows the inclusion probabilities to depend on the entire population, provided that the resulting dependence across units is sufficiently weak.

Conditions~\ref{c:survey-2} and~\ref{c:survey-4} hold simultaneously for several commonly used sampling designs. They follow readily for simple random sampling without replacement and Poisson sampling with unit-specific inclusion probabilities. We further establish that both conditions hold under stratified simple random sampling without replacement and fixed-size rejective sampling with inclusion probabilities proportional to size, under suitable regularity conditions in Section~\ref{s:designs} of the Supplement. We next introduce the regularity conditions on the sequence of auxiliary random weights $\{Y_n\}_{n \geq 1}$.

\beginweights
\begin{condition}[Distributional properties]
\label{c:wt-1}
$\{Y_n\}_{n\ge1}$ are \iid \ according to a non-negative absolutely continuous distribution with $\bbE(Y_n)=1$ and $\mathrm{Var}(Y_n)=1$.
\end{condition}

\begin{condition}[Higher-order moments]
\label{c:wt-2}
$\{Y_n\}_{n\ge1}$ have finite moments of all orders.
\end{condition}

Conditions~\ref{c:wt-1}–\ref{c:wt-2} specify the distributional assumptions on the random weights. Condition~\ref{c:wt-1} suffices for the first-order asymptotic properties of S-WLB, while Condition~\ref{c:wt-2} strengthens it, as needed for the second-order Edgeworth expansion for S-PB.

\subsection{First-order asymptotic theory for S-WLB}

We shall now state our main theorems  for the S-WLB. 
Let 
\begin{align*}
    \Tilde{S}_{w,n}(\theta) = \frac{1}{n} \sum_{i=1}^n\, Y_i \,\tilde w_{n,i} \, \nabla_\theta\, \log f_\theta(X_i), \qquad
    {S}_{w,n}(\theta) = \frac{1}{n} \sum_{i=1}^n  \,\tilde w_{n,i} \, \nabla_\theta\, \log f_\theta(X_i),
\end{align*} 
denote the score functions corresponding to the randomly weighted likelihood $\tL_{w,n}$ defined in \eqref{eq:swlb} and the powered likelihood $L_{w,n}$ defined in \eqref{eq:powered-lik}, respectively.
We establish conditional consistency and conditional asymptotic normality for a sequence of roots of the weighted likelihood equation $\tilde{S}_{w,n}(\theta) = 0$. The proofs are detailed in Section~\ref{s:asymptotics_swlb} of the Supplement.

\begin{theorem}[Conditional consistency]
\label{t:cond_consistency}
Under Conditions~\ref{c:wt-1}, \ref{c:survey-1}--\ref{c:survey-3}, and \ref{c:model-1}--\ref{c:model-7}, there exists an estimator $\ctheta$ of $\theta_0$ which is conditionally consistent, \ie for all $\epsilon > 0$, as $n \to \infty$, 
\begin{equation}
\label{eq:cond_consistent}
    P \left( \|\ctheta - \theta_0\| > \epsilon \mid \uX_n, \uw_n\right) \to 0, \quad \asP,
\end{equation}
and satisfies
\begin{equation}
\label{eq:cond_prob_score1}
    P \left(\tilde{S}_{w,n}(\ctheta) = 0 \mid \uX_n , \uw_n\right) \to 1, \quad \asP\,.
\end{equation}
Moreover, this sequence is essentially unique in the following sense: If $\bar{\theta}_n$ satisfies \eqref{eq:cond_consistent} and \eqref{eq:cond_prob_score1}, then
$P \left( \, \bar{\theta}_n = \ctheta \mid \uX_n, \uw_n\, \right) \to 1, \ \asP\,.$
\end{theorem}

\begin{theorem}[Conditional asymptotic normality]
\label{t:cond_AN}
Under Conditions~\ref{c:wt-1}, \ref{c:survey-1}--\ref{c:survey-4}, and \ref{c:model-1}--\ref{c:model-10}, let $\{\ctheta\}_{n \ge 1}$ be a conditionally consistent sequence of roots of the weighted likelihood equation $\tilde{S}_{w,n}(\theta) = 0$. Then, for any Borel set $A \subset \mathbb{R}^K$, as $n \to \infty$, 
$$
P \left( \sqrt{n} \big(\ctheta - \htheta \big) \in A \mid \uX_n, \uw_n \right) ~\to~ P(Z \in A),\quad \asP\,,
$$
where $\hat{\theta}_n$ denotes a strongly consistent solution of the powered likelihood equation ${S}_{w,n}(\theta) = 0$, and $Z \sim \cN_K \left(0, J(\theta_0)^{-1} I_w(\theta_0) J(\theta_0)^{-1} \right)$ with 
\begin{align}
    \nonumber
    J(\theta) &~=~ -\, \bbE_{\theta_0} \left[\,\nablathetasq \log f_\theta (X^*)\, \right], \quad X^*\sim f_{\theta_0},\\
    \label{eq:I_w}
    I_{w} (\theta) & ~=~ c \, \lim_{N \to \infty} \frac{1}{N}  \,\sum_{l=1}^N \, \bbE_{\theta_0} \left[\frac{1}{\pi_l}   \nabla_\theta \log f_\theta (X_l^*)  \big(\nabla_\theta \log f_\theta (X_l^*) \big)^{\!\!\top} \right].
\end{align}
\end{theorem}

The probabilities in these results refer to the distribution of $\ctheta$, induced by the random weights $\ug_n$. Theorem \ref{t:cond_AN} asserts that under specific regularity conditions, the distribution of the S-WLB estimator $\ctheta$ converges to a normal distribution centered around the PMLE $\htheta$, with covariance matching the PMLE’s sandwich covariance \eqref{eq:dist_PMLE}. The proofs involve expressing the score and information functions of the weighted likelihood $\tL_{w,n}$ as weighted averages of the observed data $\uX_n$ and the auxiliary random weights $\uY_n$. Conditional consistency in Theorem~\ref{t:cond_consistency} is established using a strong law of large numbers for dependent random variables, where the dependence induced by the sampling design is controlled by Condition~\ref{c:survey-2}. Conditions~\ref{c:survey-1} and~\ref{c:survey-3}, together with standard regularity conditions on the model $\cP_\Theta$, ensure the existence and finiteness of the relevant almost-sure limits (see Section~\ref{s:asymptotics_swlb} of the Supplement). For conditional asymptotic normality in Theorem~\ref{t:cond_AN}, Condition~\ref{c:survey-4} additionally controls the dependence in higher-order survey-weighted quantities, ensuring their almost-sure convergence. The proof then applies the Lindeberg--Feller Central Limit Theorem to $\uY_n$, followed by Slutsky's Theorem. This approach avoids the empirical process arguments commonly used to handle dependence in the survey sampling literature.

To avoid ambiguity in interpretation, we will refer to the intervals produced by S-WLB as {\em uncertainty intervals}, as they do not qualify as either frequentist confidence intervals or bonafide Bayesian credible intervals. 
The empirical distribution of WLB estimates can be viewed as a posterior-like distribution, arising from gamma-distributed random weights rather than resampling observations as in the classical frequentist bootstrap.
Since the WLB produces random draws of parameter values from a distribution that approximates a Bayesian posterior \citep{newton1994approximate}, it provides a Bayesian-inspired measure of uncertainty, despite the absence of an explicit prior. The S-WLB uncertainty intervals should be interpreted analogously, reflecting posterior-like variability induced by random weighting designed to adjust for the bias under informative sampling.
Theorem \ref{t:cond_AN} highlights that the S-WLB {uncertainty intervals} will asymptotically agree with the confidence intervals based on the PMLE and yield optimal frequentist coverage. A major advantage of S-WLB is that it does not require knowledge of $J(\theta)$ or $I_w(\theta)$ or their plug-in estimators, making it highly adaptable for any complex modeling scheme that admits a  weighted optimization.

\subsection{Second-order asymptotic theory for S-PB}
\label{s:theory_spb}

The proposed S-PB extends the S-WLB by incorporating prior information through a penalized weighted likelihood. It is designed to combine two desirable features: valid uncertainty quantification under informative sampling and principled incorporation of prior information. The random weighting mechanism determines the first-order distribution and is carefully constructed so that its covariance agrees with the sandwich covariance of the PMLE. While this ensures valid uncertainty quantification under informative sampling, it does not reveal whether the prior is incorporated into the resulting posterior-like distribution in a manner analogous to Bayesian inference. Addressing this question requires going beyond first-order asymptotic theory. Under classical Bayesian asymptotics, the Bernstein-von Mises theorem shows that the prior has no effect on the first-order asymptotic approximation, with its influence appearing only through higher-order terms. As a benchmark, suppose that instead of observing an informative survey sample, one had access to an \iid \ sample $\tilde{X}_1, \ldots, \tilde{X}_n$ from the superpopulation $P_{\theta_0}$. Let $\htheta^{\text{MLE}} = \arg \max_\theta \sum_{i=1}^n \log f_\theta(\tilde X_i)$ denote the corresponding maximum likelihood estimator, and define 
$$
I_n(\hat{\theta}_n^{\text{MLE}}) ~=~ \frac{1}{n} \sum_{i=1}^n    \nabla_\theta \,\log f_\theta (\tilde X_i)  \big(\nabla_\theta \log f_\theta (\tilde X_i) \big)^{\!\!\top} \,\Big|_{\theta = \htheta^{\text{MLE}}}
$$
as a consistent estimator of $I(\theta_0)$, where $I(\theta) = \bbE_{\theta_0}\left[\nabla_\theta \log f_\theta (\tilde X)  \nabla_\theta \log f_\theta (\tilde X)^\top \right]$. 
Following \cite{pompe2021}, Laplace approximation to the posterior density of $I_n(\htheta^{\text{MLE}})\, n^{1/2} (\theta - \htheta^{\text{MLE}})$ is
\begin{equation}
\label{eq:oracle_Bayes}
\phi(y)\left[
1 + I(\theta_0)^{-1/2}
\frac{
\left\{\nabla_{\theta}\log\pi(\theta_0)\right\}^{\!\top}\!y
}{n^{1/2}}
+
I(\theta_0)^{-3/2}
\frac{P_{3}(y)}{n^{1/2}}
\right]
~+~
o\!\left(n^{-1/2}\right), \quad \text{a.s.}\, P_{\theta_0}\,,
\end{equation}
where $\phi$ denotes the $K$-dimensional standard normal density, and $P_{3}$ is a third-order polynomial independent from the prior. This expansion shows that the leading normal approximation is independent of the prior, with prior information entering only through the $O(n^{-1/2})$ correction involving $\nabla_\theta\log\pi(\theta_0)$. The following theorem establishes the analogous expansion for the proposed S-PB under informative sampling, showing that it preserves the first-order validity of the S-WLB while incorporating prior information through the second-order term. To that end, we require the following regularity condition on the prior.

\beginprior
\begin{condition}[Smoothness of the log-prior]
\label{c:prior}
    The function $\log\pi(\theta)$ is measurable, upper bounded on $\Theta$, and three times continuously differentiable in a neighborhood of $\theta_0$.
\end{condition}

Let $\ug_0 \cdot \nabla_\theta \log \pi(\theta_0)$ denote the element-wise product of $\ug_0$ and $\nabla_\theta \log \pi(\theta_0)$. Further, define the prior-penalized weighted log-likelihood objective as \(
\tilde Q_{w,n}(\theta)
= \log \tilde L_{w,n}(\theta) ~+~\ug_0^\top \log \pi(\theta)
\).

\begin{theorem}[Second-order Edgeworth expansion]
\label{t:edgeworth}
Under Conditions~\ref{c:wt-1}--\ref{c:wt-2},
\ref{c:survey-1}--\ref{c:survey-4},
\ref{c:model-1}--\ref{c:model-12}, and \ref{c:prior},
let $\{\ttheta_n\}_{n \geq 1}$ be a conditionally consistent sequence of roots of
the S-PB estimating equation $\nabla_\theta\,\tilde Q_{w,n}(\theta)=0$.
Then, the density of $I_{w,n}(\htheta)^{-1/2}
    J_{w,n}(\htheta)\,
    n^{1/2}(\ttheta_n-\htheta)$
admits the following second-order Edgeworth expansion:
\begin{equation}
\label{eq:edgeworth_SPB}
\phi(y)\left[
1
+
I_w(\theta_0)^{-1/2}
\frac{
\left\{
    \ug_0\cdot\nabla_\theta\,\log\pi(\theta_0)
\right\}^{\!\top}y
}{n^{1/2}}
+
\frac{P_{3,w,n}(y)}{n^{1/2}}
\right]
+
o\!\left(n^{-1/2}\right),
\quad \asP,
\end{equation}
where $\htheta$ is a strongly consistent solution of the powered likelihood equation $S_{w,n}(\theta) = 0$,
$I_{w,n}(\htheta)$, $J_{w,n}(\htheta)$, and $I_w(\theta_0)$
are as defined in \eqref{eq:I_and_J} and \eqref{eq:I_w},
$\phi$ denotes the $K$-dimensional standard normal density, and
$P_{3,w,n}$ is a third-order polynomial independent of the prior.
\end{theorem}

The proof is detailed in Section~\ref{s:asymptotics_spb} of the Supplement.
Theorem~\ref{t:edgeworth} shows that the leading term of the S-PB distribution is the K-dimensional standard normal density, implying that its first-order asymptotic behavior coincides with that of the S-WLB. Consequently, the S-PB retains the same asymptotic covariance as the PMLE while accounting for informative sampling through the survey-adjusted weighting scheme.
Comparing Theorem~\ref{t:edgeworth} with the oracle Edgeworth expansion in \eqref{eq:oracle_Bayes}, we observe that the role of the prior is preserved under informative sampling. In both expansions, the leading normal approximation is independent of the prior, and prior information enters only through the $O(n^{-1/2})$ correction involving the gradient of the log-prior. Thus, the proposed S-PB incorporates prior information in the same higher-order manner as in conventional Bayesian posteriors, while replacing the first-order covariance structure by the survey-adjusted sandwich covariance appropriate for informative sampling. The remaining polynomial correction, $P_{3,w,n}$ is independent of the prior and captures the higher-order effects of the survey-adjusted weighting.
The result suggests that informative priors can improve finite-sample inference without altering the first-order asymptotic validity, as illustrated in the simulation studies of Section~\ref{s:sims}.

\subsubsection{Calibration of the prior contribution}
\label{s:choice_g0}
Theorem~\ref{t:edgeworth} shows that the prior enters the second-order expansion through the parameter $\ug_0$. We therefore adopt a data-driven choice of $\ug_0$ following the calibration principle of \citet{pompe2021}, suitably modified to account for informative sampling. When the prior factorizes across parameters, we employ the coordinate-specific calibration
\[
\ug_{0}
=
\text{diag} \left(
I_{w,n}(\htheta)^{1/2}\,
J_{w,n}(\htheta)^{-1}\,
I_{w,n}(\htheta)^{1/2}\, \right)
\]
so that each prior component is scaled according to the information contained in the survey-weighted likelihood. In contrast, when a joint multivariate prior is specified, we use a single scalar calibration factor
\(
g_0
=
(1/K)
\operatorname{trace}
\left(
I_{w,n}(\htheta)^{1/2}\,
J_{w,n}(\htheta)^{-1}\,
I_{w,n}(\htheta)^{1/2}\, \right),
\)
which is the survey-weighted analogue of the scalar calibration proposed by \citet{pompe2021}. Here, $J_{w,n}(\htheta)$ and $I_{w,n}(\htheta)$ denote the survey-weighted observed information and score covariance matrices defined in \eqref{eq:I_and_J}. Thus, our calibration extends the approach of \citet{pompe2021} to informative sampling by replacing the i.i.d.\ information matrices with their survey-weighted counterparts while preserving the original calibration principle.


\section{Simulation studies}
\label{s:sims}

We present two simulation studies with two objectives: (a) to illustrate the estimation bias that arises when sampling weights are ignored, and (b) to evaluate the finite-sample performance of the proposed S-WLB and S-PB relative to existing methods. Both simulations are based on a Gaussian model. Simulation 1 considers a diffuse prior, under which S-PB is expected to closely resemble S-WLB, whereas Simulation 2 considers an informative prior to evaluate the benefits of incorporating prior information through S-PB. We compare the proposed methods with the unweighted Bayesian estimator (UBE), which ignores the sampling weights; the pseudo maximum likelihood estimator (PMLE); the Bayesian pseudo posterior estimator \citep[BPPE;][]{savitsky2016bayesian}; and the Bayesian weighted estimator \citep[BWE;][]{gunawan2020bayesian} based on the weighted finite population Bayesian bootstrap.

We begin by outlining the data-generating mechanism for the two simulations. Based on the numerical experiments in \cite{gunawan2020bayesian}, we generate weights by introducing a Gaussian distributed selection variable $Z$ which is correlated with our  variable of interest $X^*$ through a Gaussian copula. The probability of selecting an observation in the sample depends on $Z$ via a probit link function. In both cases, we assume that the weights derived from these probabilities are observed, while the realizations of $Z$, used to compute these probabilities and weights, remain unobserved. 
The data generation steps are given below:

\noindent {\bf Step 1.} Generate a finite population of size $N \in \{10000, 20000, 40000\}$, 
\vspace{-1ex}
\begin{equation*}
    \bigg( \begin{matrix}\, X^*_l \\[-0.5ex] Z_l \,\end{matrix} \bigg) ~\sim~ \cN_2 \bigg( \bigg(\begin{matrix} \,\theta\, \\[-0.5ex] \, \eta \,\end{matrix}\bigg), \bigg( \begin{matrix} \, \sigma_x^2 &  \rho \sigma_x \sigma_z \, \\[-0.5ex] \, \rho \sigma_x \sigma_z & \sigma_z^2 \,\end{matrix}\bigg) \bigg)\,, \quad l = 1,2,\ldots, N,
\end{equation*}
with $\theta = 0$, $\eta=0$, $\sigma_x = 4$, $\sigma_z = 3$ and $\rho = 0.6$. 

\noindent {\bf Step 2.} For each $l \leq N$, define the inclusion probabilities as $\pi_l = \Phi(b_0 + b_1 Z_l)$  and generate the inclusion indicators as $\delta_l \sim \mathrm{Bernoulli}(\pi_l)$, where $\Phi$ is the cumulative distribution function of a standard normal distribution, $b_0 = -1.8$ and $b_1 \in \{0, 0.15\}$. The choice of $b_0$ yields an average sampling fraction of approximately $0.05$, resulting in observed sample sizes $n \approx 500$, $1000$, and $2000$ for the three population sizes, respectively.

\noindent {\bf Step 3.} Let $\{X_i\}_{i=1}^n$ denote the observed sample. Let $\{\tilde{w}_{n,i}\}_{i=1}^n$ denote the scaled sampling weights, where $\Tilde{w}_{n,i} = n {\pi_i}^{-1} \big/ \sum_{j=1}^n {\pi_j}^{-1}$ and $\pi_i$ denotes the inclusion probability of unit $i$.

Note that the sampling design is non-informative when either $\rho=0$ or $b_1=0$, since in both cases the inclusion probabilities are independent of the variable of interest $X^*$. When $b_1=0$, all population units have the same inclusion probability, yielding an equal-probability sampling design. Throughout the simulations, we fix $\rho=0.6$ and consider two scenarios with $b_1=0$ and $b_1=0.15$, corresponding to equal-probability sampling and unequal-probability informative sampling, respectively. Our objective is to infer about the population mean, $\theta$.  For Gaussian distributed data, the weighted maximum likelihood estimators (PMLE, S-WLB), the Bayesian posterior distributions (UBE, BPPE), and the proposed S-PB all admit closed-form expressions. For the BWE, we follow Algorithm 2 of \citet{gunawan2020bayesian}, which is designed for settings with closed-form posterior distributions. For S-WLB, S-PB, and BWE, we generate $B=2000$ bootstrap samples.
We evaluate the estimators in terms of mean squared error (MSE), coverage probability (Cov.), and average length of the corresponding 95\% interval estimates (Int.\,len.). Specifically, we consider confidence intervals for the PMLE, credible intervals for the Bayesian estimators (UBE, BPPE, and BWE), and uncertainty intervals for the proposed S-WLB and S-PB. In addition, we compare each method with an {\em oracle} Bayesian analysis based on an independently generated \iid\ sample from the superpopulation. To this end, we report the Euclidean distance between the corresponding point estimates $\|\hat{\theta} - \hat{\theta}_{\text{oracle}}\|$ (Mean dist.) and the Wasserstein distance ($W_2$ dist.)  between their estimated distributions. For the PMLE, the estimated distribution is represented by its asymptotic normal approximation. Since the oracle and survey-based analyses use different samples, these measures reflect closeness to the idealized \iid\ benchmark, incorporating both sampling variation and differences among the methods of inference. All results are summarized over 100 Monte Carlo replications.

\subsection{Simulation 1: Diffuse prior}
\label{s:sim1}

We first consider a diffuse prior, \(\theta \sim \cN(m_0,\tau^2)\) with $m_0=0.2$5 and $\tau=40$, to illustrate the behavior of the proposed methods with limited prior information. Since the prior is weakly informative relative to the likelihood, its influence is expected to be negligible, so that S-PB should closely resemble S-WLB.

\begin{table}[htp]
\centering
\resizebox{\textwidth}{!}{
\renewcommand{\arraystretch}{1.75}
\begin{tabular}[t]{lcccccccccc}
\toprule
\multicolumn{1}{c}{ } & \multicolumn{5}{c}{\bf Equal-probability sampling} & \multicolumn{5}{c}{\bf Informative sampling} \\
\cmidrule(l{3pt}r{3pt}){2-6} \cmidrule(l{3pt}r{3pt}){7-11}
Method & MSE & Cov. & Int.\,len. & Mean dist. & $W_2$ dist. & MSE & Cov. & Int.\,len. & Mean dist. & $W_2$ dist.\\
\midrule
\rowcolor{gray!15}
\multicolumn{11}{c}{\raisebox{0.35ex}{\textbf{$\bm{n=500}$}}}\\
Oracle & 0.052 (0.007) & 0.94 & 0.828 & \textemdash & \textemdash & 0.036 (0.005) & 0.94 & 0.698 & \textemdash & \textemdash\\
PMLE & 0.033 (0.004) & 0.96 & 0.830 & 0.074 (0.010) & 0.074 (0.010) & 0.107 (0.016) & 0.93 & 1.206 & 0.141 (0.023) & 0.168 (0.026)\\
S-WLB & 0.033 (0.005) & 0.96 & 0.827 & 0.075 (0.010) & 0.075 (0.010) & 0.104 (0.015) & 0.92 & 1.177 & 0.140 (0.022) & 0.164 (0.024)\\
S-PB & 0.033 (0.004) & 0.96 & 0.827 & 0.074 (0.010) & 0.075 (0.010) & 0.105 (0.015) & 0.92 & 1.172 & 0.140 (0.022) & 0.164 (0.024)\\
BPPE & 0.033 (0.004) & 0.96 & 0.828 & 0.074 (0.010) & 0.074 (0.010) & 0.107 (0.016) & 0.77 & 0.699 & 0.143 (0.023) & 0.143 (0.023)\\
BWE & 0.033 (0.005) & 1.00 & 1.419 & 0.074 (0.010) & 0.097 (0.010) & 0.107 (0.016) & 0.93 & 1.180 & 0.142 (0.023) & 0.158 (0.023)\\
UBE & 0.033 (0.004) & 0.96 & 0.829 & 0.074 (0.010) & 0.074 (0.010) & 4.173 (0.068) & 0.00 & 0.695 & 4.141 (0.104) & 4.141 (0.104)\\
\midrule
\rowcolor{gray!15}
\multicolumn{11}{c}{\raisebox{0.35ex}{\textbf{$\bm{n=1000}$}}}\\
Oracle & 0.017 (0.002) & 0.99 & 0.585 & \textemdash & \textemdash & 0.013 (0.002) & 0.95 & 0.494 & \textemdash & \textemdash\\
PMLE & 0.022 (0.003) & 0.95 & 0.580 & 0.049 (0.007) & 0.049 (0.007) & 0.055 (0.007) & 0.94 & 0.891 & 0.068 (0.008) & 0.081 (0.009)\\
S-WLB & 0.023 (0.003) & 0.95 & 0.577 & 0.050 (0.007) & 0.049 (0.007) & 0.055 (0.007) & 0.94 & 0.875 & 0.067 (0.008) & 0.079 (0.008)\\
S-PB & 0.022 (0.003) & 0.96 & 0.581 & 0.049 (0.007) & 0.049 (0.007) & 0.054 (0.007) & 0.94 & 0.877 & 0.067 (0.008) & 0.079 (0.008)\\
BPPE & 0.022 (0.003) & 0.96 & 0.583 & 0.049 (0.007) & 0.049 (0.007) & 0.055 (0.007) & 0.64 & 0.492 & 0.068 (0.008) & 0.068 (0.008)\\
BWE & 0.022 (0.003) & 1.00 & 0.991 & 0.049 (0.007) & 0.060 (0.007) & 0.056 (0.007) & 0.96 & 0.830 & 0.068 (0.008) & 0.075 (0.008)\\
UBE & 0.023 (0.003) & 0.94 & 0.584 & 0.049 (0.007) & 0.049 (0.007) & 4.147 (0.049) & 0.00 & 0.494 & 4.192 (0.069) & 4.193 (0.069)\\
\midrule
\rowcolor{gray!15}
\multicolumn{11}{c}{\raisebox{0.35ex}{\textbf{$\bm{n=2000}$}}}\\
Oracle & 0.008 (0.001) & 0.98 & 0.413 & \textemdash & \textemdash & 0.007 (0.001) & 0.99 & 0.349 & \textemdash & \textemdash\\
PMLE & 0.010 (0.001) & 0.97 & 0.415 & 0.017 (0.003) & 0.017 (0.003) & 0.034 (0.007) & 0.94 & 0.680 & 0.038 (0.006) & 0.046 (0.007)\\
S-WLB & 0.010 (0.001) & 0.97 & 0.414 & 0.017 (0.003) & 0.017 (0.003) & 0.033 (0.007) & 0.95 & 0.675 & 0.038 (0.006) & 0.046 (0.007)\\
S-PB & 0.010 (0.001) & 0.97 & 0.413 & 0.017 (0.003) & 0.017 (0.003) & 0.033 (0.007) & 0.95 & 0.675 & 0.038 (0.006) & 0.046 (0.007)\\
BPPE & 0.010 (0.001) & 0.99 & 0.413 & 0.017 (0.003) & 0.017 (0.003) & 0.034 (0.007) & 0.68 & 0.349 & 0.038 (0.006) & 0.038 (0.006)\\
BWE & 0.010 (0.001) & 1.00 & 0.706 & 0.017 (0.003) & 0.022 (0.003) & 0.034 (0.007) & 0.93 & 0.593 & 0.039 (0.007) & 0.042 (0.007)\\
UBE & 0.010 (0.001) & 0.97 & 0.414 & 0.017 (0.003) & 0.017 (0.003) & 4.062 (0.036) & 0.00 & 0.349 & 4.055 (0.048) & 4.056 (0.048)\\
\bottomrule
\end{tabular}
}
\vspace{1ex}
\caption{\label{tab:sim1} Performance of the competing methods under Simulation 1 (diffuse prior). The table reports the mean squared error (MSE), coverage probability (Cov.), average $95\%$ interval length (Int.\,len.), and the Euclidean distance between the point estimate and the oracle estimate (Mean dist.), and the Wasserstein distance between the corresponding estimated distributions ($W_2$ dist.). Results are based on 100 Monte Carlo replicates, with standard errors reported in parentheses.}
\end{table}

\begin{figure}[htp]
    \centering
    \begin{subfigure}[b]{\textwidth}
         \centering
         \includegraphics[width=0.88\textwidth]{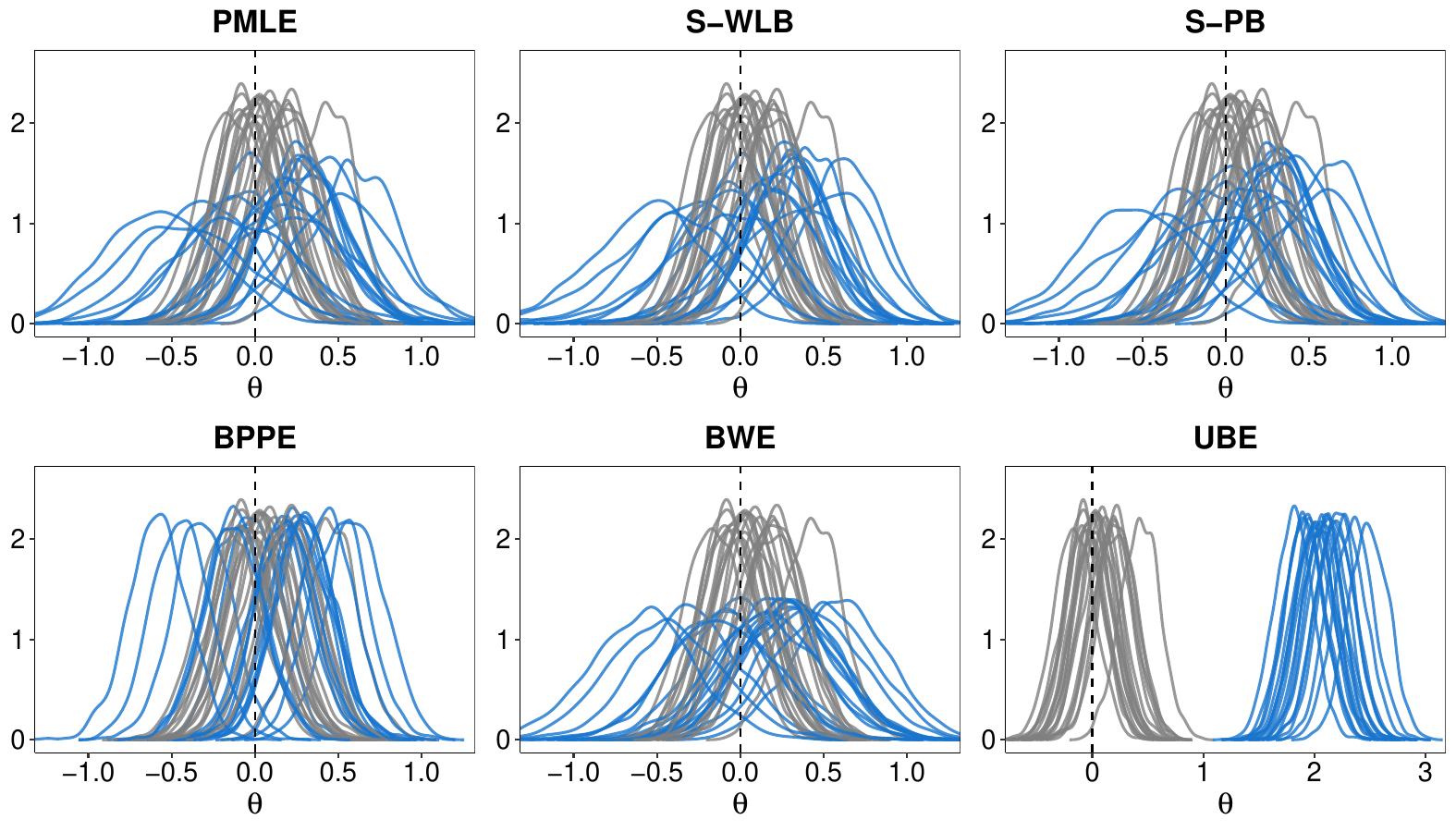}
         \caption{Simulation 1: Diffuse prior}
         \label{fig:post_sim1}
         \vspace{7ex}
     \end{subfigure}
     \begin{subfigure}[b]{\textwidth}
         \centering
         \includegraphics[width=0.88\textwidth]{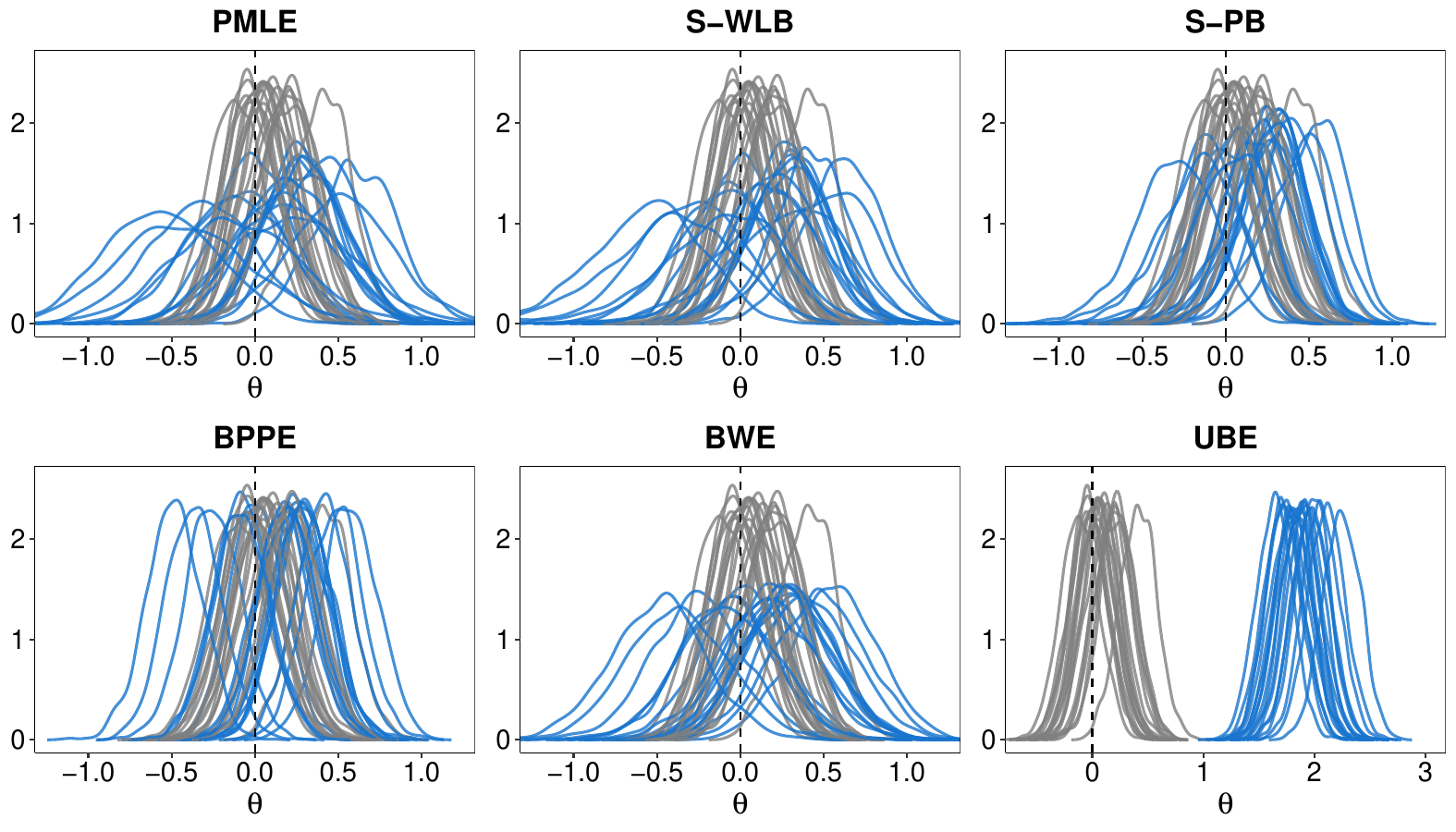}
         \caption{Simulation 2: Informative prior}
         \label{fig:post_sim2}
     \end{subfigure}
    \vspace{0.5ex}
    \caption{Comparison of estimated distributions with the oracle posterior. For each method, the colored curves correspond to the estimated distributions from $20$ Monte Carlo replicates under informative sampling with $n=500$, while the gray curves show the corresponding oracle posteriors. Panel (a) corresponds to Simulation 1 with a diffuse prior, while panel (b) corresponds to Simulation 2 with an informative prior.} 
    \label{fig:density}
\end{figure}

\begin{table}[htp]
\centering
\resizebox{\textwidth}{!}{
\renewcommand{\arraystretch}{1.75}
\begin{tabular}[t]{lcccccccccc}
\toprule
\multicolumn{1}{c}{ } & \multicolumn{5}{c}{\bf Equal-probability sampling} & \multicolumn{5}{c}{\bf Informative sampling} \\
\cmidrule(l{3pt}r{3pt}){2-6} \cmidrule(l{3pt}r{3pt}){7-11}
Method & MSE & Cov. & Int.\,len. & Mean dist. & $W_2$ dist. & MSE & Cov. & Int.\,len. & Mean dist. & $W_2$ dist.\\
\midrule
\rowcolor{gray!15}
\multicolumn{11}{c}{\raisebox{0.35ex}{\textbf{$\bm{n=500}$}}}\\
Oracle & 0.038 (0.005) & 0.94 & 0.762 & \textemdash & \textemdash & 0.030 (0.004) & 0.95 & 0.657 & \textemdash & \textemdash\\
PMLE & 0.033 (0.004) & 0.96 & 0.830 & 0.062 (0.008) & 0.063 (0.008) & 0.107 (0.016) & 0.93 & 1.206 & 0.132 (0.022) & 0.162 (0.026)\\
S-WLB & 0.033 (0.005) & 0.96 & 0.827 & 0.063 (0.008) & 0.063 (0.008) & 0.104 (0.015) & 0.92 & 1.177 & 0.130 (0.021) & 0.156 (0.024)\\
S-PB & 0.025 (0.003) & 0.97 & 0.700 & 0.054 (0.007) & 0.054 (0.007) & 0.058 (0.007) & 0.93 & 0.835 & 0.079 (0.012) & 0.082 (0.012)\\
BPPE & 0.025 (0.003) & 0.98 & 0.763 & 0.054 (0.007) & 0.054 (0.007) & 0.088 (0.012) & 0.73 & 0.658 & 0.112 (0.018) & 0.112 (0.018)\\
BWE & 0.025 (0.003) & 1.00 & 1.239 & 0.054 (0.007) & 0.068 (0.007) & 0.088 (0.012) & 0.93 & 1.070 & 0.112 (0.018) & 0.124 (0.018)\\
UBE & 0.025 (0.003) & 0.98 & 0.764 & 0.054 (0.007) & 0.054 (0.007) & 3.387 (0.055) & 0.00 & 0.655 & 3.260 (0.083) & 3.260 (0.083)\\
\midrule
\rowcolor{gray!15}
\multicolumn{11}{c}{\raisebox{0.35ex}{\textbf{$\bm{n=1000}$}}}\\
Oracle & 0.015 (0.002) & 0.98 & 0.561 & \textemdash & \textemdash & 0.012 (0.002) & 0.96 & 0.479 & \textemdash & \textemdash\\
PMLE & 0.022 (0.003) & 0.95 & 0.580 & 0.047 (0.007) & 0.047 (0.007) & 0.055 (0.007) & 0.94 & 0.891 & 0.067 (0.008) & 0.080 (0.009)\\
S-WLB & 0.023 (0.003) & 0.95 & 0.577 & 0.047 (0.007) & 0.047 (0.007) & 0.055 (0.007) & 0.94 & 0.875 & 0.065 (0.008) & 0.077 (0.008)\\
S-PB & 0.018 (0.002) & 0.96 & 0.534 & 0.041 (0.006) & 0.041 (0.006) & 0.037 (0.005) & 0.94 & 0.720 & 0.046 (0.006) & 0.051 (0.006)\\
BPPE & 0.018 (0.002) & 0.98 & 0.558 & 0.042 (0.006) & 0.041 (0.006) & 0.050 (0.006) & 0.64 & 0.477 & 0.060 (0.007) & 0.060 (0.007)\\
BWE & 0.018 (0.002) & 1.00 & 0.925 & 0.041 (0.006) & 0.050 (0.006) & 0.050 (0.006) & 0.96 & 0.789 & 0.060 (0.007) & 0.066 (0.007)\\
UBE & 0.018 (0.002) & 0.97 & 0.560 & 0.042 (0.006) & 0.041 (0.006) & 3.723 (0.044) & 0.00 & 0.479 & 3.706 (0.061) & 3.707 (0.061)\\
\midrule
\rowcolor{gray!15}
\multicolumn{11}{c}{\raisebox{0.35ex}{\textbf{$\bm{n=2000}$}}}\\
Oracle & 0.008 (0.001) & 0.98 & 0.404 & \textemdash & \textemdash & 0.007 (0.001) & 0.99 & 0.344 & \textemdash & \textemdash\\
PMLE & 0.010 (0.001) & 0.97 & 0.415 & 0.016 (0.002) & 0.016 (0.002) & 0.034 (0.007) & 0.94 & 0.680 & 0.038 (0.007) & 0.047 (0.007)\\
S-WLB & 0.010 (0.001) & 0.97 & 0.414 & 0.016 (0.003) & 0.016 (0.003) & 0.033 (0.007) & 0.95 & 0.675 & 0.038 (0.007) & 0.047 (0.007)\\
S-PB & 0.009 (0.001) & 0.98 & 0.396 & 0.015 (0.002) & 0.015 (0.002) & 0.022 (0.003) & 0.95 & 0.599 & 0.027 (0.004) & 0.032 (0.004)\\
BPPE & 0.009 (0.001) & 0.98 & 0.404 & 0.015 (0.002) & 0.015 (0.002) & 0.031 (0.006) & 0.71 & 0.343 & 0.036 (0.006) & 0.036 (0.006)\\
BWE & 0.009 (0.001) & 1.00 & 0.681 & 0.015 (0.002) & 0.020 (0.002) & 0.031 (0.006) & 0.93 & 0.578 & 0.036 (0.006) & 0.040 (0.006)\\
UBE & 0.009 (0.001) & 0.98 & 0.405 & 0.015 (0.002) & 0.015 (0.002) & 3.846 (0.034) & 0.00 & 0.343 & 3.810 (0.045) & 3.810 (0.045)\\
\bottomrule
\end{tabular}
}
\vspace{1ex}
\caption{\label{tab:sim2} Performance of the competing methods under Simulation 2 (informative prior). The table reports the mean squared error (MSE), coverage probability (Cov.), average $95\%$ interval length (Int.\,len.), and the Euclidean distance between the point estimate and the oracle estimate (Mean dist.), and the Wasserstein distance between the corresponding estimated distributions ($W_2$ dist.). Results are based on 100 Monte Carlo replicates, with standard errors reported in parentheses.}
\end{table}

Table~\ref{tab:sim1} reports the performance of the competing methods across 100 replicates. Under the diffuse prior, S-PB performs nearly identically to S-WLB across all sample sizes and sampling designs, confirming that the influence of the prior is negligible in this setting. Under equal-probability sampling, all methods exhibit comparable estimation accuracy. BPPE, BWE, and UBE perform similarly to the oracle since the sampling design is non-informative. In contrast, under informative sampling, methods that ignore the sampling weights (UBE) exhibit substantial bias, resulting in an inflated MSE and zero coverage. Although BPPE provides consistent point estimates, it fails to deliver optimal coverage. BWE improves coverage by inflating the uncertainty intervals but at the expense of considerably wider intervals. In comparison, the proposed S-WLB and S-PB closely track the PMLE in terms of MSE while producing uncertainty intervals with comparable coverage and lengths. Moreover, both methods consistently yield the smallest discrepancies from the oracle, as measured by the Euclidean and Wasserstein distances, indicating that they provide the closest approximation to the oracle uncertainty distribution under informative sampling.

Figure~\ref{fig:post_sim1} further illustrates these findings by displaying the estimated distributions from $20$ replicates under informative sampling with $n=500$. As expected under a diffuse prior, the estimated distributions produced by the S-WLB and S-PB are closely aligned, confirming that the prior has little influence in this setting. Relative to the oracle posterior, both methods exhibit a wider spread, reflecting the additional variability induced by the informative sampling design that is absent under i.i.d. sampling. In contrast, the BPPE produces distributions with spreads similar to the oracle, leading to systematic undercoverage, while the BWE inflates the spread to improve coverage at the expense of wider intervals. Finally, the UBE is substantially shifted away from the oracle, illustrating the severe bias that results from ignoring the informative sampling design.

\subsection{Simulation 2: Informative prior}
\label{s:sim2}

We next consider an informative prior, \(\theta \sim \cN(m_0,\tau^2)\), with $m_0=0.25$ and $\tau=0.5$, to assess the benefits of incorporating prior information through S-PB. This setting allows S-PB to leverage the additional prior information while accounting for informative sampling.

Table~\ref{tab:sim2} reports the performance of the competing methods across 100 replicates under the informative prior. As in Table~\ref{tab:sim1}, all methods exhibit comparable estimation accuracy under equal-probability sampling. Under informative sampling, the behavior of the competing methods also remains unchanged: UBE continues to exhibit substantial bias due to ignoring the survey weights, resulting in inflated MSE and zero coverage; BPPE continues to suffer from undercoverage; and BWE attains improved coverage at the expense of considerably wider intervals. In contrast to the diffuse-prior setting, however, the proposed S-PB now consistently outperforms S-WLB by effectively incorporating the informative prior, yielding smaller MSE, shorter uncertainty intervals, and smaller Euclidean and Wasserstein distances from the oracle across all sample sizes while maintaining comparable coverage. These results demonstrate that S-PB leverages informative prior information to improve estimation efficiency without sacrificing uncertainty quantification.

\section{Data applications}
\label{s:data_analyses}
We now apply our methodology to two widely known publicly available survey datasets: NHANES and NSDUH. Our analysis compares estimates obtained with and without incorporating survey weights. We compare the performance of S-WLB and S-PB with the frequentist PMLE, Bayesian weighted BPPE and unweighted UBE.


The NHANES, conducted by the National Center for Health Statistics (NCHS), collects data on the health and nutrition of the non-institutionalized U.S. population through interviews and standardized physical exams. NHANES uses a complex, stratified, four-stage sampling design to ensure that the collected data is nationally representative. To demonstrate our methodology, we focus on assessing periodontal disease, the leading cause of adult tooth loss, in relation to specific risk factors, using data on $4{,}299$ subjects from the NHANES 2011--12 and 2013--14 cycles. The disease progression is assessed using its most popular biomarker, the clinical attachment level (CAL, in mm) recorded for each tooth, excluding third molars. The disease status for each subject is represented by the average CAL (avCAL) taken across all teeth. Risk factors include continuous variables such as age (years), body mass index (BMI, $\text{kg}/\text{m}^2$), blood lead level (BLL, $\mu$g/dL), glycosylated hemoglobin (HbA1c, $\%$), and diet quality (measured by the Healthy Eating Index (HEI) per 2015 guidelines), and binary covariates such as gender ($1$ if female, $0$ if male), smoking status, diabetes, hypertension, and medical insurance ($1$ if present, $0$ if absent). Since HEI is used as a covariate, we follow the NCHS recommendations and apply scaled dietary weights (normalized to sum to the sample size) as sampling weights in our analysis. 
We fit a weighted multivariate Gaussian regression, after applying a Box-Cox transformation to correct for skewness in avCAL. 
To incorporate prior information, we use the NHANES 2009–10 cycle as a historical dataset. Specifically, we first fit the proposed S-WLB model to the 2009–10 data and use the resulting means and marginal variances of the regression coefficients to construct independent Gaussian priors for the analysis of the 2011–14 data. The intercept is assigned a flat prior to allow for differences in the overall disease prevalence across survey cycles. The prior contribution for each regression coefficient in S-PB is calibrated using the coordinate-specific scaling factors discussed in Section~\ref{s:choice_g0}. For comparison, S-WLB uses 5000 bootstrap draws, BPPE and UBE each use 5000 posterior draws, and PMLE is reported with its confidence intervals. 

Figure~\ref{fig:data_nhanes} compares the estimated regression coefficients and corresponding interval estimates obtained from the competing methods. The proposed S-WLB closely matches the PMLE in both point and interval estimates across all covariates, confirming that our survey-adjusted weighting scheme accurately reproduces the sampling variability of the PMLE. Incorporating historical information through the proposed S-PB yields moderate shrinkage of the regression coefficients toward the historical NHANES 2009–10 estimates while preserving the overall patterns observed under the survey-weighted methods. The effect of this borrowing is most evident for age, sex, insurance, and BLL. BPPE produces consistently narrower interval estimates, leading to different statistical significance for several covariates. In particular, hypertension is not significant under the PMLE, S-WLB, or S-PB, but is identified as a significant positive risk factor by both BPPE and UBE. The impact of incorporating historical information is illustrated by BMI and HEI. While both covariates are insignificant under the PMLE and S-WLB, they become significant under S-PB, with BMI exhibiting a positive association and HEI a negative association with periodontal disease. Overall, these results demonstrate that the proposed S-PB effectively incorporates historical information while preserving the survey-adjusted uncertainty quantification.

\begin{figure}[htp]
    \centering
    \begin{subfigure}[b]{0.5\textwidth}
         \centering
         \includegraphics[width=\textwidth]{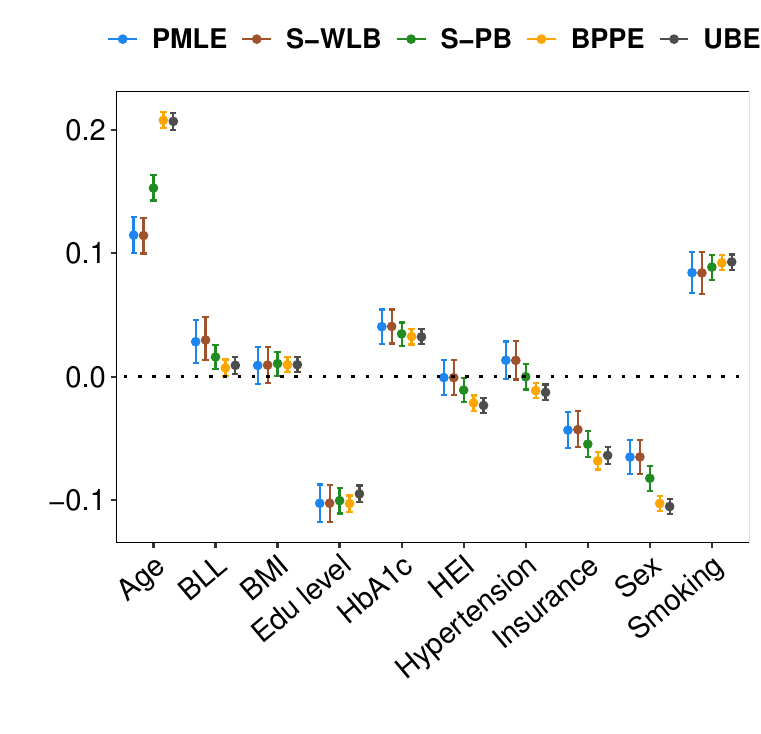}
         \vspace{-0.6in}
         \caption{}
         \label{fig:data_nhanes}
     \end{subfigure}
     \hfill
     \begin{subfigure}[b]{0.49\textwidth}
         \centering
         \includegraphics[width=\textwidth]{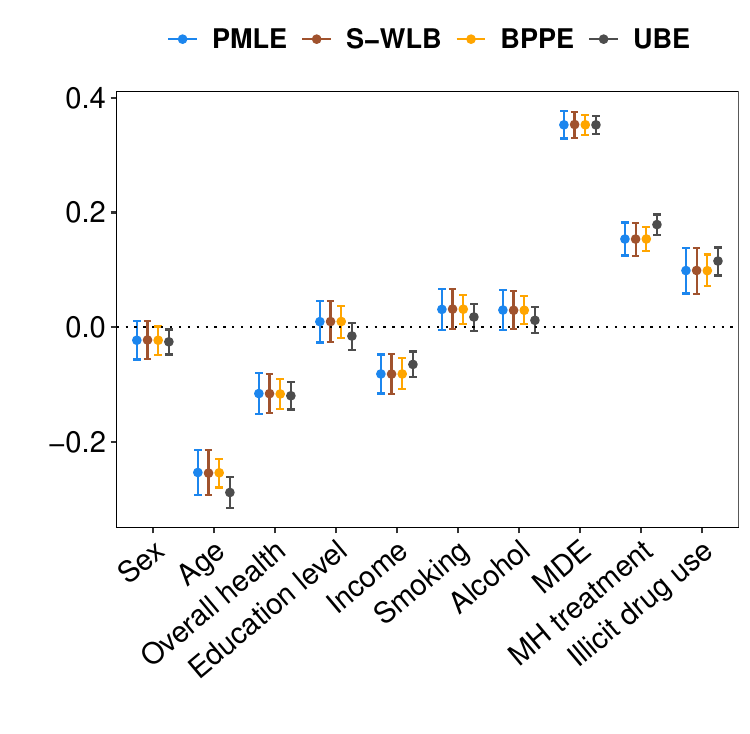}
         \vspace{-0.6in}
         \caption{}
         \label{fig:data_nsduh}
     \end{subfigure}
     
    \vspace{-0.1in}
    \caption{Box-plots of point estimates and $95\%$ intervals of regression parameters, obtained from fitting the proposed S-WLB method and the competing methods to the NHANES (Panel a) and NSDUH (Panel b) datasets.}
    \label{fig:data}
\end{figure}

As a second application, we demonstrate our method on data from the NSDUH, conducted annually by the Substance Abuse and Mental Health Services Administration. The NSDUH provides nationally representative data on tobacco, alcohol, drug use, substance use disorders, and mental health issues among the civilian, non-institutionalized U.S. population, using a multistage, state-based sampling design. To illustrate our algorithm, we focus on assessing suicidal tendencies using the NSDUH 2019 data involving $40{,}640$ subjects. A logistic regression model is applied to a binary response variable indicating whether an individual thought about suicide in the past year ($1$ if yes, $0$ if no). Covariates include ordinal levels for age (1--6, covering $12$ to $65+$ years), overall health (1--4, from poor to excellent), education (1--4, from less than high school to college graduate), and income (1--4, low to high). Binary covariates include sex ($1$ for female, $0$ for male), smoking, alcohol use, illicit drug use, major depressive episodes (MDE) in the past year, and mental health (MH) treatment in the past year ($1$ for presence, $0$ absence). We incorporate the scaled survey weights in our analysis. 
PMLE and S-WLB are implemented by maximizing the weighted logistic log-likelihood using the iteratively reweighted least squares algorithm available in the \texttt{glm()} function in \textsf{R}. For S-WLB, we generate B = 5000 bootstrap replicates. Since no external informative prior is available for this application, the S-PB is not considered. The Bayesian methods (UBE and BPPE) are fitted under a weighted logistic regression model using MCMC with the default weakly informative priors implemented in the \texttt{rstanarm} package, retaining 5000 posterior draws after discarding the initial 1000 iterations as burn-in.

Figure \ref{fig:data_nsduh} presents point estimates and $95\%$ intervals for the factors impacting suicidal tendencies.
Consistent with our findings from NHANES, S-WLB closely aligns with PMLE in both point and interval estimates across all covariates in the NSDUH data, demonstrating its reliability in maintaining optimal coverage. While BPPE produces similar point estimates, its credible intervals are narrower, potentially underestimating uncertainty. Note that BPPE suggests statistical significance for covariates such as smoking, and alcohol use, which remain insignificant under S-WLB and PMLE. Also, sex appears significant under UBE, which does not account for the sampling design, but remains insignificant in the weighted methods.

\section{Discussion}

We develop a novel framework for likelihood-based inference under informative sampling by extending the WLB to dependent survey data. To our knowledge, S-WLB provides the first adaptation of the WLB to complex survey settings. Our proposed S-WLB employs a carefully constructed random weighting scheme that yields an asymptotic covariance matching the PMLE sandwich covariance, providing valid uncertainty quantification without plug-in variance estimation. The method requires only repeated weighted likelihood optimization and is naturally parallelizable, making it computationally scalable.
We further develop S-PB to incorporate prior information with this framework. A second-order Edgeworth expansion demonstrates that the prior contributes through higher-order terms, as in conventional Bayesian inference, while preserving the first-order asymptotic properties of S-WLB.

Our current development focuses on finite-dimensional parametric models. An important direction for future work is to extend the proposed random-weighting framework to semiparametric and nonparametric models, as well as other complex modeling settings fitted to large survey datasets where conventional Bayesian computation may be difficult. Further investigation of prior calibration and the incorporation of historical information from multiple surveys would also broaden the applicability of S-PB.

\section*{Data availability statement}

The datasets analyzed in this study are publicly available. Data from the NHANES for the 2009–2010, 2011–2012, and 2013–2014 cycles are available at \href{https://wwwn.cdc.gov/nchs/nhanes}{wwwn.cdc.gov/nchs/nhanes}. Data from the 2019 NSDUH are available at \href{www.samhsa.gov/data/data-we-collect/nsduh-national-survey-drug-use-and-health/datafiles/2019}{https://www.samhsa.gov/data/data-we-collect/nsduh-national-survey-drug-use-and-health/datafiles/2019}. Code
for implementing our methods are available at the GitHub repository, \href{https://github.com/das-snigdha/S-WLB}{das-snigdha/S-WLB}.

\section*{Declaration of generative AI use}

OpenAI's ChatGPT (GPT-5.6) was used to assist with language refinement, improving the clarity of mathematical arguments, and reviewing R code for potential errors. All AI-assisted suggestions were reviewed by the authors, who take full responsibility for the manuscript.

\bibliographystyle{apalike}
\bibliography{ref}

\appendix
\newpage

\section*{\LARGE Supplementary Materials}
This Supplement provides detailed proofs of the conditional consistency and asymptotic normality of S-WLB and the second-order Edgeworth expansion for S-PB. It also verifies the survey design Conditions~\ref{c:survey-2} and \ref{c:survey-4} under stratified simple random sampling and fixed-size rejective probability-proportional-to-size sampling, along with auxiliary theoretical results.

\section{Asymptotic results for the S-WLB}
\label{s:asymptotics_swlb}

\subsection{Framework}

The notation and proofs of our results closely follow the techniques developed
in Chapter~3 of \citet{newton1991thesis}, with necessary modifications to
accommodate complex survey data. We follow a similar progression of
preliminary results leading to the conditional consistency and conditional
asymptotic normality results in Theorems~1 and~2, respectively. Throughout
these proofs, we use a strong law of large numbers (SLLN) for dependent random
variables \citep[Theorem~2]{Petrov2014}, stated here as
Theorem~\ref{t:gen_as_conv}, in place of the SLLN for i.i.d.\ random variables
employed by \citet{newton1991thesis}. The dependence arising from the
survey design is controlled through Conditions~S1--S4, while the regularity
conditions on the model ensure the existence and finiteness of the relevant
almost-sure limits.

Our regularity conditions on the model largely follow the framework of
\citet{newton1991thesis}, with stronger moment requirements where needed
to accommodate the dependence induced by the survey design. In particular,
we require the log-likelihood ratio to be square-integrable
(Condition~\ref{c:model-2}), and the first and second derivatives of the log density, as well as products of the first derivatives, to be uniformly dominated over a neighborhood $B$ of the true parameter $\theta_0$ by functions with finite second moments (Condition~\ref{c:model-6}). We further impose analogous finite second-moment conditions on envelopes for the third derivatives and products of the first and second derivatives (Conditions~\ref{c:model-9} and~\ref{c:model-10}). These conditions strengthen the corresponding integrability requirements in \citet{newton1991thesis} and allow us to control the dependence in survey data, as opposed to \iid\ data.

Conditional asymptotic normality in Theorem~\ref{t:cond_AN} is established using Slutsky's theorem after applying (i) the Lindeberg--Feller central limit theorem to the vector of unnormalized auxiliary random weights $\mathbf{Y}_n$, whose elements are i.i.d., and (ii) the SLLN to the corresponding data-dependent quantities. This approach allows us to establish the required asymptotic results without invoking empirical-process arguments for the dependence induced by the survey design.

Throughout the Supplement, we adopt the probability framework
and notation introduced in Section~\ref{s:theory} of the main document. In particular,
$P_*$ denotes the joint probability measure governing the sequence of finite
populations and sampling indicators, while
$P(\,\cdot\mid\mathbf{X}_n,\mathbf{w}_n)$ denotes the conditional
distribution with respect to the auxiliary random weights, given the
observed sample and sampling weights. The notation a.s.$[P_*]$ is read as
almost surely under $P_*$.

We now define the notions of conditional convergence used throughout the
proofs. Let $U,V_1,V_2,\ldots$ be random variables defined on
$(\Omega,\mathcal{A})$. We say that $V_n$ converges to $U$ in conditional
probability, a.s.$[P_*]$, if, for every $\epsilon>0$, as $n\to\infty$,
\[
P\left(
\|V_n-U\|>\epsilon
\,\middle|\,
\mathbf{X}_n,\mathbf{w}_n
\right)
\longrightarrow 0,
\quad \text{a.s.}[P_*].
\]
This convergence is denoted as
$V_n \cp U, \asP$. A sequence of estimators $\{\Bar{\theta}_n\}$ of $\theta_0$ is said to be {\em conditionally consistent} if $\ \Bar{\theta}_n \cp \theta_0, \ \asP\,.$

\subsection{Preliminaries}
\label{s:prelim}
We now present several preliminary results that are fundamental to establish the proofs of Theorems \ref{t:cond_consistency} and \ref{t:cond_AN}, which address conditional consistency and conditional asymptotic normality for a sequence of estimators resulting from our S-WLB algorithm. 

\vspace{-1ex}
\begin{lemma}
\label{l:as_conv}
Let $(X_1, X_2, \dots, X_n)$ denote the sampled data and $(w_1, w_2, \ldots, w_n)$ be the sampling weights. Let $W_n = \sum_{i=1}^n w_i$ and $\tilde{w}_{n,i} = nw_i/W_n$ denote the scaled sampling weights that add up to $n$. Let $q \in \{1, 2, 3, 4\}$ be fixed and $h$ be a real valued measurable function such that $\bbE_{\theta_0}|h(X^*)|^2 < \infty$, where $X ^* \sim f_{\theta_0}$. 
Under Conditions~\ref{c:survey-1}--\ref{c:survey-3} and, for
$q>1$, Condition~\ref{c:survey-4}, as $n \to \infty$,

\begin{equation}
\label{eq:as_conv_0}
    \frac{1}{n} \sum_{i=1}^{n} \tilde{w}_{n,i}^q \,h(X_i) ~\to~ \xi_q \quad \asP.
\end{equation}
where $$\xi_q ~=~ c^{\,q-1}\lim_{N \to \infty}\, \frac{1}{N}\sum_{l=1}^N \bbE_{\theta_0}\big[ {\pi_l}^{-(q-1)} \,h(X^*_l) \big]\,.$$
\end{lemma}

\begin{proof}
First, note that the weighted average in \eqref{eq:as_conv_0} can be written as
\begin{align}
    \nonumber
    \frac{1}{n}\sum_{i=1}^n \tilde{w}_{n,i}^{\,q}h(X_i)
    &=
    \frac{n^{q-1}}{W_n^q}
    \sum_{i=1}^n w_i^q h(X_i)\\
    \label{eq:wtd_avg_q}
    &=
    \frac{1}{
    \left(
    \frac{1}{N}\sum_{l=1}^N \pi_l^{-1}\delta_l
    \right)^q}
    \left(\frac{n}{N}\right)^{q-1}
    \frac{1}{N}
    \sum_{l=1}^N
    \pi_l^{-q}\delta_l h(X_l^*).
\end{align}
Define $Z_l=\pi_l^{-q}\,\delta_l\, h(X_l^*).$  Then, using iterated expectations we have, for $l \geq 1$,
\begin{align*}
    \bbE(Z_l) &~=~ \bbE \big[\pi_l^{-q}\,\delta_l \, h(X_l^*) \big] ~=~  \bbE_{\theta_0} \  \big[ h(X_l^*)\, \pi_l^{-q}\,  \bbE_{\udelta} \big(\delta_l \, \mid X^*_l \big) \big] ~=~  \bbE_{\theta_0} [\pi_l^{-(q-1)}\,h(X_l^*)]\,,\\[1.5ex]
    \text{Var}(Z_l) &~=~ 
    \bbE \big(Z_l^2 \big) - \bbE \big(Z_l \big)^2 ~\leq~ \bbE \big(Z_l^2 \big) ~=~ \bbE \big[ \pi_l^{-2q}\, \delta_l \, h(X_l^*)^2 \big] ~=~   \bbE_{\theta_0} \  \bbE_{\udelta} \big[~\pi_l^{-2q}\,\delta_l \, h(X_l^*)^2 \mid \uX^*_N \big] \\
    &~=~ \bbE_{\theta_0} \big[ h(X_l^*)^2\, ~\pi_l^{-2q}\, \bbE_{\udelta} \big(\delta_l \, \mid \uX^*_N \big) \big] ~=~ \bbE_{\theta_0} \big[ ~{\pi_l}^{-(2q-1)}\ h(X_l^*)^2 \big]\,,
\end{align*}
and for $k (\neq l) \geq 1$,
\begin{align*}
    \text{Cov}(Z_l, Z_k) 
    & ~=~ \bbE(Z_l Z_k) - \bbE (Z_l)\bbE (Z_k) \\
    & ~=~  \bbE \big[ \pi_l^{-q}\, \pi_k^{-q}\,\delta_l \delta_k h(X_l^*)h(X_k^*)\big] ~-~ \bbE_{\theta_0} [\pi_l^{-(q-1)}\,h(X_l^*)]\,\bbE_{\theta_0} [\pi_k^{-(q-1)}h(X_k^*)] \\
    & ~=~ \bbE_{\theta_0} \bigg[ \pi_l^{-(q-1)} \pi_k^{-(q-1)}\frac{\pi_{lk} }{\pi_l\,\pi_k}\, h(X_l^*)\,h(X_k^*) \bigg] ~-~ \bbE_{\theta_0} [\pi_l^{-(q-1)}\,h(X_l^*)]\,\bbE_{\theta_0} [\pi_k^{-(q-1)}h(X_k^*)]\\
    & ~=~ \bbE_{\theta_0} \left[ \pi_l^{-(q-1)} \pi_k^{-(q-1)} h(X_l^*)\,h(X_k^*) \left(\frac{\pi_{lk}}{\pi_l\, \pi_k} - 1\right) \right] \\
    & \qquad \qquad ~+~ \text{Cov}_{\theta_0}\Big(\pi_l^{-(q-1)}h(X_l^*), \pi_k^{-(q-1)}h(X_k^*) \Big)
\end{align*}
We now show that $\{Z_N:N\geq 1\}$ satisfies the condition in
Theorem \ref{t:gen_as_conv} with $r=1$. Let $T_N=\sum_{l=1}^N Z_l$, $N \geq 1$,
with $T_0=0$. For $N>M\geq 0$,
\begin{align*}
    \Var(T_N-T_M)
    &=
    \sum_{l=M+1}^N \Var(Z_l)
    +
    2\sum_{l=M+2}^N
    \sum_{k=M+1}^{l-1}
    \Cov(Z_l,Z_k).
\end{align*}
By condition \ref{c:survey-1}, the inclusion probabilities are bounded
away from zero. Hence, since $q$ is fixed and
$\bbE_{\theta_0}|h(X^*)|^2<\infty$,
\begin{align*}
    \sum_{l=M+1}^N \Var(Z_l)
    &\leq
    \sum_{l=M+1}^N
    \bbE_{\theta_0}
    \left[
    \pi_l^{-(2q-1)}h(X_l^*)^2
    \right] \leq D_1(N-M)
\end{align*}
for some finite constant $D_1$. Similarly, 
\begin{align*}
    2\sum_{l=M+2}^N
    \sum_{k=M+1}^{l-1}
    \Cov(Z_l,Z_k)
    &\leq
    2\sum_{l=M+2}^N
    \sum_{k=M+1}^{l-1}
    \bbE_{\theta_0}
    \left[
    \frac{|h(X_l^*)h(X_k^*)|}
    {\pi_l^{q-1}\pi_k^{q-1}}
    \left|
    \frac{\pi_{lk}}{\pi_l\pi_k}-1
    \right|
    \right] \\
    & \qquad ~+~ 2 \sum_{l=M+2}^N
    \sum_{k=M+1}^{l-1}
    \Big| \, \text{Cov}_{\theta_0}(\pi_l^{-(q-1)}h(X_l^*), \pi_k^{-(q-1)}h(X_k^*)) \, \Big|
\end{align*}
For the first term, Conditions~\ref{c:survey-1}
and~\ref{c:survey-2} give, for some $D_2<\infty$,
\begin{align*}
2\sum_{l=M+2}^{N}\sum_{k=M+1}^{l-1}
E_{\theta_0}\left[
\frac{|h(X_l^*)h(X_k^*)|}
{\pi_l^{q-1}\pi_k^{q-1}}
\left|
\frac{\pi_{lk}}{\pi_l\pi_k}-1
\right|
\right]  & 
\leq  2\sum_{l=M+2}^{N}\sum_{k=M+1}^{l-1}
E_{\theta_0}\left[
\frac{|h(X_l^*)h(X_k^*)|}
{\pi_l^{q-1}\pi_k^{q-1}}
\sup_{k < l} \left|
\frac{\pi_{lk}}{\pi_l\pi_k}-1
\right|
\right]\\
& \leq
D_2\sum_{l=M+2}^{N}l^{-1}(l-M-1) ~\le~ D_2(N-M).
\end{align*}
By Condition~\ref{c:survey-4}, the second term satisfies
\[
2\sum_{l=M+2}^{N}\sum_{k=M+1}^{l-1}
\left|
\operatorname{Cov}_{\theta_0}
\left\{
\pi_l^{-(q-1)}h(X_l^*),
\pi_k^{-(q-1)}h(X_k^*)
\right\}
\right|
~\leq~ D_3(N-M)
\]
for some $D_3<\infty$.  Therefore,
\begin{equation}
\label{eq:var_cond_lq}
    \Var(T_N-T_M)\leq D(N-M)
\end{equation}
for all $N>M\geq 0$ and some finite constant $D$.
Thus, by Theorem \ref{t:gen_as_conv}, as $N\to\infty$,
\begin{equation}
\label{eq:as_num_q_1}
    \frac{1}{N}
    \sum_{l=1}^N
    \pi_l^{-q}\,\delta_l \, h(X_l^*)
    ~\to~ \eta_q
    \quad \asP,
\end{equation}
where
\[
    \eta_q
    =
    \lim_{N\to\infty}
    \frac{1}{N}
    \sum_{l=1}^N
    \bbE_{\theta_0}
    \left[
    \pi_l^{-(q-1)}h(X_l^*)
    \right].
\]
The finiteness of $\eta_q$ follows from condition
\ref{c:survey-1} and the assumption
$\bbE_{\theta_0}|h(X^*)|^2<\infty$.
Taking $h(x)\equiv 1$ and $q=1$ in \eqref{eq:as_conv_0}, we obtain
\[
\frac1N\sum_{l=1}^N\pi_l^{-1}\delta_l
\to 1
\quad \asP.
\]
Since the map $x\mapsto x^q$ is continuous for every fixed $q\ge1$, the
Continuous Mapping Theorem yields
\begin{equation}
\label{eq:as_denom_q}
\left(
\frac1N\sum_{l=1}^N\pi_l^{-1}\delta_l
\right)^q
\to 1
\quad \asP.
\end{equation}
By condition \ref{c:survey-3}, $(n/N)^{q-1} \to c^{q-1}$.
Combining this with \eqref{eq:as_num_q_1} gives
\begin{equation}
\label{eq:as_num_q}
    \left(\frac{n}{N}\right)^{q-1}
    \frac{1}{N}
    \sum_{l=1}^N
    \pi_l^{-q}\delta_l h(X_l^*)
    ~\to~
    c^{q-1}\eta_q
    =
    \xi_q
    \quad \asP.
\end{equation}
The result follows by combining \eqref{eq:wtd_avg_q},
\eqref{eq:as_denom_q}, and \eqref{eq:as_num_q}.
\end{proof}

\begin{lemma}
\label{l:max}
Let $(X_1, X_2, \dots, X_n)$ denote the sampled data and $(w_1, w_2, \ldots, w_n)$ be the sampling weights. Let $W_n = \sum_{i=1}^n w_i$ and $\tilde{w}_{n,i} = nw_i/W_n$ denote the scaled sampling weights that add up to $n$. Let $q \in \{1,2,3,4\}$ be fixed and $h$ be a real valued measurable function such that $\bbE_{\theta_0}|h(X^*)|^2 < \infty$, where $X ^* \sim f_{\theta_0}$. 
Under Conditions~\ref{c:survey-1}--\ref{c:survey-3} and, for $q>1$, Condition~\ref{c:survey-4}, as $n \to \infty$,
\vspace{-1ex}
$$
\frac{1}{n} \,\max_{1 \leq i \leq n} \tilde{w}_{n,i}^q\, |h(X_i)| ~\to~ 0\quad \asP.
$$
\end{lemma}

\begin{proof} 
Observe that
\begin{align*}
    \frac{1}{n} \max_{1 \leq i \leq n} \tilde{w}_{n,i}^q \,|h(X_i)|
    & ~=~ \frac{n^{(q-1)}}{W_n^q} \max_{1 \leq i \leq n}  {w}_{i}^q \,|h(X_i)|\\
    & ~=~ \frac{1}{\big(\frac{1}{N}\sum_{l=1}^N \pi_l^{-1}\,\delta_l \big)^q}\ \Big(\frac{n}{N} \Big)^{q-1} \ \frac{1}{N}\max_{1 \leq l \leq N}\pi_l^{-q}\,\delta_l \, |h(X_l^*)|
\end{align*}
Using \eqref{eq:as_denom_q} and condition \ref{c:survey-3}, it suffices to prove that 
\vspace{-2ex}
\begin{equation*}
    \frac{1}{N}\max_{1 \leq l \leq N}\pi_l^{-q}\,\delta_l \, h(X_l^*) ~\as~ 0\,.
\end{equation*}
The proof follows by defining $Z_l = \pi_l^{-q}\, \delta_l \, h(X_l^*)$, and applying Corollary \ref{c:max_conv} with \eqref{eq:var_cond_lq} and \eqref{eq:as_num_q_1}.
\end{proof}

\begin{Rem}
For $q=1$, Condition~\ref{c:survey-4} is not required in
Lemmas~\ref{l:as_conv}--\ref{l:max}, since the corresponding superpopulation covariance terms involve only independent superpopulation units.
\end{Rem}

\begin{lemma}
\label{l:data-wts}
Let $\{X_l^*\}_{l=1}^N$ denote the finite population generated from
$P_{\theta_0}$ having density $f_{\theta_0}$, and let $\delta_l$ be the corresponding inclusion
indicator under the sampling design. Let $\{X_i\}_{i=1}^n=\{X_l^*:\delta_l=1\}_{l=1}^N$
denote the observed sample, with corresponding sampling weights
$\{w_i\}_{i=1}^n$. Let $W_n = \sum_{i=1}^n w_i$ and $\tilde{w}_{n,i} = nw_i/W_n$ denote the scaled sampling weights that add up to $n$.
Let
\[
    g_{n,i}
    =
    \tilde w_{n,i}Y_i, \quad Y_i \sim \text{Exp}(1),
    \qquad i=1,\ldots,n,
\]
where $Y_1,\ldots,Y_n$ are independent of $(X_n,w_n)$ and satisfy
Condition~\ref{c:wt-1}. Let $h$ be a real-valued measurable function such that $\bbE_{\theta_0}|h(X^*)|^2 < \infty$, where $X ^* \sim P_{\theta_0}$. 
Then, under Conditions~\ref{c:survey-1}--\ref{c:survey-3} and
\ref{c:wt-1},
\[
    \frac{1}{n}
    \sum_{i=1}^n
    g_{n,i}\,h(X_i)
    \;\cp\;
    \bbE_{\theta_0}\!\left[h(X^*)\right],
    \quad
    \asP\,.
\]
\end{lemma}
\begin{proof} Note that 
$$
\frac{1}{n} \sum_{i=1}^n \, g_{n,i} \,h(X_i) ~=~ \frac{1}{n} \sum_{i=1}^n \, Y_i \,\tilde{w}_{n,i} \,h(X_i) 
$$
Using Lemma \ref{l:as_conv} with $q=1$, we have 
\begin{align*}
    \frac{1}{n}\sum_{i=1}^n \tilde{w}_{n,i} \, h(X_i) & ~\to~ \bbE_{\theta_0} [h(X^*)] \quad \asP,\\
    \frac{1}{n}\sum_{i=1}^n \tilde{w}_{n,i} |h(X_i)| & ~\to~ \bbE_{\theta_0} |h(X^*)| \ < \infty \quad \asP,
\end{align*}
and Lemma \ref{l:max} with $q=1$ gives 
\vspace{-3ex}
$$\frac{1}{n} \,\max_{1 \leq i \leq n} \tilde{w}_{n,i}\, |h(X_i)| ~\to~ 0\quad \asP.$$
The proof follows by applying Theorem \ref{t:prob_conv}, whose conditions are satisfied with $a_{n,i} = \tilde{w}_{n,i} \,h(X_i)$.
\end{proof}

Next, we state some regularity conditions on the model that are needed to study the conditional consistency of the maximizers of the weighted likelihood function.

\bigskip
\beginmodel
\textbf{Regularity Conditions 1}
\begin{condition}[Identifiability]
\label{c:model-1}
    For any $\theta_1 \neq \theta_2$, both in $\Theta$, there exists a measurable set $A$ such that $P_{\theta_1}(A) \neq P_{\theta_2}(A)$.
\end{condition}
\begin{condition}
\label{c:model-2}
For the true density $f_{\theta_0}$ and another density $f_{\theta_1}$ 
in the model,
\begin{align*}
    (a)~ & D(f_{\theta_0} \| f_{\theta_1})= \int
    \log\left(
        \frac{f_{\theta_0}(x)}{f_{\theta_1}(x)}
    \right)
    f_{\theta_0}(x)\,d\mu(x) < \infty,\\
    (b)~ & V(f_{\theta_0} \| f_{\theta_1}) = \int
    \left|
        \log\left(
            \frac{f_{\theta_0}(x)}{f_{\theta_1}(x)}
        \right)
    \right|^2
    f_{\theta_0}(x)\,d\mu(x) < \infty,
\end{align*}
where $D(f_{\theta_0} \| f_{\theta_1})$ and $V(f_{\theta_0} \| f_{\theta_1})$ denote the Kullback--Leibler 
divergence and the $V$-divergence, respectively.
\end{condition}

The following theorem is an analog of Theorem 4 of \citet{newton1991thesis}. It establishes that, conditioned on the data, the weighted likelihood function converges to its maximum at the true parameter value. In this and the subsequent theorems, we will adhere closely to the proof techniques outlined in Chapter 3 of \citet{newton1991thesis}.

\begin{theorem}
\label{t:wtd_lik_max}
    Let $\theta_1 (\neq \theta_0) \in \Theta$. Let $E_n \in \cA_2$ be defined by
    $$
    E_n(\omega_1) = \left\{\omega_2 \in \Omega_2: \Tilde{L}_n (\theta_0) > \Tilde{L}_n(\theta_1) \right\}\,,\quad \text{ for } \omega_1 \in \Omega_1.
    $$
    Under conditions \ref{c:model-1}--\ref{c:model-2}, \ref{c:survey-1}--\ref{c:survey-3}, and \ref{c:wt-1}, as $n \to \infty$, $P \left( \,E_n \mid \uX_n, \uw_n \,\right) \to 1 \quad \asP\,.$
\end{theorem}

\begin{proof}
$E_n(\omega_1) \in \mathcal{A}_2$, since the weighted likelihood is a measurable function. 
Also, $\Tilde{L}_n(\theta_0) > 0$ with $P_1$ probability one (for a fixed set of weights). Thus,
$$E_n(\omega_1) = \left\{ \omega_2 \in \Omega_2 : \sum_{i=1}^{n} g_{n,i} \log f_{\theta_0}(X_i) > \sum_{i=1}^{n} g_{n,i} \log f_{\theta_1}(X_i) \right\}\,,$$ 
where the right hand side (RHS) may equal $-\infty$, but not the left hand side (LHS), with $P_1$ probability 1.  Note that
$
E_n(\omega_1) = \left\{ \omega_2 \in \Omega_2 :~ {n}^{-1} \sum_{i=1}^{n} g_{n,i}\;h(X_i) < 0 \right\}$, where $h(X_i) := \log ({f_{\theta_1}(X_i)}/{f_{\theta_0}(X_i)})$. Suppose $X ^* \sim f_{\theta_0}$.
By Jensen's inequality,
\begin{equation*}
    \phi ~=~ \bbE_{\theta_0} [h(X^*)] ~\leq~ \log \ \bbE_{\theta_0} \bigg[ \frac{f_{\theta_1}(X^*)}{f_{\theta_0}(X^*)} \bigg] ~=~  0 \,,
\end{equation*}
where $\phi = -D(f_{\theta_0} \| f_{\theta_1})\ ( > -\infty)$ is as defined in condition \ref{c:model-2}(a). Condition \ref{c:model-1} ensures that the inequality is strict. The proof follows by observing that 
\begin{equation}
\label{eq:wtd_Vi}
    \frac{1}{n} \sum_{i=1}^{n} \tilde{w}_{n,i}\; h(X_i) ~\as~ \phi\,,
\end{equation}
which follows by condition \ref{c:model-2}(b) and applying Lemma \ref{l:as_conv} with $q=1$. Using \eqref{eq:wtd_Vi} and Lemma \ref{l:data-wts}, we have
$$
\frac{1}{n} \sum_{i=1}^{n} g_{n,i}\, h(X_i) ~\cp~ \phi \,(< 0) \quad \asP\,,
$$
which completes the proof.
\end{proof}

If $\Theta$ is finite, then Theorem \ref{t:wtd_lik_max} directly implies that a sequence of maximizers of the weighted likelihood exists, is conditionally consistent and is unique (see Corollary 1 of \citet{newton1991thesis}). The next set of conditions are required to establish the conditional consistency result to a more general class of models.

\bigskip
\textbf{Regularity Conditions 2}
\begin{condition}
\label{c:model-3}
    $\Theta \in \bbR^K$ contains an open ball $B$ containing $\theta_0$.
\end{condition}
\begin{condition}
\label{c:model-4}
    For each $\theta \in B$, all the first partial derivatives of $\log f_\theta(x)$ with respect to the components of $\theta$ exist and are continuous for almost all $x$.
\end{condition}

Note that if $\Tilde{\theta}$ (a maximizer of the weighted likelihood, $\Tilde{L}_n$) is not on the boundary of $\Theta$ and if the log weighted likelihood is differentiable (condition \ref{c:model-4}), then $\Tilde{\theta}$ must be a root of the weighted likelihood equation, $\nablatheta \log \Tilde{L}_n(\theta) ~=~ 0$.

\subsection{Score and information corresponding to the weighted likelihood}
\label{s:score-swlb}
Before proceeding to prove the result on conditional consistency of a sequence of roots of the weighted likelihood equations, we shall define and study properties of the weighted score and information functions. Define the vector of first partial derivatives and matrix of second partial derivatives (when they exist) of $\log f_\theta(X_i)$ as
\begin{equation}
\begin{split}
\label{eq:partials}
    \psi_i(\theta)  ~=~ \nablatheta \log f_\theta(X_i) \ \in \ \bbR^K \,, \qquad
    \psi_i^\prime(\theta)  ~=~ \nablathetasq \log f_\theta(X_i)  \ \in \ \bbR^{K \times K}\,.
\end{split}
\end{equation}
We define a weighted score function and a weighted information matrix corresponding to our weighted likelihood, $\Tilde{L}_{w,n} (\theta) = \prod_{i=1}^n f_\theta(X_i)^{\,g_{n,i}}$, as follows
\begin{align*}
    \Tilde{S}_{w,n}(\theta) & ~=~ \frac{1}{n}\sum_{i=1}^n g_{n,i} \, \psi_i(\theta)\,, \quad
    \Tilde{J}_{w,n}(\theta)  ~=~ - \frac{1}{n}\sum_{i=1}^n g_{n,i} \, \psi_i^\prime(\theta)\,.
\end{align*}
Recall that the powered likelihood \ie the likelihood raised to the power of the scaled sampling weights is given by $$L_{w,n}(\theta) ~=~ \prod_{i=1}^n f_\theta(X_i)^{\tilde{w}_{n,i}}.$$
The score and information functions corresponding to the powered likelihood are given by
\begin{align*}
    S_{w,n}(\theta) & ~=~ \frac{1}{n} \sum_{i=1}^n \tilde{w}_{n,i} \, \psi_i(\theta)\,,\quad
    J_{w,n}(\theta) ~=~ - \frac{1}{n} \sum_{i=1}^n \tilde{w}_{n,i}\, \psi_i^\prime(\theta)\,.
\end{align*}
We further define the following,
\begin{align*}
    I_{w,n} (\theta) & ~=~ \, \frac{1}{n} \sum_{i=1}^n \tilde{w}_{n,i}^2  \ \psi_i(\theta) \, \psi_i(\theta)^\top\, ,\\
    J(\theta) & ~=~ - \; \bbE_{\theta_0} \left[\nablathetasq \log f_\theta(X) \right], \quad X ^* \sim f_{\theta_0}\,,\\
    \text{and} \qquad I_{w} (\theta) & ~=~ c \, \lim_{N \to \infty} \frac{1}{N}  \,\sum_{l=1}^N \, \bbE_{\theta_0} \left[\frac{1}{\pi_l}  \Big( \nablatheta \log f_\theta (X_l^*) \Big) \Big(\nablatheta \log f_\theta (X_l^*) \Big)^{\!\!\top} \right].
\end{align*}
To study the weighted score and information functions, we require some smoothness conditions on the model.

\bigskip

\textbf{Regularity Conditions 3}
\begin{condition}
\label{c:model-5}
    All second partial derivatives of $\log f_\theta(x)$ with respect to the components of $\theta$ exist and are continuous for $\theta \in B$ at almost all $x$.
\end{condition}
\begin{condition}
\label{c:model-6}
    For $1 \leq j, k \leq K$, there exists functions $G_j(x)$, $G_{jk}(x)$ and $H_{jk}(x)$, possibly depending on $\theta_0$ and $B$, such that for all $\theta \in B$, the following relations hold for almost all x:
    \begin{align*}
        \text{(a) } & \left| \frac{\partial \log f_\theta(x)}{\partial \theta_j}\right| ~ \leq ~ G_j(x) \ \text{ with } \ \bbE_{\theta_0}[G_{j}(X^*)^2] < \infty\,,\\
        \text{(b) } & \left| \frac{\partial^2 \log f_\theta(x)}{\partial \theta_j \partial \theta_k}\right| ~ \leq ~ G_{jk}(x) \ \text{ with } \ \bbE_{\theta_0}[G_{jk}(X^*)^2] < \infty\,,\\
        \text{(c) } & \left| \frac{\partial \log f_\theta(x)}{\partial \theta_j} \cdot \frac{\partial \log f_\theta(x)}{\partial \theta_k} \right| ~ \leq ~ H_{jk}(x) \ \text{ with } \ \bbE_{\theta_0}[H_{jk}(X^*)^2] < \infty\,.
    \end{align*}
\end{condition}
\begin{condition}
\label{c:model-7}
    $I_w(\theta)$ and $J(\theta)$ are positive definite for $\theta \in B$, and all of their elements are finite.
\end{condition}

The set $B$ in conditions \ref{c:model-5} and \ref{c:model-6} is the same one of the conditions \ref{c:model-3} and \ref{c:model-4}. Note that the continuity of $J(\theta)$ in $B$ follows from conditions \ref{c:model-5} and \ref{c:model-6}(b). 
Using condition \ref{c:model-6}(b), we have $\bbE_{\theta_0}\big|\left[ \psi_i^\prime(\theta_0) \right]_{jk} \big|^2 <\infty $ for each $j, k$, and hence, by Lemma \ref{l:as_conv} with $q=1$, it follows directly that 
\begin{equation}
\label{eq:as_conv_J}
    J_{w,n}(\theta_0) ~\as~ J(\theta_0)\,.
\end{equation}
Conditions \ref{c:model-6}(c) ensure that $\bbE_{\theta_0}\big|\left[ \psi_i(\theta_0) \psi_i(\theta_0)^\top \right]_{jk} \big|^2 <\infty $ for each $j,k$, and using Lemma \ref{l:as_conv} with $q=2$, we have
\begin{equation}
\label{eq:as_conv_Iw}
    I_{w,n}(\theta_0) ~\as~ I_w(\theta_0)\,.
\end{equation}
We now proceed to prove some properties of the weighted score and information functions.

\begin{lemma}
\label{l:wtd_score}
Under conditions \ref{c:model-1}--\ref{c:model-4}, \ref{c:model-6}(a), \ref{c:survey-1}--\ref{c:survey-3}, and \ref{c:wt-1}, then as $n \to \infty$,
$$
\big\| \Tilde{S}_{w,n}(\theta_0) \big\| \cp 0 \ \asP\,.
$$
\end{lemma}
\begin{proof} We begin with the definition of the weighted score,
\begin{align*}
    \big\| \Tilde{S}_{w,n}(\theta_0) \big\|^2 
    & ~=~  \bigg\|\, \frac{1}{n} \sum_{i=1}^n g_{n,i} \, \psi_i(\theta_0)\, \bigg\|^2 ~=~  \sum_{k=1}^K \left( \frac{1}{n} \sum_{i=1}^n g_{n,i} \left[ \psi_i(\theta_0) \right]_k\right)^2
\end{align*}
It suffices to prove that for any $k$, $1 \leq k \leq K$,
\begin{equation}
\label{eq:wtdscore_conv}
    \frac{1}{n} \sum_{i=1}^n g_{n,i} \left[ \psi_i(\theta_0) \right]_k \cp 0 \quad \asP\,.
\end{equation}
By condition \ref{c:model-6}(a), $\bbE_{\theta_0}\big|\left[ \psi_i(\theta_0) \right]_k \big|^2 <\infty $ and $\bbE_{\theta_0} \left[ \psi_i(\theta_0) \right]_k = 0$. Thus, \eqref{eq:wtdscore_conv} follows by applying Lemma \ref{l:data-wts} with $h(X_i) = \left[ \psi_i(\theta_0) \right]_k$\,.
\end{proof}

\medskip
The next three lemmas state properties of the weighted information function.

\begin{lemma}
\label{l:wtd_info_1}
Under conditions \ref{c:model-1}--\ref{c:model-5}, \ref{c:model-6}(b), \ref{c:survey-1}--\ref{c:survey-3}, and \ref{c:wt-1}, as $n \to \infty$,
\begin{equation}
\label{eq:conv_wtd_info}
    \big\| \Tilde{J}_{w,n}(\theta) - J(\theta) \big\| ~\cp~ 0 \quad \asP\,.
\end{equation}
Moreover, if the sample information function corresponding to the powered likelihood $J_{w,n}(\theta)$ converges $\asP$ to $ J(\theta)$ uniformly in $\theta \in B$, then the convergence in \eqref{eq:conv_wtd_info} is also uniform in $\theta$ for $\theta \in B$.
\end{lemma}
\begin{proof}
By Lemma \ref{l:mat_conv}, it suffices to show that for all $u,v$,
\begin{equation}
\label{eq:conv_wtd_info_uv}
    P\left( \left| [\Tilde{J}_{w,n}(\theta)]_{uv} - [J(\theta)]_{uv} \right| < \epsilon \; \Big| \; \uX_n \right) \to 1 \quad \asP.
\end{equation}
The $(u,v)^{\text{th}}$ element of $\Tilde{J}_{w,n}(\theta)$ is $[\Tilde{J}_{w,n}(\theta)]_{uv} = -{n}^{-1}\sum_{i=1}^{n} g_{n,i} \left[ \psi_i'(\theta) \right]_{uv}$.
The corresponding element of $J(\theta)$ is: $\left[J(\theta) \right]_{uv} = - \bbE_{\theta_0} \left[ \frac{\partial^2}{\partial \theta_u \partial \theta_v}\log f_\theta(X) \right]$, where $X ^* \sim f_{\theta_0}$.
Condition \ref{c:model-6}(b) ensures that $\bbE \left| \left[  \psi_i'(\theta)  \right]_{uv}\right|^2 < \infty$. So applying Lemma \ref{l:data-wts} with $h(X_i) = \left[ \psi_i'(\theta) \right]_{uv}$, we have
\begin{equation}
\label{eq:conv_wtd_psi}
    \frac{1}{n} \sum_{i=1}^{n} g_{n,i} \left[ \psi_i'(\theta) \right]_{uv} ~\cp~ \left[J(\theta) \right]_{uv} \quad \asP,
\end{equation}
which completes proof of \eqref{eq:conv_wtd_info_uv}. The uniform convergence in \eqref{eq:conv_wtd_info} follows by proving uniform convergence in \eqref{eq:conv_wtd_psi}. To that end, let $s_i(\theta) := - [\psi_i^\prime(\theta)]_{uv}$  and $\epsilon>0$ be given. Applying Markov's inequality, we have
\begin{align*}
    P\bigg( \bigg| \frac{1}{n} \sum_{i=1}^{n} g_{n,i}\;s_i(\theta) - [J(\theta)]_{uv} \bigg| > \epsilon \mid \uX_n, \uw_n \bigg)
     ~\leq~ \frac{1}{\epsilon^2} \bbE \bigg( \bigg[ \frac{1}{n} \sum_{i=1}^{n} g_{n,i}\;s_i(\theta) - [J(\theta)]_{uv} \bigg]^2 ~\bigg|~ \uX_n, \uw_n \bigg)\\
     ~=~ \frac{1}{\epsilon^2} \ \bbE \bigg( \frac{1}{n} \sum_{i=1}^{n} g_{n,i} \; s_i(\theta) - [J_{w,n}(\theta)]_{uv} + [J_{w,n}(\theta)]_{uv} - [J(\theta)]_{uv} ~\bigg|~ \uX_n, \uw_n  \bigg)^2
\end{align*}
\vspace{-1cm}
\begin{align*}
    & ~=~ \frac{1}{n^2 \epsilon^2} \sum_{i=1}^{n} s_i(\theta)^2 \, \tilde{w}_{n,i}^2 \text{ Var}(Y_{i}) ~+~  \frac{1}{\epsilon^2} \big( [J_{w,n}(\theta)]_{uv} - [J(\theta)]_{uv} \big)^2\\
    & ~=~  \frac{1}{\epsilon^2 \, n^2} \sum_{i=1}^{n} \tilde{w}_{n,i}^2\ s_i(\theta)^2 ~+~ \frac{1}{\epsilon^2} \left( [J_{w,n}(\theta)]_{uv} - [J(\theta)]_{uv} \right)^2 
\end{align*}
The second term converges uniformly in $\theta$ to $0$ \ $\asP$ by assumption. The first term converges to $0$ \ $\asP$ by using condition \ref{c:model-6}(b) and observing that 
\begin{align*}
    \frac{1}{\epsilon^2 \, n^2} \sum_{i=1}^{n} \tilde{w}_{n,i}^2\ s_i(\theta)^2 
    & ~\leq~  \frac{1}{n^2} \sum_{i=1}^{n} \tilde{w}_{n,i}^2\, [G_{uv}(X_i)]^2 \\
    & ~\leq~ \left(\frac{1}{n} \max_{1 \leq i \leq n} \tilde{w}_{n,i} \, G_{uv}(X_i) \right) \left( \frac{1}{\epsilon^2\, n} \sum_{i=1}^{n} \tilde{w}_{n,i} \,G_{uv}(X_i) \right) \to 0 \quad \asP\,,
\end{align*}
which follows since the second factor remains finite by Lemma \ref{l:as_conv} (with $q=1$) and the first factor converges to $0$ \ $\asP$ by Lemma \ref{l:max} (with $q=1$).
\end{proof}

A final set of regularity conditions is needed for conditional consistency and conditional asymptotic normality.

\medskip

\textbf{Regularity Conditions 4}
\begin{condition}
\label{c:model-8}
    All third partial derivatives of $\log f_\theta(x)$ with respect to the components of $\theta$ exist and are continuous for $\theta \in B$ at almost all $x$.
\end{condition}
\begin{condition}
\label{c:model-9}
    For $1 \leq j,k,l \leq K$, there exists a function $G_{jkl}(x)$ such that for all $\theta \in B$,
    $$\left| \frac{\partial^3 \log f_\theta(x)}{\partial \theta_j \partial \theta_k \partial \theta_l}\right| ~ \leq ~ G_{jkl}(x) \ \text{ with } \ \bbE_{\theta_0}[G_{jkl}(X^*)^2] < \infty\,.$$
\end{condition}
\begin{condition}
\label{c:model-10}
    For $1 \leq j,k,l \leq K$, there exists a function $F_{jkl}(x)$ such that for all $\theta \in B$,
    $$\left| \frac{\partial \log f_\theta(x)}{\partial \theta_j} \cdot \frac{\partial^2 \log f_\theta(x)}{\partial \theta_k \partial \theta_l}\right| ~ \leq ~ F_{jkl}(x) \ \text{ with } \ \bbE_{\theta_0}[F_{jkl}(X^*)^2] < \infty\,.$$
\end{condition}

The following lemma is similar to Lemma \ref{l:wtd_info_1}, but differs in that the point at which the weighted likelihood is evaluated can vary with $n$.
    
\begin{lemma}
\label{l:wtd_info_2}
    If $\htheta$ is a strongly consistent estimator of $\theta_0$ and if conditions \ref{c:model-1}--\ref{c:model-9}, \ref{c:survey-1}--\ref{c:survey-3}, and \ref{c:wt-1} hold, then for all $\epsilon>0$, as $n \to \infty$,
    $$P\left( \big\| [\Tilde{J}_{w,n}(\htheta)] - [J(\theta_0)] \big\| > \epsilon \mid \uX_n, \uw_n \right) \to 0 \quad \asP.$$
\end{lemma}
\begin{proof}
Let $R = \Tilde{J}_{w,n}(\htheta) - J(\theta_0)$. By Lemma \ref{l:mat_conv}, it suffices to show that for each $1 \leq j, k \leq K$,
\begin{equation}
\label{eq:rem_wtd_info}
    P \big( \ |[R]_{jk}| > \epsilon \mid \uX_n, \uw_n \ \big) \to 0 \quad \asP\,.
\end{equation}
By definition, the components of $R$ have the following form
\begin{equation*}
    [R]_{jk} ~=~ - \frac{1}{n } \sum_{i=1}^n g_{n,i}\, [\psi_i^\prime (\htheta)]_{jk} ~+~ [J (\theta_0)]_{jk}
\end{equation*}
If we can show that as $n \to \infty$,
\begin{align}
    \label{eq:psi(i)}
    \frac{1}{n} \sum_{i=1}^{n} \tilde{w}_{n,i} \ \big|[\psi_i^\prime (\htheta)]_{jk} \big| & ~\to~ c, \quad 0<c<\infty\,, \\
    \label{eq:psi(ii)}
    \frac{1}{n} \sum_{i=1}^{n} \tilde{w}_{n,i} \ [\psi_i^\prime (\htheta)]_{jk}  & ~\to~ -[J(\theta_0)]_{jk}\,,\\
    \label{eq:psi(iii)}
    \frac{1}{n} \max_{1 \leq i \leq n} \tilde{w}_{n,i} \ \big|[\psi_i^\prime (\htheta)]_{jk} \big| & ~\to~ 0\,,
\end{align}
then \eqref{eq:rem_wtd_info} follows by applying Theorem \ref{t:prob_conv}. To prove them, consider the following Taylor expansion of $[\psi_i^\prime (\htheta)]_{jk}$ about $\theta_0$ evaluated at $\htheta$. Let $g_i(\theta) = [\psi_i^\prime (\theta)]_{jk}$.
\begin{equation*}
    g_i(\htheta) ~=~ g_i(\theta_0) ~+~ g_i^\prime(\htheta, \theta_0) \,(\htheta - \theta_0)
\end{equation*}
where $g_i^\prime(\htheta, \theta_0)$ denotes the vector of partial derivatives of $g_i(\theta)$ evaluated at some point $\theta_i^*$ on the line between $\htheta$ and $\theta_0$. The expansion is well defined by the differentiability conditions \ref{c:model-4}, \ref{c:model-5} and \ref{c:model-8}.

Considering the points $\omega_1 \in \Omega_1$ where $\htheta$ converges to $\theta_0$, we can find $N_1(\omega_1)$ such that for all $n > N_1(\omega_1)$, for all $1 \leq k \leq K$, $\big|[\htheta(\omega_1)]_k - [\theta_0]_k \big| < \delta$ for some $\delta > 0$, so that each third order partial derivative of the log-likelihood at $\theta^*$ can by bounded by a random variable with finite second moments (condition \ref{c:model-9}). To prove \eqref{eq:psi(i)}, note that for $n > N_1(\omega_1)$,
\begin{align*}
    \frac{1}{n} \sum_{i=1}^{n} \tilde{w}_{n,i} \, \big| g_i(\htheta) \big| 
    & ~=~ \frac{1}{n} \sum_{i=1}^{n} \tilde{w}_{n,i} \,\big|g_i(\theta_0) \big| ~+~ \frac{1}{n} \sum_{i=1}^{n} \tilde{w}_{n,i} \,\big| g_i^\prime(\htheta, \theta_0) \,(\htheta - \theta_0)\big|\\
    & ~\leq~ \frac{1}{n} \sum_{i=1}^{n} \tilde{w}_{n,i} \,\big| g_i(\theta_0) \big| ~+~ \frac{1}{n} \sum_{i=1}^{n} \tilde{w}_{n,i} \, \left(\sum_{l=1}^K G_{jkl}(X_i) \ \delta \right)\\
    & ~=~ \frac{1}{n} \sum_{i=1}^{n} \tilde{w}_{n,i} \,\big| g_i(\theta_0) \big| ~+~ \frac{\delta}{n} \sum_{i=1}^{n} \tilde{w}_{n,i}  \, \bar{G}_{jk}(X_i) 
\end{align*}
where $\bar{G}_{jk}(x) = \sum_{l=1}^K G_{jkl}(x)$. As $\bbE [\bar{G}_{jk}(X^*)^2] < \infty $ (by condition \ref{c:model-9}), $\bbE |g_i(\theta_0)|^2 < \infty$ (by condition \ref{c:model-6}(b)) and $\delta$ is arbitrarily small, \eqref{eq:psi(i)} follows by Lemma \ref{l:as_conv}.
\eqref{eq:psi(ii)} follows from Lemma \ref{l:as_conv} as well by condition \ref{c:model-6}(b) and by noting that for large enough $n$,
$$\left| \frac{1}{n} \sum_{i=1}^{n} \tilde{w}_{n,i} \, g_i (\htheta) ~-~ \frac{1}{n} \sum_{i=1}^{n} \tilde{w}_{n,i} \, g_i (\theta_0) \right| ~\leq~ \frac{\delta}{n} \sum_{i=1}^{n} \tilde{w}_{n,i} \,\bar{G}_{jk}(X_i)\,.$$
Finally, \eqref{eq:psi(iii)} follows by noting that 
$$
\frac{1}{n} \max_{1\leq i \leq n} \tilde{w}_{n,i} \, |g_i(\htheta)| ~\leq~ \frac{1}{n} \max_{1\leq i \leq n} \tilde{w}_{n,i} \, |g_i(\theta_0)| ~+~ \frac{\delta}{n} \max_{1\leq i \leq n} \tilde{w}_{n,i} \, F_{jk}(X_i)
$$
for large enough n, $\asP$. The maxima on the RHS converge to $0 \ \asP$, by Lemma \ref{l:max} with $q=1$.
\end{proof}

The next lemma gives conditions under which the weighted information matrix is positive definite with high conditional probability.

\begin{lemma}
\label{l:wtd_info_pd}
Let $ u \in \bbR^{K} $ be given with $ u \neq 0 $. If conditions \ref{c:model-1}--\ref{c:model-7}, \ref{c:survey-1}--\ref{c:survey-3}, and \ref{c:wt-1} hold, then uniformly in $ u $, as $ n \rightarrow \infty $,
\begin{equation}
\label{eq:wtd_info_pd_0}
    P \left( u^\top \Tilde{J}_{w,n}(\theta_0) u > 0 \mid \uX_n, \uw_n \right) \rightarrow 1 \quad \asP\,.
\end{equation}
If conditions \ref{c:model-8} and \ref{c:model-9} also hold, then for any strongly consistent estimator $ \htheta$,
\begin{equation}
\label{eq:wtd_info_pd_hat}
P \left( u^\top \Tilde{J}_{w,n}(\hat{\theta}_n) u > 0 \mid \uX_n, \uw_n \right) \rightarrow 1 \quad \asP\,,
\end{equation}
also uniformly in $ u $.
\end{lemma}

\begin{proof}
To prove \eqref{eq:wtd_info_pd_0}, let $\epsilon > 0$ and $u \,(\neq 0) \in \bbR^{K}$ be given.
Suppose that for almost every $\omega_1 \in \Omega_1$, we can find $N_1(\epsilon, \omega_1)$ not depending on $u$ such that for all $n > N_1(\epsilon, \omega_1)$,
$$
P \left( \frac{u^\top \Tilde{J}_{w,n}(\theta_0) u}{u^\top u} > 0 \,\bigg|\, \uX_n(\omega_1) \right) > 1 - \epsilon\,.
$$
If this holds, then \eqref{eq:wtd_info_pd_0} holds. So it suffices to prove \eqref{eq:wtd_info_pd_0} for vectors $u$ such that $u^\top u = 1$.
For $u$ having unit norm, we have
\begin{equation}
\label{eq:expand_uJu}
    u^\top \Tilde{J}_{w,n}(\theta_0) u ~=~ u^\top J(\theta_0) u ~+~ u^\top R u 
\end{equation}
where $R = \Tilde{J}_{w,n}(\theta_0) - J(\theta_0)$.
The proof follows by showing that the first term in \eqref{eq:expand_uJu} is bounded below by the smallest eigenvalue of $J(\theta_0)$ and the second term can be made arbitrarily small uniformly in $u$ with high conditional probability.

Consider the spectral decomposition of $J(\theta_0) = \Gamma \Lambda \Gamma^\top$, where $\Gamma$ is an orthogonal matrix and $\Lambda$ is the diagonal matrix of eigen values of $J(\theta_0)$. Since $J(\theta_0)$ is positive definite (by condition \ref{c:model-7}), its smallest eigen value $\lambda_1 > 0$. Now, observe that
\begin{align*}
    u^\top J(\theta_0) u 
    & ~=~ u^\top \big(\Gamma \Lambda \Gamma^\top \big) u  ~=~ s^\top \Lambda \, s, \qquad \qquad s = \Gamma^\top u, \quad s^\top s = 1\\
    & ~=~ \lambda_1 \sum_{k=1}^{K} s_k^{2} \ \frac{\lambda_k}{\lambda_1}  ~\geq~ \lambda_1 \sum_{k=1}^{K} s_k^{2} ~=~ \lambda_1 \,.
\end{align*}
Finally, we shall show that with conditional probability larger than $1-\epsilon$, $|u^\top R u| < \lambda_1$ for large $n$, uniformly in $u$. Note that
\begin{align*}
    |u^\top R u| & ~\leq~ \sum_{r} \sum_{s} |u_r||u_s||R_{rs}| ~\leq~ \sum_{r} \sum_{s} |R_{rs}| \\
    & ~\leq~ K^{2} \max_{1 \leq r, s \leq K} |R_{rs}| ~=~
    K^{2} M_R\,,
\end{align*}
where $M_R := \max_{1 \leq r, s \leq K} |R_{rs}| < \infty$. Then, 
$|u^\top R u| < \lambda_1$ if $M_R < \frac{\lambda_1}{K^{2}}\,,$
\ie if $|R_{rs}| < \lambda_1/K^2$ for all $1 \leq r, s \leq K$, for large enough $n$.
By Lemma \ref{l:wtd_info_1}, for each $r,s$, $|R_{rs}| \cp 0 \ \asP$, and hence we conclude the proof of \eqref{eq:wtd_info_pd_0}. The proof of \eqref{eq:wtd_info_pd_hat} follows by the exact same argument as \eqref{eq:wtd_info_pd_0} but with
$R ~=~ \Tilde{J}_{w,n} (\htheta) ~-~ J(\theta_0)$.
In this case as well, for each $r,s$, $|R_{rs}| \cp 0 \ \asP$ by Lemma \ref{l:wtd_info_2}.
\end{proof}

\subsection{Conditional consistency}
In the following, we shall prove Theorem \ref{t:cond_consistency} which is our main theorem on conditional consistency of a sequence of roots of the S-WLB estimating equation $\tilde{S}_{w,n}(\theta) = 0$. 

\bigskip

\noindent 
\begin{theorem-non}[{\bf Theorem}~\ref{t:cond_consistency} (Conditional consistency)]
Under Conditions~\ref{c:wt-1}, \ref{c:survey-1}--\ref{c:survey-3}, and \ref{c:model-1}--\ref{c:model-7}, there exists a conditionally consistent estimator $\ctheta$ of $\theta_0$ such that 
\begin{equation}
\label{eq:cond_prob_score}
    P \left(\tilde{S}_{w,n}(\ctheta) = 0 \mid \uX_n, \uw_n \right) \to 1 \quad \asP\,.
\end{equation}
Moreover this sequence is essentially unique in the following sense: If $\bar{\theta}_n$ satisfies \eqref{eq:cond_prob_score}, and is conditionally consistent, then
\begin{equation}
\label{eq:cond_consistency_uniqueness}
    P \left( \, \bar{\theta}_n = \ctheta \mid \uX_n, \uw_n \, \right) \to 1 \quad \asP\,.
\end{equation}
\end{theorem-non}

\begin{proof}
The proof closely follows that of Theorem 6 of \cite{newton1991thesis} and is included here in full for completeness. {\bf Key ideas} of the proof are the following:
\begin{enumerate}
    \item Find a neighborhood $U_\delta$ of $\theta_0$ such that $\tilde{S}_{w,n}(\theta)$ is one-to-one from $U_\delta$ onto $\tilde{S}_{w,n}(U_\delta)$.
    \item The image set $\tilde{S}_{w,n}(U_\delta)$ contains $0$ with high conditional probability. On this image set, $\tilde{S}^{-1}_n(\cdot)$ is well defined, so $\tilde{S}^{-1}_n(0) = \ctheta$ is a root of the weighted likelihood equation and is close to $\theta_0$.
\end{enumerate}  
To make this more precise, define $\alpha ~=~ \left(4 \, \|J(\theta_0)^{-1}\| \right)^{-1}\,$
using condition $\ref{c:model-7}$. By conditions \ref{c:model-5} and \ref{c:model-6}(b), it follows that $J(\theta)$ is continuous on $B$. Thus, there exists $\delta >  0$ sufficiently small such that 
\begin{equation*}
    \big\|J(\theta) - J(\theta_0) \big\| ~<~ \frac{\alpha}{3} \quad \text{whenever} \quad \theta \in U_\delta \subset B\,,
\end{equation*}
where $U_\delta = \{\theta : \|\theta - \theta_0\| < \delta\}$. Let $(n/N) (\omega_1) \subset \Omega_2$ be the set where 
\begin{equation*}
    \big\|\tilde{J}_{w,n}(\theta) - J(\theta) \big\| ~<~ \frac{\alpha}{3} \quad \text{for all } \theta \in U_\delta\,.
\end{equation*}
By Lemma \ref{l:wtd_info_1}, it follows that $P((n/N) \mid \uX_n, \uw_n) \to 1 \ \asP$. Next, let $D_n (\omega_1) \subset \Omega_2$ be the set where $\tilde{J}_{w,n}(\theta_0)$ is invertible. On $D_n (\omega_1)$, define $\alpha_n ~=~ \left( \|\tilde{J}_{w,n}(\theta_0)^{-1}\| \right)^{-1}$.
By Lemma \ref{l:wtd_info_pd}, we have $P(D_n \mid \uX_n, \uw_n) \to 1 \ \asP$. We need to show that $\alpha_n \cp \alpha \ \asP$. To that end, observe that Lemma \ref{l:wtd_info_1} implies $\tilde{J}_{w,n}(\theta_0) \cp J(\theta_0) \ \asP\,,$
and Lemma \ref{l:wtd_info_pd} implies that $\tilde{J}_{w,n}(\theta_0)$ is invertible with high conditional probability $\asP$. By continuity of matrix inversion, 
\vspace{-2ex}
\begin{align*}
    \tilde{J}_{w,n}(\theta_0)^{-1} &~\cp~ J(\theta_0)^{-1} \quad \asP\,,
\end{align*}
which implies that $\| \tilde{J}_{w,n}(\theta_0)^{-1}\| \cp \|J(\theta)^{-1}\| \ \asP$, and hence $\ \alpha_n \cp \alpha, \ \asP$. 

Further, let $E_n(\omega_1) \subset ((n/N) (\omega_1) \cap D_n(\omega_1))$ where $|\alpha_n - \alpha|< \alpha/2$. Noting that $P(E_n \mid \uX_n, \uw_n) \to 1 \ \asP$, we apply triangle inequality on the set $E_n$ for $\theta \in U_\delta$,
\begin{align*}
    \big\| \tilde{J}_{w,n}(\theta) - \tilde{J}_{w,n}(\theta_0) \big\| 
    & ~=~ \big\| \tilde{J}_{w,n}(\theta) - J(\theta) + J(\theta) - J(\theta_0) + J(\theta_0) - \tilde{J}_{w,n}(\theta_0) \big\| \\
    & ~\leq~ \big\| \tilde{J}_{w,n}(\theta) - J(\theta) \big\| ~+~ \big\|J(\theta) - J(\theta_0)\big\| ~+~ \big\|J(\theta_0) - \tilde{J}_{w,n}(\theta_0) \big\| \\
    & ~<~ \frac{\alpha}{3} ~+~ \frac{\alpha}{3} ~+~ \frac{\alpha}{3} ~=~ \alpha ~\leq~ 2\alpha_n
\end{align*}
Thus, the conditions for applying Theorem \ref{t:inverse_func} to $\tilde{S}_{w,n}(\theta)$ are satisfied on $E_n(\omega_1) \ \asP$. The weighted score function is continuously differentiable on $B$ and has a matrix of derivatives which is invertible at $\theta_0 \in B$. The Inverse Function Theorem implies that on $E_n(\omega_1), \ \asP$,
\begin{enumerate}
    \item [(a)] For every $\theta_1, \theta_2 \in U_\delta$,
    $\big\| \tilde{S}_{w,n} (\theta_1) - \tilde{S}_{w,n} (\theta_2) \big\| ~\geq~ 2\alpha_n \,\|\theta_1 - \theta_2\|$.
    \item [(b)] The image set
    $\tilde{S}_{w,n}(U_\delta) ~=~ \big\{s \in \bbR^K ~:~ \tilde{S}_{w,n}(\theta) = s, \ \theta \in U_\delta \big\}$
    contains the open ball of radius $\alpha_n\delta$ centered at $\tilde{S}_{w,n}(\theta_0)$.
\end{enumerate}
(a) ensures that on $E_n, \ \asP$, the weighted score $\tilde{S}_{w,n}(\theta)$ is a one-to-one function from $U_\delta$ onto the image set $\tilde{S}_{w,n}(U_\delta)$ and so the inverse function $\tilde{S}_{w,n}^{-1}(\cdot)$ mapping $\tilde{S}_{w,n}(U_\delta)$ onto $U_\delta$ is well-defined. Further, since $|\alpha_n - \alpha|< \alpha/2$ on $E_n(\omega_1)$, $\asP$, (b) implies the following,
\begin{enumerate}
    \item [($\text{b}^*$)] The image set $\tilde{S}_{w,n}(U_\delta)$
    contains the open ball of radius $\alpha\delta/2$ centered at $\tilde{S}_{w,n}(\theta_0)$.
\end{enumerate}
Let $E_n^\prime(\omega_1) \subset E_n(\omega_1)$ where $\big\| \tilde{S}_{w,n}(\theta_0) \big\| < \alpha \delta / 2$. By Lemma \ref{l:wtd_score}, $\tilde{S}_{w,n}(\theta_0) \cp 0 \ \asP$, which implies that 
$
P(E_n^\prime \mid \uX_n, \uw_n) \to 1 \ \asP.
$
On $E_n^\prime(\omega_1)$, $0 \in \tilde{S}_{w,n}(U_\delta) \ \asP$.
Let $N_1 \subset \Omega_1$ be the set of data sequences where all the previous convergences fail. Define
$$
\ctheta ~=~ 
\begin{cases} 
    \tilde{S}_{w,n}^{-1}(0) & \text{if } \omega_1 \in N_1^c \text{ and } \omega_2 \in E_n^\prime(\omega_1), \\
    \text{arbitrary} & \text{otherwise}.
\end{cases}
$$
Clearly, $\{\ctheta\}$ forms a sequence of roots of the weighted likelihood equation such that  $P(\tilde{S}_{w,n}(\ctheta) = 0 \mid \uX_n, \uw_n) \to 1 \ \asP$, 
which establishes \eqref{eq:cond_prob_score}. Also, since $\delta > 0$ is arbitrarily small, this sequence of roots is conditionally consistent,
$$
P \big( |\ctheta - \theta_0 | < \delta \mid \uX_n, \uw_n \big) ~\geq~ P(E_n^\prime \mid \uX_n, \uw_n) \to 1 \quad \asP.
$$
One-to-oneness of $\tilde{S}_{w,n}(\cdot)$ on $U_\delta$ 
establishes \eqref{eq:cond_consistency_uniqueness} and completes our proof.
\end{proof}

\subsection{Conditional asymptotic normality}

Before stating our result on conditional asymptotic normality, we first establish the following lemma, which plays a critical role in achieving our final result.

\begin{lemma}
\label{l:a_ni}
Let $z$ be a unit vector in $\bbR^K$, and let $a_{n,i}  = \sum_{k=1}^K z_k \, \big[\psi_i(\htheta) \big]_k = z^\top \psi_i(\htheta)$ for a strongly consistent estimator $\htheta$. Under conditions \ref{c:survey-1}--\ref{c:survey-4}, \ref{c:model-1}--\ref{c:model-8}, \ref{c:model-10}, and \ref{c:wt-1},
\begin{align}
    \label{eq:a_sq_mean}
    \frac{1}{n} \sum_{i=1}^n \tilde{w}_{n,i}^2\ a_{n,i}^2 &~\to~ z^\top I_w(\theta_0)\,z \quad \asP\,,\\
    \label{eq:a_sq_max}
    \frac{1}{n} \max_{1\leq i \leq n} \tilde{w}_{n,i}^2\ a_{n,i}^2 & ~\to~ 0 \quad \asP\,.
\end{align}
\end{lemma}

\begin{proof}
Define $\tilde{a}_{n,i} = \tilde{w}_{n,i}\,a_{n,i}$ and $h_i(\theta) = \{\tilde{w}_{n,i} \, z^\top \psi_i(\theta)\}^2$. Note that $\tilde{a}_{n,i}^2 = h_i(\htheta)$. By differentiability conditions \ref{c:model-4}, \ref{c:model-5}, and \ref{c:model-8}, we consider a second order Taylor series expansion of each $h_i(\theta)$ about $\theta_0$ evaluated at $\htheta$,
\begin{equation}
\label{eq:a_ni_taylor}
    \tilde{a}_{n,i}^2 = h_i(\theta_0) + [h_i^\prime(\theta_0, \htheta)]^\top(\htheta - \theta_0)
\end{equation}
where $h_i^\prime(\htheta, \theta_0)$ denotes the vector of partial derivatives of $h_i(\theta)$ evaluated at some point $\theta_i^*$ on the line between $\htheta$ and $\theta_0$. By taking an average over $i = 1, 2, \ldots, n$,
\begin{equation}
\label{eq:a_ni_taylor_avg}
    \frac{1}{n}\sum_{i=1}^n \tilde{a}_{n,i}^2 ~=~ \frac{1}{n}\sum_{i=1}^n h_i(\theta_0) + R_{w,n}\,, 
\end{equation}
where $R_{w,n}$ is the average of the remainder terms from \eqref{eq:a_ni_taylor}. Writing the first term of \eqref{eq:a_ni_taylor_avg} explicitly, 
\begin{align*}
    \frac{1}{n}\sum_{i=1}^n h_i(\theta_0) & ~=~ \frac{1}{n} \sum_{i=1}^n \tilde{w}_{n,i}^2 \ z^\top \psi_i(\theta_0)  \psi_i(\theta_0)^\top z
    ~=~ z^\top I_{w,n}(\theta_0)\,z.
\end{align*}
From \eqref{eq:as_conv_Iw}, we have
\begin{align}
    \frac{1}{n}\sum_{i=1}^n h_i(\theta_0) ~\to~ z^\top I_w(\theta_0) \,z \quad \asP\,.
\end{align}
Thus, it suffices to prove that $R_{w,n} ~\to~ 0 \ \asP$ to get \eqref{eq:a_sq_mean}. To that end, note that 
\begin{align}
\label{eq:h_i_prime}
    [h_i^\prime(\theta_0, \htheta)]_u & ~=~ 2 \sum_{j=1}^K \sum_{k=1}^K \tilde{w}_{n,i}^2 \, z_j z_k \, \frac{\partial\log f_\theta (X_i)}{\partial \theta_j} \, \frac{\partial^2 \log f_\theta (X_i)}{\partial \theta_u\, \partial \theta_k}\,,
\end{align}
where all derivatives are evaluated at $\theta_i^*$. Suppose that $n$ is large enough such that for all $1 \leq k \leq K$, $\big|[\htheta]_k - [\theta_0]_k \big| < \delta$ for some $\delta > 0$. Further, $\delta$ is small enough so that the bounds on the products of the partials given in  condition \ref{c:model-10} hold for all $1 \leq u,j, k \leq K$. Then,
\begin{align*}
    |R_{w,n}| & ~\leq~ \frac{\delta}{n}\ \sum_{i=1}^n \sum_{u=1}^K \tilde{w}_{n,i}^2 \ \big| [h_i^\prime(\theta_0, \htheta)]_u \big| ~\leq~ 2 \,\frac{\delta}{n}\sum_{i=1}^n \tilde{w}_{n,i}^2 \ \bar{F}(X_i)\,,
\end{align*}
where $\bar{F}(X) = \sum_{u=1}^K \sum_{j=1}^K \sum_{k=1}^K F_{ujk}(X)$ and $F_{ujk}$ are as defined in condition \ref{c:model-10}. Condition \ref{c:model-10} implies $\bbE_{\theta_0} [\bar{F}(X^*)]^2 < \infty$ and hence by Lemma \ref{l:as_conv} with $q=2$, the weighted average converges to a finite limit. Since $\delta > 0$ is arbitrarily small, $R_{w,n} ~\to~ 0 \ \asP$. 

Next, to prove \eqref{eq:a_sq_max}, we shall use \eqref{eq:a_ni_taylor} and \eqref{eq:h_i_prime} to arrive at the following,
\begin{align*}
    \frac{1}{n} \max_{1 \leq i \leq n} \tilde{a}_{n,i}^2 ~\leq~ \frac{1}{n} \max_{1 \leq i \leq n}h_i(\theta_0) ~+~ \frac{2\delta}{n} \max_{1 \leq i \leq n} \tilde{w}_{n,i}^2\bar{G}\,(X_i)\,.
\end{align*}
\eqref{eq:a_sq_max} follows directly by applying Lemma \ref{l:max} with $q=2$ on the first and second terms of the RHS.
\end{proof}

Finally, we shall prove Theorem \ref{t:cond_AN} (of the main document) which establishes the asymptotic distribution for a conditionally consistent sequence of roots of the weighted likelihood equation.

\bigskip
\noindent \begin{theorem-non}[{\bf Theorem}~\ref{t:cond_AN} (Conditional asymptotic normality)]
Under Conditions~\ref{c:wt-1}, \ref{c:survey-1}--\ref{c:survey-4}, and \ref{c:model-1}--\ref{c:model-10}, let $\{\htheta\}$ be a strongly consistent sequence of estimators of $\theta_0$ satisfying
\begin{equation}
\label{eq:cond_consistency}
    \left\| \sqrt{n}\, S_{w,n}(\hat{\theta}_n) \right\| ~\to~ 0 \quad \asP\,.
\end{equation}
If $\{\ctheta\}$ is a conditionally consistent sequence of roots of the weighted likelihood equation, then for any Borel set $A \subset \mathbb{R}^K$, as $n \to \infty$,
\begin{equation}
\label{eq:cond_AN}
    P \left( \sqrt{n} (\ctheta - \htheta) \in A \mid \uX_n, \uw_n \right) ~\to~ P(Z \in A) \quad \asP
\end{equation}
where $Z \sim \cN_K \left(0, J(\theta_0)^{-1} I_w(\theta_0) J(\theta_0)^{-1} \right)$.
\end{theorem-non}

\bigskip

Note that condition in \eqref{eq:cond_consistency} is satisfied if $\hat{\theta}_n$ is the pseudo maximum likelihood estimator (PMLE) obtained by solving the powered likelihood equation, $S_{w,n}(\theta) = 0$. 

\bigskip
\begin{proof}
Under differentiability conditions \ref{c:model-4}, \ref{c:model-5}, and \ref{c:model-8}, the weighted score function can be written as a second-order Taylor expansion, evaluated at $\ctheta$ and centered at $\htheta$,
\begin{equation}
\label{eq:taylor_score}
    \tilde{S}_{w,n}(\ctheta) = \tilde{S}_{w,n}(\htheta) - \tilde{J}_{w,n}(\htheta) (\ctheta - \htheta) + R_{w,n}(\ctheta - \htheta)
\end{equation}
where $R_{w,n}$ is a $K \times K$ matrix, whose $j$-th row is
\begin{align*}
    [R_{w,n}]_{j\cdot} & 
    ~=~ \frac{1}{2 n} \,\sum_{i=1}^{n} \, g_{n,i}\, (\ctheta - \htheta)^\top \,\Psi_i^j(\theta^*)
\end{align*}
where $\theta^*$ is a point on the line between $\htheta$ and $\ctheta$ and $\Psi_i^j(\theta)$ is the matrix of third order partial derivatives of $\log f_\theta(X_i)$, \ie
$$
[\Psi_i^j(\theta)]_{lk} = \frac{\partial^3}{\partial \theta_i \partial \theta_j \partial \theta_k} \log f_{\theta}(X_i)\,.
$$
Since $\tilde{S}_{w,n}(\ctheta) = 0$, \eqref{eq:taylor_score} can be re-written as
\begin{equation}
\label{eq:taylor_score_2}
    \tilde{S}_{w,n}(\htheta) =  \big(\tilde{J}_{w,n}(\htheta) - R_{w,n} \big) \big(\ctheta - \htheta \big)
\end{equation}

The proof can be broken down in {\bf two} parts. The {\bf first part} amounts to showing that the matrix $\tilde{J}_{w,n}(\htheta) - R_{w,n}$ converges in some sense to $J(\theta_0)$ and is invertible with high conditional probability. In the {\bf second part}, we will show that 
\begin{equation}
\label{eq:score_dist}
    \sqrt{n}\, \tilde{S}_{w,n} (\htheta) ~\cd~ \cN_K \big(0, I_w(\theta_0) \big)\,,
\end{equation}
where the notation $\cd$ means convergence in conditional distribution given ($\uX_n, \uw_n$). Then, the proof of \eqref{eq:cond_AN} follows by applying Slutsky's Theorem.

Let $B_n(w_1) \subset \Omega_2$ be the set where $\tilde{J}_{w,n}(\hat{\theta}_n) - R_{w,n}$ is invertible.
We shall show that
\begin{equation}
\label{eq:B_n_hcp}
    P \left( B_n \mid \uX_n, \uw_n \right) \to 1 \quad \asP
\end{equation}
To prove \eqref{eq:B_n_hcp}, it suffices to show that
$$
P \left( u^T \big( \tilde{J}_{w,n}(\hat{\theta}_n) - R_{w,n} \big) u > 0 \mid \uX_n, \uw_n \right) \to 1 \quad \asP
$$
uniformly for $u \in \mathbb{R}^K$ with $\|u\| = 1$.
By Lemma \ref{l:wtd_info_pd}, $\tilde{J}_{w,n}(\hat{\theta}_n)$ is positive definite with high conditional probability. Therefore, it suffices to show that $| u^\top R_{w,n} u |$ can be made arbitrarily small uniformly in $u$.
Note that 
\begin{align*}
    | u^\top R_{w,n} u | & ~=~ \sum_{j=1}^{K} \sum_{k=1}^{K} [u]_j [u]_k [R_{w,n}]_{jk} ~\leq~ K^2 \max_{j,k} \left| [R_{w,n}]_{jk} \right|.
\end{align*}
Now, by construction,  
\begin{align*}
    [R_{w,n}]_{jk} 
    & ~=~ \frac{1}{2 n} \sum_{i=1}^n g_{n,i} \, \bigg( \sum_{l=1}^K \big[\ctheta - \htheta \big]_{l} \big[\Psi_i^j(\theta^*) \big]_{lk} \bigg)
\end{align*}
Let $\delta > 0 $ be given. Let $\epsilon > 0 $ be such that the bounds on the third order partial derivatives hold for all $\|\theta - \theta_0\| < \epsilon$ (condition \ref{c:model-9}). Since $\htheta$ is strongly consistent for $\theta_0$, therefore $\asP$, 
$$
\big\|\htheta - \theta_0 \big\| ~<~ \frac{\epsilon}{2} \quad \text{for all } n > N_1(\omega_1, \epsilon)\,.$$
Let $(n/N)(\omega_1)\subset \Omega_2$ be the set where $\big\| \ctheta - \theta_0\big\| < \epsilon/2$. Since $\ctheta$ is conditionally consistent for $\theta_0$, therefore for $n > N_1(\omega_1,\epsilon, \delta)$, $P((n/N)\mid \uX_n, \uw_n) > 1- \delta/2, \ \asP$. Thus, on the set $(n/N)(\omega_1)$,
\begin{align*}
    \big\| \ctheta - \htheta \big\|
    & ~=~ \big\| \ctheta - \theta_0 + \theta_0 -\htheta \big\| \\
    & ~\leq~ \big\| \ctheta - \theta_0 \big\| ~+~ \big\|\theta_0 -\htheta \big\|  ~<~ \frac{\epsilon}{2} ~+~ \frac{\epsilon}{2} ~=~ \epsilon\,.
\end{align*}
Moreover, $\theta^*$ is within $\epsilon/2$ of $\theta_0$ since all points between $\htheta$ and $\ctheta$ are within $\epsilon/2$ of $\theta_0$. Therefore by condition \ref{c:model-9}, on $(n/N)(\omega_1)$, $\asP$, we have
\begin{align}
    \label{eq:R_{w,n}}
    [R_{w,n}]_{jk} & ~\leq~ \frac{\epsilon}{2 n} \sum_{i=1}^n g_{n,i} \,  \sum_{l=1}^K G_{jkl}(X_i)
    ~=~\frac{\epsilon}{2 n} \sum_{i=1}^n g_{n,i} \,  \bar{G}_{jk}(X_i)\,.
\end{align}
which follows since $\big| [\Psi_i^j(\theta^*)]_{lk}\big| \leq G_{jkl}(X_i)$ where $\bbE_{\theta_0}|G_{jkl}(X^*)|^2 < \infty$, and $\bar{G}_{jk}(x) = \sum_{l} G_{jkl}(x)$. By Lemma \ref{l:data-wts}, 
$$\frac{1}{n} \sum_{i=1}^n g_{n,i}\, \bar{G}_{jk}(X_i) ~\cp~ \xi = \bbE_{\theta_0} [\bar{G}_{jk}(X)] \quad \asP\,.$$
Thus, 
$({2n})^{-1} \sum_{i=1}^n g_{n,i}\, \bar{G}_{jk}(X_i) ~\cp~ {\xi}/{2} \ \asP\,.$
Therefore, on a set $D_n(\omega_1)$ having conditional probability larger than $1 - \delta/2$, we have 
\begin{equation}
\label{eq:R_{w,n}_2}
    \left|\frac{1}{2n} \sum_{i=1}^n g_{n,i}\, \bar{G}_{jk}(X_i) ~-~ \frac{\xi}{2} \right| ~<~ \epsilon
\end{equation}
Combining \eqref{eq:R_{w,n}} and \eqref{eq:R_{w,n}_2}, on the set $(n/N)(\omega_1) \cap D_n(\omega_1)$ having conditional probability larger than $1-\delta$,
$$
\big| [R_{w,n}]_{jk} \big| ~\leq~ \epsilon \bigg(\epsilon + \frac{\xi}{2} \bigg) \quad \text{ for } n > N_1(\omega_1, \epsilon, \delta)\,. 
$$
Since $\epsilon>0$ is arbitrary, therefore the proof of \eqref{eq:B_n_hcp} follows by noting that 
$[R_{w,n}]_{jk} \cp 0 \ \asP$. Having proved that, it now follows from Lemma \ref{l:wtd_info_2} that 
\begin{equation}
\label{eq:part1_CAN}
    \big(\tilde{J}_{w,n}(\htheta) - R_{w,n} \big) ~\cp~ J(\theta_0) \quad \asP\,.
\end{equation}
Define 
$$
\tilde{\Sigma}_n = 
\begin{cases} 
    \big(\tilde{J}_{w,n}(\htheta) - R_{w,n} \big)^{-1} & \text{ on the set } B_n(\omega_1), \\
    \text{arbitrary} & \text{ on } B_n(\omega_1)^c.
\end{cases}
$$
By continuity of matrix inversion and using the Continuous Mapping Theorem, \eqref{eq:part1_CAN} implies
$\tilde{\Sigma}_n ~\cp~ J(\theta_0)^{-1} \ \asP$.
This completes the first part of the proof. 

For proving the second part, we rewrite \eqref{eq:taylor_score_2} as 
\begin{equation*}
    \sqrt{n}\,(\ctheta - \htheta) ~=~ \tilde{\Sigma}_n^{-1} \ \sqrt{n}\, \tilde{S}_{w,n}(\htheta)\,
\end{equation*}
on the set $B_n(\omega_1)$ having conditional probability converging to $1 \ \asP$. To prove \eqref{eq:cond_AN}, it suffices to prove \eqref{eq:score_dist}, followed by applying Slutsky's theorem conditionally along $P_1-$ almost every sequence of sampled data and sampling weights. We shall prove \eqref{eq:score_dist} by using Cramer-Wold device and showing that for any $z\in \bbR^K$ with $\|z\|=1$,
\begin{equation}
\label{eq:cond_AN_CMW}
    t_n(z) ~\equiv~ \sqrt{n}\,z^\top\tilde{S}_{w,n}(\htheta) ~\cd~ \cN_K \big(0, z^\top I_w(\theta_0)\,z \big) \quad \asP\,,
\end{equation}
where the notation $\cd$ denotes convergence in conditional distribution given ($\uX_n, \uw_n$).
Let $g_i^k(\htheta) := \big[\psi_i(\htheta) \big]_k$ and write $t_n(z)$ explicitly,
\begin{align*}
    t_n(z) & ~=~ \sqrt{n} \, \sum_{k=1}^K z_k \bigg( \frac{1}{n} \sum_{i=1}^n g_{n,i} \,\big[\psi_i^\prime(\htheta) \big]_k \bigg) 
    ~=~ \frac{1}{\sqrt{n}} \,  \sum_{i=1}^n Y_i \, \tilde{w}_{n,i} \,  a_{n,i}\,,
\end{align*}
where $a_{n,i}  = \sum_{k=1}^K z_k \, g_i^k(\htheta) = z^\top \psi_i(\htheta)$. 
Using $\|z\| = 1$ and \eqref{eq:cond_consistency}, we have
\begin{align}
    \nonumber
    \bigg\| \frac{1}{\sqrt{n}} \sum_{i=1}^n a_{n,i} \ \tilde{w}_{n,i}\bigg\|
    \nonumber
    & ~=~ \bigg\| \sqrt{n} z^\top \bigg(\frac{1}{n} \sum_{i=1}^n \tilde{w}_{n,i}  \psi_i(\htheta) \bigg) \bigg\|\\
    \label{eq:mean_AN}
    & ~\leq~ \|z\| \ \big\| \sqrt{n} \,S_{w,n}(\htheta) \big\| ~\to~ 0 \quad \asP\,.
\end{align}
Using Lemma \ref{l:a_ni}, it follows that
\begin{align}
\label{eq:var_AN}
    \frac{1}{n} \sum_{i=1}^n \tilde{w}_{n,i}^2\ a_{n,i}^2 & ~\to~ z^\top I_w(\theta_0)z \quad \asP\,,\\
    \label{eq:var_AN_2}
    \frac{1}{n} \max_{1\leq i \leq n} \tilde{w}_{n,i}^2\ a_{n,i}^2 & ~\to~ 0 \quad \asP\,.
\end{align}
Combining \eqref{eq:mean_AN}, \eqref{eq:var_AN} and \eqref{eq:var_AN_2}, the proof of \eqref{eq:cond_AN_CMW} follows by applying Theorem \ref{t:clt} with $b_{n,i} = \tilde{w}_{n,i}\,a_{n,i}$.
\end{proof}

\section{Asymptotic results for the S-PB}
\label{s:asymptotics_spb}

\subsection{Framework and preliminaries}

To begin with, let us define the prior-penalized weighted log-likelihood objective as
\[
    \tilde Q_{w,n}(\theta)
    ~=~
    \log \tilde L_{w,n}(\theta)
    ~+~
    \ug_0^\top \log \pi(\theta),
\]
where $\tilde L_{w,n}(\theta)$ denotes the randomly weighted
pseudo-likelihood defined in (\ref{eq:swlb}) in Section~\ref{s:S-WLB}. 
Along the lines of \eqref{eq:partials}, we further define the tensor of third-order partial derivatives (when they exist) of $\log f_\theta(X_i)$ as $\psi_i^{\prime\prime}(\theta) ~=~ \nablathetac \log f_\theta(X_i)  \ \in \ \bbR^{K \times K \times K}.$
For $X_l^*\sim f_{\theta_0}$, define
\[
    \psi_l(\theta)
    =
    \nablatheta\log f_\theta(X_l^*),
    \qquad
    \psi_l'(\theta)
    =
    \nablathetasq\log f_\theta(X_l^*),
    \qquad
    \psi_l''(\theta)
    =
    \nablathetac\log f_\theta(X_l^*).
\]
For a generic random variable $X^*\sim f_{\theta_0}$, we similarly write
\[
    \psi(\theta)
    =
    \nablatheta\log f_\theta(X^*),
    \qquad
    \psi'(\theta)
    =
    \nablathetasq \log f_\theta(X^*),
    \qquad
    \psi''(\theta)
    =
    \nablathetac \log f_\theta(X^*).
\]
Following \cite{pompe2021}, we further define the following
\begin{align*}
    \tilde{S}_{1,w,n} (\theta) &~=~ \frac{1}{\sqrt{n}} \sum_{i=1}^n (Y_i-1) \,\tilde{w}_{n,i} \,\psi_i(\theta),&
    \tilde{S}_{2,w,n}(\theta)  &~=~ \frac{1}{\sqrt{n}} \sum_{i=1}^n (Y_i-1) \,\tilde{w}_{n,i} \,\big\{-\psi_i^\prime(\theta)\big\},\\
    \tilde{\mu}_{3,w,n}(\theta)  & ~=~ \frac{1}{n} \sum_{i=1}^n Y_i \,\tilde{w}_{n,i} \,\psi_i''(\theta), &
    B_{1,2,w,n}(\theta) & = \frac{1}{n} \sum_{i=1}^n  \,\tilde{w}_{n,i}^2 \,\psi_i^\prime(\theta) \,J_{w,n}(\theta)^{-1} \,\psi_i(\theta),\\
    {\mu}_{3}(\theta) & = \bbE_{\theta_0} \big[ \psi''(\theta)\big], &
    B_{1,2,w}(\theta) & = c \lim_{N \to \infty} \frac{1}{N}\sum_{l=1}^N \bbE_{\theta_0}   \Big[\pi_l^{-1} \, \psi_l' (\theta)\, J(\theta)^{-1}\, \psi_l(\theta) \Big].
\end{align*}

Before deriving the higher-order asymptotic expansion for the S-PB, we first establish its conditional consistency. Since the prior contributes an additional term of order $n^{-1}$ to the normalized S-WLB estimating equation $\tilde{S}_{w,n}(\theta) = 0$, the S-PB can be viewed as a small perturbation of the S-WLB. The following lemma establishes the existence of a conditionally consistent sequence of roots of the S-PB estimating equation $\nablatheta \tilde{Q}_{w,n}(\theta) = 0$.

\begin{lemma}[Conditional consistency of the S-PB estimator]
\label{l:spb_consistency}
Suppose Conditions~\ref{c:survey-1}--\ref{c:survey-3}, \ref{c:model-1}--\ref{c:model-7}, \ref{c:prior}, and \ref{c:wt-1} hold. Then there exists a sequence
$\{\tilde\theta_n\}$ of roots of the S-PB estimating equation $\nabla_\theta\,
    \tilde Q_{w,n}(\tilde\theta_n)=0,$
such that, for every $\epsilon>0$,
\[
    P\!\left(
        \|\tilde\theta_n-\theta_0\|>\epsilon
        \,\middle|\, \uX_n,\uw_n
    \right)
    \longrightarrow 0,
    \quad \asP.
\]
Consequently, if $\hat\theta_n$ is a strongly consistent solution of
the powered-likelihood estimating equation $S_{w,n} (\htheta)=0$, then $\tilde\theta_n-\hat\theta_n ~\cp~ 0$, $\asP.$
\end{lemma}

\begin{proof}
Let $\ctheta$ denote the conditionally consistent sequence of roots
of the S-WLB estimating equation
$\tilde S_{w,n}(\ctheta)=0$, whose existence is established in
Theorem~\ref{t:cond_consistency}. Thus, $\ctheta ~\cp~ \theta_0$, $\asP$.
The first-order condition corresponding to the S-PB objective is
\[
    \tilde S_{w,n}(\theta)
    ~+~
    \frac{1}{n}\,
    \ug_0\cdot\nablatheta\log\pi(\theta)
    ~=~0.
\]
By Condition P1, $\log\pi(\theta)$ is sufficiently smooth in a
neighborhood of $\theta_0$. Hence, for a sufficiently small closed
neighborhood $B_0\subset B$ of $\theta_0$, both
$\nablatheta\log\pi(\theta)$ and its derivative are bounded on $B_0$.
In particular,
\[
    \sup_{\theta\in B_0}
    \left\|
        \frac{1}{n}\,
        \ug_0\cdot\nablatheta\log\pi(\theta)
    \right\|
    ~=~O(n^{-1}).
\]
Next, define
\[
    \Psi_n(\theta) ~=~ \nablatheta \tilde{Q}_{w,n} (\theta)
    ~=~
    \tilde S_{w,n}(\theta)
    +
    \frac{1}{n}\,
    \ug_0\cdot\nablatheta\log\pi(\theta).
\]
Since $\ctheta~\cp~\theta_0$, we have $\ctheta\in B_0$ with conditional
probability tending to one, and hence
\[
    \Psi_n(\ctheta)
    ~=~
    \frac{1}{n}\,
    \ug_0\cdot\nablatheta\log\pi(\ctheta)
    ~=~O_P(n^{-1}),
\]
conditionally on $(\uX_n,\uw_n)$.
Moreover,
\[
    \nablatheta\Psi_n(\theta)
    ~=~
    -\tilde J_{w,n}(\theta)
    ~+~
    \frac{1}{n}\,
        \ug_0\cdot\nabla_\theta^2\ \log\pi(\theta).
\]
Under Conditions~\ref{c:survey-1}--\ref{c:survey-3} and M1--M7, the arguments used in the proof of
Theorem~\ref{t:cond_consistency} imply that $\tilde J_{w,n}(\theta)$ is
non-singular in a sufficiently small neighborhood of $\theta_0$ with
conditional probability tending to one, $\asP$. Since the second term
is uniformly $O(n^{-1})$ on $B_0$, $\nablatheta\Psi_n(\theta)$ remains
non-singular in this neighborhood with conditional probability tending
to one, $\asP$.
Since $\Psi_n$ is locally invertible in a neighborhood of $\ctheta$
and $\Psi_n(\ctheta)=O_P(n^{-1})$ conditionally on $(\uX_n,\uw_n)$,
the corresponding root $\tilde\theta_n$ satisfies $\tilde\theta_n-\ctheta~\cp~0,$ $\asP$.
Since $\ctheta~\cp~\theta_0$, it follows that $ \tilde\theta_n~\cp~\theta_0$, $\asP$.
Finally, since $\htheta\to\theta_0$ almost surely under $P_*$, $\|\tilde\theta_n-\htheta\|
    ~\leq~
    \|\tilde\theta_n-\theta_0\|
    +
    \|\htheta-\theta_0\|$,
and hence $\tilde\theta_n-\htheta~\cp~0$, $\asP$.
\end{proof}

\bigskip
Having established the conditional consistency of the S-PB estimator, we now introduce additional regularity conditions on the model likelihood needed for the higher-order expansion in Theorem~\ref{t:edgeworth}. These conditions complement Conditions~\ref{c:model-1}--\ref{c:model-10}, previously introduced for the S-WLB in Section~\ref{s:asymptotics_swlb}, and are motivated by the higher-order regularity assumptions of \citet{pompe2021}. Condition~\ref{c:model-11} strengthens the moment requirements on the first- and second-order derivatives of the log-likelihood in Condition~\ref{c:model-6}, while Condition~\ref{c:model-12} imposes a linear independence condition needed for the Edgeworth expansion.

\begin{condition}
\label{c:model-11}
The envelope functions in Condition~\ref{c:model-6}(a) and
\ref{c:model-6}(b) satisfy
\[
    E_{\theta_0}
    \left[
        G_j(X^*)^8
    \right]
    <\infty,
    \qquad
    E_{\theta_0}
    \left[
        G_{jk}(X^*)^8
    \right]
    <\infty, \qquad \text{for all } 1\leq j,k\leq K.
\]
\end{condition}

\begin{condition}
\label{c:model-12}
The collection of functions
\[
    \left\{
        \frac{\partial}{\partial\theta_j}
        \log f_{\theta}(x),
        \;
        \frac{\partial^2}{\partial\theta_k\partial\theta_l}
        \log f_{\theta}(x)
        :
        1\leq j\leq K,\;
        1\leq k\leq l\leq K
    \right\}
\]
is linearly independent as functions of $x$, a.s.\,$P_{\theta_0}$.
\end{condition}

\subsection{Edgeworth expansion}

In this section, we establish the higher-order Edgeworth expansion for the conditional distribution of the S-PB estimator. We first derive the following lemma concerning the joint Edgeworth expansion of $(\tilde S_{1,w,n}(\htheta), \tilde S_{2,w,n}(\htheta))$ that will be used in the proof of Theorem~\ref{t:edgeworth}.

\begin{lemma}
\label{l:joint_edgeworth}
Let $\{\htheta\}$ be a strongly consistent sequence of roots of the powered likelihood equation $S_{w,n}(\theta) = 0$. Under Conditions~\ref{c:wt-1}--\ref{c:wt-2},
\ref{c:survey-1}--\ref{c:survey-4}, and
\ref{c:model-1}--\ref{c:model-12}, the joint distribution of $(\tilde S_{1,w,n}(\htheta), \tilde S_{2,w,n}(\htheta))$
admits a valid Edgeworth expansion,
conditionally on $(\uX_n,\uw_n)$, $\asP$.
\end{lemma}

\begin{proof}
Define $Z_i = Y_i-1$, and $\ua_{w,i} = \tilde{w}_{n,i}\, \ub_i (\htheta)$, where $\ub_i(\theta) \in\mathbb R^{K+K(K+1)/2}$ denotes the vector collecting the first and second-order derivatives of the log-likelihood,
\[
b_i(\theta)
=
\left\{
    \big(\psi_i(\theta)\big)_j,
    \;
    \big(-\psi_i'(\theta)\big)_{k,l}
    :
    1\leq j\leq K,\;
    1\leq k\leq l\leq K
\right\}.
\]
Then the vector $n^{-1/2}\sum_{i=1}^n Z_i\, \ua_{w,i}$
contains $\tilde S_{1,w,n}(\htheta)$ and $\tilde S_{2,w,n}(\htheta)$ as
its corresponding blocks. It therefore suffices to establish a valid Edgeworth expansion for this vector. Let $d=K+K(K+1)/2$. For $X_l^*\sim f_{\theta_0}$, define
$b_l(\theta)\in\mathbb{R}^d$ analogously to $b_i(\theta)$. We further
define
\[
A_{2,w}
=
c\lim_{N\to\infty}\frac{1}{N}\sum_{l=1}^N
E_{\theta_0}\left[
\pi_l^{-1}b_l(\theta_0)^{\otimes 2}
\right],
\qquad
A_{3,w}
=
c^2\lim_{N\to\infty}\frac{1}{N}\sum_{l=1}^N
E_{\theta_0}\left[
\pi_l^{-2}b_l(\theta_0)^{\otimes 3}
\right],
\]
\vspace{-2ex}
\[
a_{4,w}
=
c^3\lim_{N\to\infty}\frac{1}{N}\sum_{l=1}^N
E_{\theta_0}\left[
\pi_l^{-3}\|b_l(\theta_0)\|^4
\right].
\]
To that end, we shall verify Conditions (a)--(c) of Lemma~\ref{l:edgeworth}. We define
\[
    \bar a_{2,w,n}
    ~=~
    \frac{1}{n}\sum_{i=1}^n \ua_{w,i}^{\otimes 2}
    ~=~
    \frac{1}{n}\sum_{i=1}^n
    \tilde w_{n,i}^{\,2}\,
    \ub_i(\htheta)^{\otimes 2}, \qquad
    \bar a_{3,w,n}
    ~=~
    \frac{1}{n}\sum_{i=1}^n \ua_{w,i}^{\otimes 3}
    ~=~
    \frac{1}{n}\sum_{i=1}^n
    \tilde w_{n,i}^{\,3}
    \ub_i(\htheta)^{\otimes 3},
\]
where each element is of the form
\[
(\bar a_{2,w,n})_{k,l}
=
\frac{1}{n}\sum_{i=1}^n
\tilde w_{n,i}^2
b_{i,k}(\hat\theta_n)b_{i,l}(\hat\theta_n),
\qquad
(\bar a_{3,w,n})_{k,l,r}
=
\frac{1}{n}\sum_{i=1}^n
\tilde w_{n,i}^3
b_{i,k}(\hat\theta_n)b_{i,l}(\hat\theta_n)b_{i,r}(\hat\theta_n).
\]
By Condition~\ref{c:model-11}, $\bbE_{\theta_0}\big|b_{i,k}(\theta_0)\, b_{i,l}(\theta_0)\big|^2<\infty$ and $\bbE_{\theta_0}\big|b_{i,k}(\theta_0)\, b_{i,l}(\theta_0), b_{i,k}(\theta_0)\big|^2<\infty$. Therefore, the consistency of $\htheta$ and Lemma~\ref{l:as_conv} with $q=2$ and $q=3$ respectively give
\begin{align}
\label{eq:cond_a}
    \bar a_{2,w,n}~\rightarrow~ A_{2,w}, \quad \bar a_{3,w,n}~\rightarrow~ A_{3,w},
    \quad \asP.
\end{align}
It remains to verify that $A_{2,w}$ is positive definite. For any
nonzero $\uv\in\mathbb{R}^d$,
\[
\uv^\top A_{2,w}\uv
=
c\lim_{N\to\infty}\frac{1}{N}\sum_{l=1}^N
\bbE_{\theta_0}\left[
\pi_l^{-1}
\{\uv^\top b_l(\theta_0)\}^2
\right].
\]
By Condition~\ref{c:model-12}, for every nonzero $\uv\in\mathbb{R}^d$, $P_{\theta_0}\{v^\top b(X^*,\theta_0)\neq 0\}>0.$
Since $\pi_l^{-1}\geq 1$, it follows that
$\uv^\top A_{2,w}\uv>0$ for every nonzero $\uv$, and hence
$A_{2,w}$ is positive definite. Therefore, \eqref{eq:cond_a}
verifies Condition~(a) of Lemma~\ref{l:edgeworth}.
Next, Condition~\ref{c:model-11} also ensures that $E_{\theta_0}\|b(\theta_0)\|^8<\infty$. Applying
Lemma~\ref{l:as_conv} with $q=4$ along with the consistency of $\htheta$ yields
\[
    \frac{1}{n}\sum_{i=1}^n\|\ua_{w,i}\|^4
    =
    \frac{1}{n}\sum_{i=1}^n
    \tilde w_{n,i}^{\,4}\|\ub_i(\htheta)\|^4
    ~\rightarrow~ a_{4,w},
    \qquad \asP.
\]
Consequently,
\[
    \sup_n\frac{1}{n}\sum_{i=1}^n\|a_{w,i}\|^4<\infty,
    \qquad \asP,
\]
which verifies Condition (b) of Lemma~\ref{l:edgeworth}.
Finally, Condition (c) is satisfied by observing that
\[
    \max_{1\leq i\leq n}\|a_{w,i}\|^4
    \leq
    \sum_{i=1}^n\|a_{w,i}\|^4
    =O(n),
    \qquad \asP,
\]
and therefore
\[
    \max_{1\leq i\leq n}\|a_{w,i}\|
    =O(n^{1/4})
    =\epsilon_n n^{1/2},
    \qquad
    \epsilon_n=O(n^{-1/4})\rightarrow0.
\]
Together with Conditions~\ref{c:wt-1}--\ref{c:wt-2}, these results establish a valid Edgeworth expansion for $n^{-1/2}\sum_{i=1}^n Z_i \ua_{w,i}$, and hence for the joint distribution of
$(\tilde S_{1,w,n}(\htheta),\tilde S_{2,w,n}(\htheta))$,
conditionally on $(\uX_n,\uw_n)$, $\asP$.
\end{proof}

\bigskip

We now use the joint Edgeworth expansion established in Lemma~\ref{l:joint_edgeworth} to prove Theorem~\ref{t:edgeworth}.

\bigskip

\noindent 
\begin{theorem-non}[{\bf Theorem}~\ref{t:edgeworth} (Second-order Edgeworth expansion)]
Under Conditions~\ref{c:wt-1}--\ref{c:wt-2},
\ref{c:survey-1}--\ref{c:survey-3},
\ref{c:model-1}--\ref{c:model-12}, and \ref{c:prior},
let $\{\ttheta_n\}_{n \geq 1}$ be a conditionally consistent sequence of roots of
the S-PB estimating equation $\nablatheta\tilde Q_{w,n}(\theta)=0$.
Then, the density of $I_{w,n}(\htheta)^{-1/2}
    J_{w,n}(\htheta)\,
    n^{1/2}(\ttheta_n-\htheta)$
admits the following second-order Edgeworth expansion:
\begin{equation*}
\phi(y)\left[
1
+
I_w(\theta_0)^{-1/2}
\frac{
\left\{
    \ug_0\cdot\nablatheta\log\pi(\theta_0)
\right\}^{\!\top}y
}{\sqrt{n}}
+
\frac{P_{3,w,n}(y)}{\sqrt{n}}
\right]
~+~
o\!\left(n^{-1/2}\right),
\quad \asP,
\end{equation*}
where $\htheta$ is a strongly consistent solution of the powered likelihood equation $S_{w,n}(\theta) = 0$, $I_{w,n}(\htheta)$, $J_{w,n}(\htheta)$, and $I_w(\theta_0)$
are as defined in Section~\ref{s:score-swlb},
$\phi$ denotes the $K$-dimensional standard normal density, and
$P_{3,w,n}(y)$ is a third-order polynomial independent of the prior.
\end{theorem-non}

\bigskip
\begin{proof}
Following the general argument of \citet{pompe2021}, we begin with a Taylor expansion of the S-PB estimating equation around $\hat\theta_n$. Under Conditions~\ref{c:model-11} and \ref{c:prior}, we have
\begin{align*}
    0
    &~=~ \nablatheta \tilde{Q}_{w,n}(\ttheta_n) \\
    &~=~ \nablatheta \tilde{Q}_{w,n}(\htheta)
    + \nablathetasq \tilde{Q}_{w,n}(\htheta)(\ttheta_n-\htheta) 
    + \frac{1}{2}\nablathetac \tilde{Q}_{w,n}(\htheta)
        [\ttheta_n-\htheta,\ttheta_n-\htheta] + o_P(\|\ttheta_n-\htheta\|^2),
\end{align*}
where, for $i=1,\ldots,K$, the $i$th component of the third-order term is
\[
    \left\{
    \nablathetac \tilde Q_{w,n}(\htheta)
    [\ttheta_n-\htheta,\ttheta_n-\htheta]
    \right\}_i
    =
    \sum_{j=1}^K\sum_{k=1}^K
    \frac{\partial^3\tilde Q_{w,n}(\htheta)}
    {\partial\theta_i\partial\theta_j\partial\theta_k}
    (\ttheta_{n,j}-\hat{\theta}_{n,j})
    (\ttheta_{n,k}-\hat{\theta}_{n,k}).
\]
Consequently, we have
\begin{align}
\Big\{ -\frac1n \nablathetasq \tilde{Q}_{w,n}(\hat\theta_n)\Big\}
\sqrt{n}(\ttheta_n-\htheta)
\nonumber
& ~=~
\frac{\nablatheta \tilde{Q}_{w,n}(\hat\theta_n)}{\sqrt{n}}
~+~
\frac{
\nablathetac \tilde{Q}_{w,n}(\htheta)
        [\ttheta_n-\htheta,\ttheta_n-\htheta]
}{2\sqrt{n}}\\
\label{eq:taylor}
& \qquad ~+~
o_P\bigg(
\frac{\|\ttheta_n-\htheta\|^2}{\sqrt{n}}
\bigg).
\end{align}
By Lemma~\ref{l:spb_consistency}, we have $(\ttheta_n-\htheta) = o_P(1)$. Further, since $\|J_{w,n} (\htheta)-J(\htheta)\|=o_P(1)$ and $\|I_{w,n}(\htheta)-I_w(\htheta)\|=o_P(1)$, the positive definiteness of the corresponding limiting matrices in Condition~\ref{c:model-7} implies that, with probability tending to one, $I_{w,n}(\htheta)$, $J_{w,n}(\htheta)$, and $\tilde{J}_{w,n}(\htheta)$ are positive definite, and hence non-singular, for sufficiently large $n$. Moreover, the moment conditions in Conditions~\ref{c:model-9} and \ref{c:model-11} ensure the finiteness of the higher-order quantities appearing in the expansion, including $A_{3,w}$, $B_{1,2,w}$, and $\mu_3$. 
Further, since $\sum_{i=1}^n \tilde{w}_{n,i} \,\psi_i(\htheta) = 0$, we have $\tilde{S}_{1,w,n} (\htheta) ~=~ n^{-1/2} \sum_{i=1}^n Y_i \,\tilde{w}_{n,i} \,\psi_i(\htheta),$
which thereby implies
\begin{align}
\label{eq:1-Q}
    \frac{\nablatheta \tilde{Q}_{w,n} (\htheta)}{\sqrt{n}} & ~=~ \tilde{S}_{1,w,n} (\htheta) ~+~ \frac{1}{\sqrt{n}}\,\ug_0^\top \nablatheta \log \pi(\htheta)
\end{align}
Next, note that
\begin{align}
\label{eq:2-Q}
    -\frac1n \nablathetasq \tilde{Q}_{w,n}(\hat\theta_n) & ~=~ \tilde{J}_{w,n} (\htheta) ~-~ \frac1n \,\ug_0^\top \nablathetasq \log \pi(\htheta) ~=~ \tilde{J}_{w,n} (\htheta) ~+~ o_P(n^{-1/2}),\\[2ex]
\label{eq:3-Q}
    \frac1n \nablathetac \tilde{Q}_{w,n}(\hat\theta_n) & ~=~ \tilde{\mu}_{3,w,n} (\htheta) ~+~ \frac1n \,\ug_0^\top \nablathetac \log \pi(\htheta) ~=~ \tilde{\mu}_{3,w,n} (\htheta) ~+~ o_P(n^{-1/2}),
\end{align}
since $\nablathetasq \log \pi(\htheta) = O_P(1)$ and $\nablathetac \log \pi(\htheta) = O_P(1)$ by Condition~\ref{c:prior}. 
Combining \eqref{eq:1-Q}--\eqref{eq:3-Q} with \eqref{eq:taylor} yields
\begin{align}
    \nonumber & \Big\{
        \tilde J_{w,n}(\htheta)+o_P(n^{-1/2})
    \Big\}  n^{1/2}(\ttheta_n-\htheta)\\
    \nonumber  & \quad  =~
    \tilde S_{1,w,n}(\htheta)
    +\frac{1}{\sqrt{n}}\,
        \ug_0^\top\nablatheta\log\pi(\htheta) 
    +\frac{1}{2\sqrt{n}}\,
        \tilde\mu_{3,w,n}(\htheta)[ n^{1/2}(\ttheta_n-\htheta), n^{1/2}(\ttheta_n-\htheta)]\\
    \label{eq:taylor2} & \quad \qquad 
     + o_P\bigg(
\frac{\|\ttheta_n-\htheta\|^2}{\sqrt{n}}
\bigg)\,.
\end{align}
By definition, we have $\tilde J_{w,n}(\htheta)
    =
    J_{w,n}(\htheta)
    +
    n^{-1/2}\tilde S_{2,w,n}(\htheta)$.
Thus,
\begin{align}
    \tilde J_{w,n}(\htheta)^{-1}
    &=
    J_{w,n}(\htheta)^{-1}
    -
    \frac{1}{\sqrt{n}}
    J_{w,n}(\htheta)^{-1}
    \tilde S_{2,w,n}(\htheta)
    J_{w,n}(\htheta)^{-1}
    +o_P(n^{-1/2}).
    \label{eq:spb_J_inverse}
\end{align}
Moreover, by Lemma~\ref{l:spb_consistency}, we have 
$\|\ttheta_n-\htheta\|=o_P(1)$. Since
$\tilde\mu_{3,w,n}(\htheta)=O_P(1)$, the quadratic term in
\eqref{eq:taylor2} is of order $o_P(1)\,
    \big\|\sqrt{n}(\ttheta_n-\htheta)\big\|$.
Together with the asymptotic non-singularity of
$\tilde J_{w,n}(\htheta)$ and
$\tilde S_{1,w,n}(\htheta)=O_P(1)$, it follows from
\eqref{eq:taylor2} that $\sqrt{n}(\ttheta_n-\htheta)=O_P(1)$.
Consequently, using \eqref{eq:spb_J_inverse} in \eqref{eq:taylor2}, the first-order
approximation is
\begin{equation}
    \sqrt{n}(\ttheta_n-\htheta)
    =
    J_{w,n}(\htheta)^{-1}
    \tilde S_{1,w,n}(\htheta)
    +o_P(1).
    \label{eq:first-order-approx}
\end{equation}
Then, using \eqref{eq:spb_J_inverse} and \eqref{eq:first-order-approx} in \eqref{eq:taylor2}, we have
\begin{align}
    \sqrt{n}(\ttheta_n-\htheta)
    \nonumber & ~=~ 
    J_{w,n}(\hat\theta_n)^{-1} \bigg[ \tilde S_{1,w,n}(\hat\theta_n) +
    \frac{\ug_0^\top \nabla_\theta\log \pi(\hat\theta_n) }{\sqrt n}  - \frac{\tilde S_{2,w,n}(\hat\theta_n) J_{w,n}(\hat\theta_n)^{-1} \tilde S_{1,w,n}(\hat\theta_n)}{\sqrt n} \bigg] \\
    \nonumber &
    \qquad +  J_{w,n}(\hat\theta_n)^{-1} \frac{\tilde\mu_{3,w,n}(\htheta)[ J_{w,n}(\hat\theta_n)^{-1} \tilde S_{1,w,n}(\hat\theta_n), J_{w,n}(\hat\theta_n)^{-1} \tilde S_{1,w,n}(\hat\theta_n)\Big]}{2\sqrt{n}} 
    ~+~ o_P(n^{-1/2})\\[2ex]
    \nonumber & ~=~ 
    J_{w,n}(\hat\theta_n)^{-1} \bigg[ \tilde S_{1,w,n}(\hat\theta_n) +
    \frac{\ug_0^\top \nabla_\theta\log \pi(\hat\theta_n) }{\sqrt n}  - \frac{\tilde S_{2,w,n}(\hat\theta_n) J_{w,n}(\hat\theta_n)^{-1} \tilde S_{1,w,n}(\hat\theta_n)}{\sqrt n} \bigg]   \\
    \label{eq:taylor3} &
    \qquad +~ J_{w,n}(\hat\theta_n)^{-1} \frac{\mu_{3}(\theta_0) \Big[ J_{w,n}(\hat\theta_n)^{-1} \tilde S_{1,w,n}(\hat\theta_n), J_{w,n}(\hat\theta_n)^{-1} \tilde S_{1,w,n}(\hat\theta_n)\Big]}{2\sqrt n}   ~+~ o_P(n^{-1/2}),
\end{align}
where the last equality follows Lemma~\ref{l:as_conv} with $q=1$, and the consistency of $\htheta$, which yield
$\tilde\mu_{3,w,n}(\htheta)=\mu_3(\theta_0)+o_P(1)$. 
Finally, \eqref{eq:taylor3} gives
\begin{align}
    I_{w,n}(\hat\theta_n)^{-1} J_{w,n}(\hat\theta_n)\sqrt{n}(\ttheta_n-\htheta) 
    \nonumber & ~=~ 
    I_{w,n}(\hat\theta_n)^{-1} \Bigg[\, \tilde S_{1,w,n}(\hat\theta_n) ~+~
    \frac{\ug_0^\top \nabla_\theta\log \pi(\hat\theta_n) }{\sqrt n}\\
    \nonumber& \qquad \qquad -~ \frac{\tilde S_{2,w,n}(\hat\theta_n) J_{w,n}(\hat\theta_n)^{-1} \tilde S_{1,w,n}(\hat\theta_n)}{\sqrt n}   \\
    \nonumber & \qquad \qquad +~  \frac{\mu_{3}(\theta_0) \Big[ J_{w,n}(\hat\theta_n)^{-1} \tilde S_{1,w,n}(\hat\theta_n), J_{w,n}(\hat\theta_n)^{-1} \tilde S_{1,w,n}(\hat\theta_n)\Big]}{2\sqrt n}\, \Bigg]  \\
    \label{eq:taylor4} &
    \qquad +~ o_P\big(n^{-1/2}\big).
\end{align}
Define
\[
    T_{w,n}
    :=
    I_{w,n}(\htheta)^{-1/2}
    J_{w,n}(\htheta)
    n^{1/2}(\ttheta-\htheta).
\]
Conditional on $(\uX_n,\uw_n)$, the right-hand side of
\eqref{eq:taylor4} is a polynomial in
$\big(\tilde S_{1,w,n}(\htheta),\tilde S_{2,w,n}(\htheta)\big)$,
with coefficients determined by the observed sample. By
Lemma~\ref{l:edgeworth}, this pair admits a valid joint Edgeworth
expansion, and hence so does $T_{w,n}$. Following the argument of \citet{pompe2021}, let
$\kappa_{1,w,n}$ and $\kappa_{3,w,n}$ denote, respectively, the first and third-order cumulants of $T_{w,n}$, retained up to order
$n^{-1/2}$. The characteristic function of $T_{w,n}$ then satisfies
\begin{equation}
\label{eq:cf}
    \psi_{T_{w,n}}(t)
    =
    \exp\left(-\frac{t^\top t}{2}\right)
    \left\{
        1
        +\kappa_{1,w,n}^{\top}(it)
        +\frac{1}{6}
        \kappa_{3,w,n}[it,it,it]
    \right\}
    +o_P(n^{-1/2}),
\end{equation}
conditionally on $(\uX_n,\uw_n)$, $\asP$. Define $Z_{w,n} :=\big(Z_{w,n,1},\ldots,Z_{w,n,K}\big)$, with
\[
    Z_{w,n,q}
    :=
    \frac{1}{n}\sum_{i=1}^n
    \tilde w_{n,i}^{\,2}\,
    \psi_i(\htheta)^\top
    J_{w,n}(\htheta)^{-1}
    \mu_{3,q}(\theta_0)
    J_{w,n}(\htheta)^{-1}
    \psi_i(\htheta), \qquad 1\leq q\leq K,
\]
where $\mu_{3,q}(\theta_0)\in\mathbb R^{K\times K}$ denotes a matrix with $\big\{\mu_{3,q}(\theta_0)\big\}_{r,s} = \big\{\mu_3(\theta_0)\big\}_{q,r,s}$, $1\leq r,s\leq K$. We further define $ Z_w:=(Z_{w,1},\ldots,Z_{w,K})$ with
\[
    Z_{w,q}
    :=
    c\lim_{N\to\infty}\frac{1}{N}
    \sum_{l=1}^N
    E_{\theta_0}\!\left[
        \pi_l^{-1}
        \psi_l(\theta_0)^\top
        J(\theta_0)^{-1}
        \mu_{3,q}(\theta_0)
        J(\theta_0)^{-1}
        \psi_l(\theta_0)
    \right],  \qquad 1\leq q\leq K.
\]
Using $\bbE(Y_i)=1$ and $\bbE[(Y_i-1)(Y_j-1)]=\mathds{1}(i=j)$, we have
\[
    E\!\left[
        \tilde S_{2,w,n}(\htheta)
        J_{w,n}(\htheta)^{-1}
        \tilde S_{1,w,n}(\htheta)
        \,\middle|\,\uX_n,\uw_n
    \right]
    =
    -B_{1,2,w,n}(\htheta),
\]
and, by the definition of $Z_{w,n}$,
\[
    E\!\left[
        \mu_3(\theta_0)
        \left[
            J_{w,n}(\htheta)^{-1}\tilde S_{1,w,n}(\htheta),
            J_{w,n}(\htheta)^{-1}\tilde S_{1,w,n}(\htheta)
        \right]
        \,\middle|\,\uX_n,\uw_n
    \right]
    =
    Z_{w,n}.
\]
By Lemma~\ref{l:as_conv} with $q=2$ and the consistency of $\htheta$,
\[
    Z_{w,n}\rightarrow Z_w, \quad B_{1,2,w,n}(\htheta)\rightarrow B_{1,2,w}(\theta_0),
    \qquad \asP.
\]
Therefore, from \eqref{eq:taylor4}, we have
\[
\begin{aligned}
\kappa_{1,w,n}
={}&
E\!\left(T_{w,n}\mid\uX_n,\uw_n\right) \\
={}&
\frac{I_{w,n}(\htheta)^{-1/2}}{\sqrt{n}}
\left\{
    \ug_0^\top\nablatheta\log\boldsymbol\pi(\htheta)
    +B_{1,2,w,n}(\htheta)
    +\frac{Z_{w,n}}{2}
\right\}
+o_P(n^{-1/2}).
\end{aligned}
\]
Using the consistency of $\htheta$ and the convergence of
$I_{w,n}(\htheta)$, $B_{1,2,w,n}(\htheta)$, and $Z_{w,n}$, it follows that
\[
\kappa_{1,w,n}
=
\frac{I_w(\theta_0)^{-1/2}}{\sqrt{n}}
\left\{
    \ug_0^\top\nablatheta\log\boldsymbol\pi(\theta_0)
    +B_{1,2,w}
    +\frac{Z_w}{2}
\right\}
+o_P(n^{-1/2}).
\]
Next, consider the third-order cumulant $\kappa_{3,w,n}$ of $T_{w,n}$. Conditional on
$(\uX_n,\uw_n)$, the term
\[
    \frac{1}{\sqrt{n}}
    I_{w,n}(\htheta)^{-1/2}
    \ug_0^\top\nablatheta\log\boldsymbol\pi(\htheta)
\]
is non-random and therefore cancels upon centering $T_{w,n}$. Hence $\kappa_{3,w,n}$, up to order $n^{-1/2}$, does not
depend on the prior. Applying Lemma~\ref{l:fourier-identities} to \eqref{eq:cf}, the corresponding Edgeworth density
can be written as
\[
\phi(y)\left[
1
+
I_w(\theta_0)^{-1/2}
\frac{
\left\{
    \ug_0\cdot\nablatheta\log\pi(\theta_0)
\right\}^{\!\top}y
}{\sqrt{n}}
+
\frac{P_{3,w,n}(y)}{\sqrt{n}}
\right]
~+~
o\!\left(n^{-1/2}\right),
\quad \asP,
\]
where $P_{3,w,n}(y)$ collects the remaining terms arising from the
first and third-order cumulants and is independent of the prior.
\end{proof}

\section{Verification of Conditions S2 and S4 under sampling designs}
\label{s:designs}

Conditions~\ref{c:survey-2} and~\ref{c:survey-4} are readily satisfied under simple random sampling without replacement, where the pairwise inclusion dependence is $O(N^{-1})$ and the first-order inclusion probabilities are constant. Under Poisson sampling with unit-specific inclusion probabilities, Condition~\ref{c:survey-2} follows from independent selection, while Condition~\ref{c:survey-4} holds when the inclusion probabilities depend only on independent unit-level characteristics. We now verify both conditions for two additional fixed-size sampling designs: stratified simple random sampling without replacement and rejective sampling with first-order inclusion probabilities proportional to size ($\pi$ps), under suitable regularity conditions.

\subsection{Stratified simple random sampling without replacement}
\label{s:stratified-srs}

Consider stratified simple random sampling without replacement
\citep{cochran1977, sarndal1992}, where the finite population is
partitioned into $H$ strata, with $H$ fixed. Let $A_l \in \{1,2,\ldots,H\}$ denote the stratum
membership of population unit $l$. We assume that the pairs
$(X_l^*,A_l)$ are i.i.d.\ across population units, with
$P(A_l=a)=p_a>0$, for $a=1,2,\ldots,H$.
The stratum membership $A_l$ may be associated with $X_l^*$,
allowing the sampling design to be informative.
Let $N_a=\sum_{l=1}^N \mathds{1}(A_l=a)$ denote the population size of stratum $a$, and let $\rho_a\in(0,1)$ be its prespecified sampling fraction, which may vary across strata. Within each nonempty stratum, a simple random sample of size $n_a=\max\{1,\lfloor \rho_a N_a\rfloor\}$
is selected without replacement, independently across strata, where $\lfloor x\rfloor$ denotes the greatest integer less than or equal to $x$. For empty strata, we set $n_a=0$. Consequently, the first-order inclusion probability of unit $l$ is
\(
\pi_l={n_{A_l}}/{N_{A_l}}.
\)
Since $N_a/N \to p_a$ almost surely and
$n_a/N_a\to\rho_a$, the overall sampling fraction satisfies
\[
\frac{n}{N}\longrightarrow
\sum_{a=1}^H p_a\rho_a\in(0,1), \qquad n=\sum_{a=1}^H n_a.
\]

\paragraph{Verification of Condition~\ref{c:survey-2}.}
For two distinct population units $l$ and $k$ belonging
to the same stratum $a$, the first and second-order
inclusion probabilities are given by
\[
\pi_l=\pi_k=\frac{n_a}{N_a},
\qquad
\pi_{lk}=\frac{n_a(n_a-1)}{N_a(N_a-1)}.
\]
Consequently,
\[
\left|\frac{\pi_{lk}}{\pi_l\pi_k}-1\right|
=\frac{N_a-n_a}{n_a(N_a-1)} 
=\frac{1-n_a/N_a}{n_a/N_a}
\cdot\frac{N}{N_a}
\cdot\frac{1}{N}
\cdot\frac{N_a}{N_a-1}
=O(N^{-1}),
\]
almost surely, since $N_a/N\to p_a>0$,
$n_a/N_a\to\rho_a>0$ and $N_a/(N_a-1)\to 1$ almost surely. For units belonging to
different strata, independence of sampling across
strata implies $\pi_{lk}=\pi_l\pi_k$.
Since $H$ is fixed, the bound holds uniformly
over distinct pairs of population units. Hence,
Condition~\ref{c:survey-2} is satisfied.

\paragraph{Verification of Condition~\ref{c:survey-4}.}
Let $h\in\mathcal H$ and fix $r\in\{1,2,3\}$. Throughout, $C$ denotes a
finite constant independent of $l$, $k$, $M$ and $N$, whose value may
change from line to line, and expectations are taken under the joint law
of $(X_l^*,A_l)$. Define
\[
U_{l,N}=\pi_l^{-r}h(X_l^*),
\qquad
V_l=\rho_{A_l}^{-r}h(X_l^*),
\qquad
\Delta_{l,N}=U_{l,N}-V_l .
\]
Since the pairs $(X_l^*,A_l)$ are i.i.d.\ across population units, the
variables $\{V_l\}_{l=1}^N$ are independent. Moreover, letting
$\rho_{\min}=\min_{1\leq a\leq H}\rho_a$, we have
$E_{\theta_0}(V_l^2)\leq \rho_{\min}^{-2r}E_{\theta_0}|h(X^*)|^2<\infty$.

\smallskip
\noindent\emph{Step 1: Approximation of the inclusion probabilities.}
Consider a nonempty stratum $a$. If $\rho_aN_a<1$, then $n_a=1$ and
$0<n_a/N_a-\rho_a\leq 1/N_a$. If $\rho_aN_a\geq 1$, then
$n_a=\lfloor\rho_aN_a\rfloor$, so that $\rho_aN_a-1<n_a\leq\rho_aN_a$.
In either case,
\[
\left|\frac{n_a}{N_a}-\rho_a\right|\leq\frac{1}{N_a}.
\]
Furthermore, if $\rho_aN_a<1$, then $n_a/N_a=1/N_a>\rho_a$, while if
$\rho_aN_a\geq 1$, then $n_a/N_a\geq\rho_a/2$, since
$\lfloor x\rfloor\geq x/2$ for $x\geq 1$. Hence both $n_a/N_a$ and
$\rho_a$ lie in $[\rho_{\min}/2,\,1]$. Since $x\mapsto x^{-r}$ is
Lipschitz on this interval with constant $r(\rho_{\min}/2)^{-(r+1)}$,
and since the stratum $A_l$ contains unit $l$ and is therefore nonempty,
we obtain
\[
|\Delta_{l,N}|
=\bigg|\left(\frac{n_{A_l}}{N_{A_l}}\right)^{-r}-\rho_{A_l}^{-r}~\bigg|
\,|h(X_l^*)|
~\leq~ \frac{C}{N_{A_l}}\,|h(X_l^*)| .
\]

\smallskip
\noindent\emph{Step 2: Moment bound for the stratum sizes.}
Conditional on $A_l=a$, the remaining $N-1$ units are i.i.d.\ and
independent of $(X_l^*,A_l)$, so that
\[
N_a\overset{d}{=}1+B,
\qquad
B\sim\operatorname{Binomial}(N-1,p_a),
\]
and $N_a$ is conditionally independent of $X_l^*$ given $A_l=a$. For
$B\sim\operatorname{Binomial}(m,p)$ with $q=1-p$, the identity
$\binom{m}{k}\{(k+1)(k+2)\}^{-1}=\binom{m+2}{k+2}\{(m+1)(m+2)\}^{-1}$
gives
\begin{align*}
E\left\{\frac{1}{(1+B)(2+B)}\right\}
&~=~\frac{1}{(m+1)(m+2)p^2}\sum_{j=2}^{m+2}\binom{m+2}{j}p^jq^{m+2-j}\\
&~\leq~ \frac{1}{(m+1)(m+2)p^2}.
\end{align*}
Since $2(1+B)\geq 2+B$, we have $(1+B)^{-2}\leq 2\{(1+B)(2+B)\}^{-1}$.
Taking $m=N-1$ and $p=p_a$,
\[
E_{\theta_0}\left(N_a^{-2}\mid A_l=a\right)
~\leq~ \frac{2}{N(N+1)\,p_a^2},
\qquad N\geq 1 .
\]

\smallskip
\noindent\emph{Step 3: Mean-square bound for $\Delta_{l,N}$.}
By Steps 1 and 2 and the conditional independence of $N_a$ and $X_l^*$
given $A_l=a$,
\[
\begin{aligned}
E_{\theta_0}(\Delta_{l,N}^2)
&\leq C\sum_{a=1}^H p_a\,
E_{\theta_0}\left[\frac{h(X_l^*)^2}{N_a^2}\,\middle|\,A_l=a\right]\\
&= C\sum_{a=1}^H p_a\,
E_{\theta_0}\left\{h(X_l^*)^2\mid A_l=a\right\}
E_{\theta_0}\left(N_a^{-2}\mid A_l=a\right)\\
&\leq \frac{2C}{N(N+1)\,\min_a p_a^2}\,
E_{\theta_0}|h(X^*)|^2
\;\leq\;\frac{C}{N^2},
\end{aligned}
\]
uniformly over $l$ and for all $N\geq 1$.

\smallskip
\noindent\emph{Step 4: Covariance bound.}
For distinct units $l$ and $k$, the independence of $V_l$ and $V_k$
gives
\[
\operatorname{Cov}_{\theta_0}(U_{l,N},U_{k,N})
=\operatorname{Cov}_{\theta_0}(V_l,\Delta_{k,N})
+\operatorname{Cov}_{\theta_0}(\Delta_{l,N},V_k)
+\operatorname{Cov}_{\theta_0}(\Delta_{l,N},\Delta_{k,N}).
\]
By the Cauchy--Schwarz inequality, together with
$E_{\theta_0}(V_l^2)\leq C$ and Step 3, the first two terms are bounded
in absolute value by $C/N$ and the last by $C/N^2$. Therefore,
\[
\left|
\operatorname{Cov}_{\theta_0}
\left(\pi_l^{-r}h(X_l^*),\,\pi_k^{-r}h(X_k^*)\right)
\right|
\leq\frac{C}{N},
\]
uniformly over $l\neq k$ and for all $N\geq 1$.
For $0\leq M<N$,
\[
\sum_{M<k<l\leq N}
\left|
\operatorname{Cov}_{\theta_0}
\left(\pi_l^{-r}h(X_l^*),\,\pi_k^{-r}h(X_k^*)\right)
\right|
\leq\frac{C}{N}\binom{N-M}{2}
\leq\frac{C}{2}\,(N-M),
\]
since $N-M-1<N$. Since the bound holds uniformly over $0\leq M<N$,
Condition~\ref{c:survey-4} holds.

\subsection{Fixed-size rejective \texorpdfstring{$\pi$ps}{pi-ps} sampling}
\label{s:rejective-pips}

Consider fixed-size rejective sampling \citep{hajek1964, tille2006},
also known as conditional Poisson or maximum entropy sampling
\citep{chen1994}, with first-order inclusion probabilities proportional
to auxiliary size variables. Let $s_l>0$ denote
the size variable associated with population unit $l$, available on the
sampling frame for every unit, and assume that the pairs $(X_l^*,s_l)$ are
i.i.d.\ across population units, while allowing dependence between $X_l^*$
and $s_l$ within each unit. We assume that there exist constants
$0<s_{\min}<s_{\max}<\infty$ such that $s_{\min}\leq s_l\leq s_{\max}$
a.s. The first-order inclusion probabilities are
\[
\pi_l=\frac{ns_l}{\sum_{j=1}^N s_j}
=\frac{n}{N}\,\frac{s_l}{\bar s_N},
\qquad
\bar s_N=\frac{1}{N}\sum_{j=1}^N s_j,
\]
so that $\sum_{l=1}^N\pi_l=n$. Here the sample size $n$ is fixed by design,
and we assume that the sampling fraction $c$ in Condition~\ref{c:survey-3} satisfies
$c\,s_{\max}<E(s_1)$. Since $\bar s_N\to E(s_1)$ a.s.\ and $n/N\to c$, it
follows that, almost surely, for all sufficiently large $N$,
\[
\kappa_1\leq\pi_l\leq\kappa_2,
\qquad l=1,\ldots,N,
\]
where $\kappa_1=c\,s_{\min}/\{2E(s_1)\}$ and
$\kappa_2=\{1+c\,s_{\max}/E(s_1)\}/2$ satisfy $0<\kappa_1<\kappa_2<1$.

Under rejective sampling, the sample is obtained by conditioning independent Bernoulli selections on the event that exactly $n$ units are selected. Let $q_{l,N}\in(0,1)$ denote the underlying Bernoulli selection probabilities. By Theorem~1 of \citet{chen1994}, there exist weights $w_l>0$, unique up to a common positive factor, such that the rejective design with $q_{l,N}/(1-q_{l,N})\propto w_l$ has first-order inclusion probabilities $\pi_l$. Since $\sum_{l=1}^N \lambda w_l/(1+\lambda w_l)$ is continuous and strictly increasing in $\lambda>0$, with limits $0$ and $N$, there is a unique choice of $\lambda$ such that $\sum_{l=1}^N q_{l,N}=n$; we adopt this normalization throughout.

\paragraph{Verification of Condition~\ref{c:survey-2}.}
Let
\[
D_N=\sum_{l=1}^N q_{l,N}(1-q_{l,N})
\]
denote the variance of the sample size under the underlying
independent Bernoulli sampling design. We first show that
$D_N$ grows at least linearly with $N$.
By Lemma~1 of \citet{ruetz2022}, the first-order inclusion
probabilities under rejective sampling satisfy
\[
\pi_l\leq 2q_{l,N}.
\]
Applying the same result to the complementary rejective
sample of size $N-n$, with underlying Bernoulli probabilities
$1-q_{l,N}$, gives
\[
1-\pi_l\leq 2(1-q_{l,N}).
\]
Since $\kappa_1\leq\pi_l\leq\kappa_2<1$ almost surely
for all sufficiently large $N$, it follows that
\[
q_{l,N}\geq\frac{\kappa_1}{2},
\qquad
1-q_{l,N}\geq\frac{1-\kappa_2}{2}.
\]
Consequently,
\[
\begin{aligned}
D_N =\sum_{l=1}^N q_{l,N}(1-q_{l,N}) ~\geq~ \frac{\kappa_1(1-\kappa_2)}{4}N
=c_1N,
\end{aligned}
\]
where $c_1=\kappa_1(1-\kappa_2)/4>0$.
Since the Bernoulli selection probabilities are normalized
so that $\sum_{l=1}^N q_{l,N}=n$, Theorem~1(ii) of
\citet{boistard2012} applies and yields
\[
\pi_{lk}
=\pi_l\pi_k\left[
1-\frac{(1-\pi_l)(1-\pi_k)}{D_N}
+O(D_N^{-2})
\right],\qquad l\neq k,
\]
where the remainder is uniform over distinct population
units. Since $0<(1-\pi_l)(1-\pi_k)<1$, it follows that
\[
\begin{aligned}
\sup_{l\neq k}
\left|\frac{\pi_{lk}}{\pi_l\pi_k}-1\right|
&\leq\frac{1}{D_N}+O(D_N^{-2})\\
&\leq\frac{1}{c_1N}+O(N^{-2})\\
&=O(N^{-1}).
\end{aligned}
\]
Hence, Condition~\ref{c:survey-2} is satisfied.

\paragraph{Verification of Condition~\ref{c:survey-4}.}
Let $h\in\mathcal H$ and fix $r\in\{1,2,3\}$. Throughout, $C$ denotes a
finite constant independent of $l$, $k$, $M$ and $N$, whose value may
change from line to line, and expectations are taken under the joint law
of $(X_l^*,s_l)$. Since the sample size $n$ is fixed by design, $n\geq 1$
and $n/N\to c>0$ by Condition~\ref{c:survey-3}, we have
$c'=\inf_N n/N>0$.

\smallskip
\noindent\emph{Step 1: Representation.}
For every unit $l$,
\[
\pi_l^{-r}h(X_l^*)
=\left(\frac{n}{N}\right)^{-r}(\bar s_N)^r\,B_l,
\qquad
B_l=s_l^{-r}h(X_l^*).
\]
Since the pairs $(X_l^*,s_l)$ are i.i.d.\ across population units, the
variables $\{B_l\}_{l=1}^N$ are i.i.d., and since $s_l\geq s_{\min}$,
\[
\bbE_{\theta_0}(B_l^2)\leq s_{\min}^{-2r}\,\bbE_{\theta_0}|h(X^*)|^2<\infty.
\]
Moreover, $(n/N)^{-r}\leq c'^{-r}$ for all $N$.

\smallskip
\noindent\emph{Step 2: Leave-two-out approximation.}
For distinct units $l$ and $k$, define
\[
R_{lk,N}=\frac{1}{N}\sum_{j\neq l,k}s_j,
\qquad\text{so that}\qquad
\bar s_N=R_{lk,N}+\frac{s_l+s_k}{N}.
\]
Since the pairs $(X_j^*,s_j)$ are independent across units,
$R_{lk,N}$, $B_l$ and $B_k$ are mutually independent. Both $R_{lk,N}$
and $\bar s_N$ lie in $[0,s_{\max}]$, so by the mean value theorem, for
some $\xi_{lk,N}$ between $R_{lk,N}$ and $\bar s_N$,
\[
\left|(\bar s_N)^r-(R_{lk,N})^r\right|
=r\,\xi_{lk,N}^{\,r-1}\,\frac{s_l+s_k}{N}
\leq\frac{2r\,s_{\max}^r}{N}.
\]
Hence we may write
\[
\pi_l^{-r}h(X_l^*)
=\left(\frac{n}{N}\right)^{-r}(R_{lk,N})^rB_l+\Delta_{l,k,N},
\qquad
\Delta_{l,k,N}
=\left(\frac{n}{N}\right)^{-r}
\left\{(\bar s_N)^r-(R_{lk,N})^r\right\}B_l,
\]
and an analogous representation holds for $\pi_k^{-r}h(X_k^*)$, with
the same $R_{lk,N}$ and remainder $\Delta_{k,l,N}$. By Step 1,
$|\Delta_{l,k,N}|\leq C|B_l|/N$, so that
\[
\bbE_{\theta_0}(\Delta_{l,k,N}^2)\leq\frac{C}{N^2},
\]
uniformly over $l\neq k$ and for all $N$.

\smallskip
\noindent\emph{Step 3: Covariance of the leading terms.}
By the mutual independence of $R_{lk,N}$, $B_l$ and $B_k$,
\[
\begin{aligned}
\operatorname{Cov}_{\theta_0}\left\{(R_{lk,N})^rB_l,\,(R_{lk,N})^rB_k\right\}
&=\bbE_{\theta_0}\left\{(R_{lk,N})^{2r}\right\}
\bbE_{\theta_0}(B_l)\,\bbE_{\theta_0}(B_k)
-\left[\bbE_{\theta_0}\left\{(R_{lk,N})^r\right\}\right]^2
\bbE_{\theta_0}(B_l)\,\bbE_{\theta_0}(B_k)\\
&=\bbE_{\theta_0}(B_l)\,\bbE_{\theta_0}(B_k)\,
\operatorname{Var}_{\theta_0}\left\{(R_{lk,N})^r\right\}.
\end{aligned}
\]
Since $R_{lk,N}$ is an average of $N-2$ i.i.d.\ bounded variables,
\[
\operatorname{Var}_{\theta_0}(R_{lk,N})
=\frac{N-2}{N^2}\operatorname{Var}_{\theta_0}(s_1)
\leq\frac{C}{N}.
\]
Since $x\mapsto x^r$ is Lipschitz on $[0,s_{\max}]$ with constant
$r\,s_{\max}^{r-1}$, we have
$\operatorname{Var}_{\theta_0}\{(R_{lk,N})^r\}
\leq r^2s_{\max}^{2(r-1)}\operatorname{Var}_{\theta_0}(R_{lk,N})
\leq C/N$. Together with Step 1,
\[
\left|\operatorname{Cov}_{\theta_0}
\left\{(R_{lk,N})^rB_l,\,(R_{lk,N})^rB_k\right\}\right|
\leq\frac{C}{N}.
\]

\smallskip
\noindent\emph{Step 4: Covariance bound.}
By Step 2,
\[
\begin{aligned}
\operatorname{Cov}_{\theta_0}\left\{\pi_l^{-r}h(X_l^*),\,
\pi_k^{-r}h(X_k^*)\right\}
={}&\left(\frac{n}{N}\right)^{-2r}
\operatorname{Cov}_{\theta_0}\left\{(R_{lk,N})^rB_l,\,
(R_{lk,N})^rB_k\right\}\\
&+\left(\frac{n}{N}\right)^{-r}
\operatorname{Cov}_{\theta_0}\left\{(R_{lk,N})^rB_l,\,
\Delta_{k,l,N}\right\}\\
&+\left(\frac{n}{N}\right)^{-r}
\operatorname{Cov}_{\theta_0}\left\{\Delta_{l,k,N},\,
(R_{lk,N})^rB_k\right\}
+\operatorname{Cov}_{\theta_0}\left(\Delta_{l,k,N},\,
\Delta_{k,l,N}\right).
\end{aligned}
\]
The first term is bounded in absolute value by $C/N$, by Steps 1
and 3. Since $|(R_{lk,N})^rB_l|\leq s_{\max}^r|B_l|$, the
Cauchy--Schwarz inequality together with Steps 1 and 2 bounds the
second and third terms by $C/N$ and the last term by $C/N^2$.
Therefore,
\[
\left|\operatorname{Cov}_{\theta_0}\left\{\pi_l^{-r}h(X_l^*),\,
\pi_k^{-r}h(X_k^*)\right\}\right|
\leq\frac{C}{N},
\]
uniformly over $l\neq k$ and for all $N$.
For $0\leq M<N$,
\[
\sum_{M<k<l\leq N}
\left|\operatorname{Cov}_{\theta_0}\left\{\pi_l^{-r}h(X_l^*),\,
\pi_k^{-r}h(X_k^*)\right\}\right|
\leq\frac{C}{N}\binom{N-M}{2}
\leq\frac{C}{2}\,(N-M),
\]
since $N-M-1<N$. Since the bound holds uniformly over $0\leq M<N$,
Condition~\ref{c:survey-4} holds.
\section{Auxiliary results}
This section contains auxiliary results that have been used in the proofs of Sections \ref{s:asymptotics_swlb} and \ref{s:asymptotics_spb}.

\begin{lemma}
\label{l:conv_properties}
Consider two sequences $\{Z_n\}$ and $\{U_n\}$ and two other random variables $Z$, $U$ all defined on the space $(\Omega, \mathcal{A})$. If
$
Z_n \cp Z \ \asP$ and $U_n \cp U \ \asP
$,
then
$$
Z_n U_n \cp ZU \quad \asP \quad \text{and} \quad Z_n + U_n \cp Z + U \quad \asP.
$$
\end{lemma}

\begin{theorem}[Theorem 2, \citeauthor{Petrov2014}, \citeyear{Petrov2014}]
\label{t:gen_as_conv}

Let $\{Z_n: n \geq 1\}$ be a sequence of random variables with finite variances. Define $S_n = \sum_{i=1}^nZ_i$, $n \geq 1$ with $S_0 = 0$. If 
    $\text{Var}(S_n - S_m) ~\leq~ C(n-m)^{2r-1}$
holds for all $n>m\geq0$, where $r \geq 1$ and $C$ is a constant, then as $n \to \infty$,
$$\frac{S_n - \bbE(S_n)}{n^r} ~\as~ 0\,.$$

\end{theorem}

\begin{corollary}
\label{c:max_conv}
    Let $\{Z_n: n \geq 1\}$ be a sequence of random variables with finite variances. Define $S_n = \sum_{i=1}^nZ_i$, $n \geq 1$ with $S_0 = 0$. If
    \begin{enumerate}
        \item[(i)] $\text{Var}(S_n - S_m) ~\leq~ C(n-m)$ holds for all $n>m\geq0$\,, 
        \item[(ii)] $\lim_{n \to \infty} \ n^{-1}\ \bbE(S_n) = \mu$\,, 
    \end{enumerate}  
    where $C$ and $\mu$ are constants, then as $n \to \infty$,
    $$\frac{1}{n}\,\max_{1\leq i \leq n} |Z_i| ~\as~ 0\,.$$
\end{corollary}
\begin{proof}
By Theorem \ref{t:gen_as_conv}, it follows that ${S_n}/{n} ~\as~ \mu\,.$ Using this 
and noting that
$$
\frac{S_n}{n} - \frac{S_{n-1}}{n-1} ~=~ \left( \frac{n-1}{n} - 1 \right)\frac{S_{n-1}}{n-1} + \frac{Z_n}{n}\,,
$$
we have,
\begin{equation}
\label{eq:Z_n}
    \frac{Z_n}{n} \as 0\,.
\end{equation}
Thus for a point $\omega$ where the convergence in \eqref{eq:Z_n} holds, there exists $N_1 = N_1(\omega, \epsilon) \in \bbN$ for a given $\epsilon >0$ such that
\begin{equation*}
    \frac{|Z_n(\omega)|}{n}  < \frac{\epsilon}{2}, \quad \text{for all } N >N_1.
\end{equation*}
For $n>N_1$,
\begin{align*}
    \frac{1}{n} \max_{1\leq i \leq n} \{ |Z_i(\omega)|\} 
    & ~ = ~ \frac{1}{n} \max_{1\leq i \leq N_1} \{|Z_i(\omega)|\} \ \vee \ \frac{1}{n} \max_{1\leq i \leq n} \{|Z_i(\omega)||\} \\
    & ~\leq~ \frac{1}{n} \max_{1\leq i \leq N_1} \{|Z_i(\omega)|\} \ + \ \max_{1\leq i \leq n} \left\{ \frac{|Z_i(\omega)|}{i} \right\} \\
    & ~\leq~ \frac{1}{n} \max_{1\leq i \leq N_1} \{|Z_i(\omega)|\} \ + \ \frac{\epsilon}{2}  
\end{align*}
Note that $\max_{1\leq i \leq N_1} \{|Z_i(\omega)|\}$ does not depend on $n$. Thus, there exists $N_2 = N_2(\omega, \epsilon)$ such that
$$
\frac{1}{n} \max_{1\leq i \leq N_1} \{|Z_i(\omega)|\} < \frac{\epsilon}{2}, \quad \text{for all } N > N_2\,.
$$
For $n > \max\{N_1, N_2\}$,
$$
\frac{1}{n} \max_{1\leq i \leq n} \{ |Z_i(\omega)|\}   ~<~ \frac{\epsilon}{2} + \frac{\epsilon}{2} ~=~\epsilon\,,
$$
which holds for almost every $\omega$ and hence, completes the proof. 
\end{proof}

\begin{theorem}[Theorem 13, \citeauthor{newton1991thesis}, \citeyear{newton1991thesis}]
\label{t:prob_conv}
Let $Y_1,Y_2,\ldots$ be i.i.d.\ random variables with mean $1$ and
finite variance $\sigma^2$. If $a_{n,i}$, $i=1,2,\ldots,n$, form a
triangular array of real numbers satisfying
\[
\frac{1}{n}\sum_{i=1}^n |a_{n,i}| \rightarrow c,
\quad 0<c<\infty, \quad \text{and} \quad \frac{1}{n}\max_{1\leq i\leq n}|a_{n,i}| \rightarrow 0
\]
as $n\to\infty$, then
\[
\frac{1}{n}\sum_{i=1}^n a_{n,i}Y_i
\;\xrightarrow{P}\;
\alpha,
\]
where $\alpha = \lim_{n\to\infty} \sum_{i=1}^n a_{n,i}\big/n$.
\end{theorem}

\begin{lemma}[Lemma 15, \citeauthor{newton1991thesis}, \citeyear{newton1991thesis}]
\label{l:mat_conv}
Let $\{A_n\}$ be a sequence of matrices such that for all $i, j$, the sequence of $(A_n)_{ij}$ converges to a number $A_{ij}$; the matrix of such numbers denoted by A, then as $n \to \infty$,
$\|A_n - A\| ~\to~ 0$\,.    
\end{lemma}

\begin{theorem}[Inverse function theorem, \citeauthor{rudin1964principles}, \citeyear{rudin1964principles}]
\label{t:inverse_func}
Suppose $f$ is a mapping from an open set $\Theta \subset \bbR^r$ into $\bbR^r$, $r \geq 1$, that the partial derivatives of $f$ exist and are continuous on $\Theta$, and that the matrix of derivatives $f^\prime(\theta^*)$ has an inverse $f^\prime(\theta^*)^{-1}$ at some point $\theta^* \in \Theta$.
Define 
$
\lambda = (4 \| f^\prime(\theta^*)^{-1} \| )^{-1}\,,
$
and the ball $U_\delta \subset \Theta$ of $\theta^*$ of radius $\delta > 0$ sufficiently small so that 
$$
\| f^\prime(\theta) - f^\prime(\theta^*) \| ~<~ 2\lambda \quad \text{whenever} \quad \theta \in U_\delta\,.
$$
Then for every $\theta_1$ and $\theta_2$ in $U_\delta$,
\[
|f(\theta_1) - f(\theta_2)| ~\geq~ 2 \,\lambda \, |\theta_1 - \theta_2|
\]
and the image set $f(U_\delta)$ contains the open ball with radius $\lambda \delta$ about $f(\theta^*)$.
\end{theorem}


\begin{theorem}[Theorem 15, \citeauthor{newton1991thesis}, \citeyear{newton1991thesis}]
\label{t:clt}
Let $Z_1,Z_2,\ldots$ be i.i.d.\ random variables with mean $\mu$ and 
variance $\sigma^2$. Let $\{b_{n,i}\}$ be a non-vanishing sequence of 
constants satisfying
\[
    \frac{\sum_{i=1}^n b_{n,i}^2}
    {\max_{1\leq i\leq n} b_{n,i}^2}
    ~\rightarrow~ \infty
\]
as $n\to\infty$. Put $T_n=\sum_{i=1}^n b_{n,i}Z_i$, $\mu_n=\mu\sum_{i=1}^n b_{n,i}$, and $\sigma_n^2=\sigma^2\sum_{i=1}^n b_{n,i}^2$.
Then, as $n \to \infty$,
\[
    \frac{T_n-\mu_n}{\sigma_n}
    ~\xrightarrow{d}~ N(0,1)\,.
\]
\end{theorem}

\begin{lemma}[Lemma 1, \citeauthor{pompe2021}, \citeyear{pompe2021}]
\label{l:edgeworth}
We consider $\mathbb{R}$-valued independent random variables
$Z_1,\ldots,Z_n$ distributed according to some continuous distribution
$Z$ having all moments and satisfying $E(Z)=0$ and $\operatorname{Var}(Z)=1$.
Let $(a_j)_{j=1}^n\in\mathbb{R}^d$ be a sequence of vectors in
$\mathbb{R}^d$, possibly depending on $n$ and satisfying conditions
(a)--(c) below. Assume that there exists a sequence $\epsilon_n=o(1)$
such that
\begin{enumerate}
    \item[a)] There exist $A_2\in\mathbb{R}^{d\times d}$, $A_2>0$ and
    $A_3\in\mathbb{R}^{d\times d\times d}$ such that $\|\bar a_{2,n}-A_2\|\leq\epsilon_n$ and $\|\bar a_{3,n}-A_3\|\leq\epsilon_n$ for $n\geq1$.

    \item[b)] We have $\sup_n\frac{1}{n}\sum_{j=1}^n\|a_j\|^4<\infty.$

    \item[c)] For any $n\geq1$ we have $\sup_{j\leq n}\|a_j\|\leq\epsilon_n n^{1/2}.$
\end{enumerate}
Then $S_n$, defined as
\[
    S_n
    =
    \frac{\bar a_{2,n}^{-1/2}}{n^{1/2}}
    \sum_{j=1}^n a_j Z_j,
\]
has a valid Edgeworth expansion. Moreover, the expansion of the
characteristic function $\psi_{S_n}$ of $S_n$ up to the term
$o(n^{-1/2})$ is given by
\[
    \psi_{S_n}(t)
    ~=~
    \exp\left(-\frac{t^\top t}{2}\right)
    \left(
        1 + i^3\gamma
        \frac{
            \sum_{1\leq q,r,s\leq d}
            B_{q,r,s}t_qt_rt_s
        }{6n^{1/2}}
    \right)
    +
    o(n^{-1/2}),
\]
where $\gamma=E(Z^3)$ and
$B=\lim_{n\to\infty}\bar b_{3,n}$.
\end{lemma}

\begin{lemma}[Fourier inversion identities]
\label{l:fourier-identities}
Let $Y\in\mathbb R^K$ have density $g_Y$ and characteristic function
$\psi_Y$. If $\psi_Y$ is integrable, then
\[
    g_Y(y)
    =
    \frac{1}{(2\pi)^K}
    \int_{\mathbb R^K}
    \exp(-it^\top y)\psi_Y(t)\,dt.
\]
In particular, letting $\phi(y)=\prod_{j=1}^K\phi(y_j)$ denote the
density of the $K$-dimensional standard normal distribution and $H_k$
the $k$th probabilists' Hermite polynomial,
\[
    \frac{1}{(2\pi)^K}
    \int_{\mathbb R^K}
    \exp(-it^\top y)
    \exp\left(-\frac{t^\top t}{2}\right)
    (it_q)\,dt
    =
    H_1(y_q)\phi(y),
\]
and
\[
    \frac{1}{(2\pi)^K}
    \int_{\mathbb R^K}
    \exp(-it^\top y)
    \exp\left(-\frac{t^\top t}{2}\right)
    (it_q)^3\,dt
    =
    H_3(y_q)\phi(y).
\]
Moreover, for $q\neq r$,
\[
    \frac{1}{(2\pi)^K}
    \int_{\mathbb R^K}
    \exp(-it^\top y)
    \exp\left(-\frac{t^\top t}{2}\right)
    (it_q)^2(it_r)\,dt
    =
    H_2(y_q)H_1(y_r)\phi(y),
\]
while, for pairwise distinct $q,r,s$,
\[
    \frac{1}{(2\pi)^K}
    \int_{\mathbb R^K}
    \exp(-it^\top y)
    \exp\left(-\frac{t^\top t}{2}\right)
    (it_q)(it_r)(it_s)\,dt
    =
    H_1(y_q)H_1(y_r)H_1(y_s)\phi(y).
\]
\end{lemma}

\end{document}